\documentclass[aps,prl,twocolumn,superscriptaddress,nofootinbib,longbibliography,floatfix]{revtex4-2}

\usepackage{amsmath,amssymb,mathtools,bm}
\usepackage{amsthm}
\usepackage{graphicx}
\usepackage{xcolor}
\usepackage{booktabs}
\usepackage{hyperref}
\usepackage{mathrsfs}

\graphicspath{{figures/}}

\hypersetup{
  colorlinks=true,
  linkcolor=blue,
  citecolor=blue,
  urlcolor=blue
}

\newcommand{\ii}{\mathrm{i}}
\newcommand{\dd}{\mathrm d}

\newcommand{\1}{\mathbb{I}}
\newcommand{\prlhead}[1]{\emph{#1.---}}

\newtheorem{lemma}{Lemma}
\newtheorem{proposition}{Proposition}
\newtheorem{theorem}{Theorem}

\begin{document}

\title{Exact solutions of hidden free fermion models}

\author{Mingchen Zheng}
\email{zhengmc@iphy.ac.cn}
\affiliation{Beijing National Laboratory for Condensed Matter Physics,
Institute of Physics, Chinese Academy of Sciences, Beijing 100190, China}

\author{Bal\'azs Pozsgay}
\email{pozsgay.balazs@ttk.elte.hu}
\affiliation{MTA-ELTE ``Momentum'' Integrable Quantum Dynamics Research Group,
ELTE E\"otv\"os Lor\'and University, Budapest, Hungary}

\author{Shu Chen}
\email{schen@iphy.ac.cn}
\affiliation{Beijing National Laboratory for Condensed Matter Physics,
Institute of Physics, Chinese Academy of Sciences, Beijing 100190, China}
\affiliation{School of Physical Sciences, University of Chinese Academy of Sciences,
Beijing 100049, China}
\date{\today}

\begin{abstract}  
Hidden-free-fermion (HFF) Hamiltonians exhibit free-fermion spectra
despite lacking an explicit quadratic representation, yet their algebraic
provenance and periodic-boundary behavior remain open questions.  Here we
establish that a broad class of HFF models arises from distinct Lax
operator realizations of a twisted multistate Perk--Schultz \(R\) matrix.
Its reflected degeneration generates a closed finite hierarchy of
mutually commuting transfer matrices, whose functional relations
determine the complete many-body spectrum.   As
a decisive example, we solve the inhomogeneous
free-fermions-in-disguise (FFD) chain, which exhibits HFF,
nested HFF, or interacting spectra, depending on its coupling.  
The homogeneous FFD chain with both one-wrap and periodic
boundary conditions exhibits an emergent cubic finite-size gap
while retaining the exact \(z=3/2\) bulk behavior. 
The construction extends to higher-state and face-type
realizations, providing a common algebraic framework for
coupling-dependent spectral organization and multiscale
behavior in interacting quantum systems. 
\end{abstract}

\maketitle

\prlhead{Introduction}
Exactly solvable free-fermion systems, such as the transverse-field Ising
chain \cite{Pfeuty1970}, the XY model \cite{LiebSchultzMattis1961}, and the
Kitaev chain \cite{Kitaev2001}, provide fundamental insights into quantum
phase transitions \cite{Sachdev2011}, nonequilibrium dynamics
\cite{CalabreseCardy2006,Polkovnikov2011}, and topological phases
\cite{HasanKane2010,QiZhang2011}.  Recently, increasing attention has
focused on a class of interacting systems known as hidden free fermions
(HFFs), whose many-body spectra retain a free-fermion organization despite
the absence of any apparent quadratic representation.  Fendley introduced
the canonical HFF construction in a chain with local four-Majorana
interactions \cite{Fendley2019}, and subsequent studies extended HFF
solvability to broad families of interacting lattice models
\cite{AlcarazPimenta2020,Elman2021,ChapmanElmanMann2023,
FendleyPozsgay2024, FukaiPozsgayVona2026}.  Studies of these systems have revealed anomalous criticality with dynamical exponent \(z=3/2\) \cite{Fendley2019}, 
exponentially large spectral degeneracies \cite{Fendley2019,VernierPiroli2026}, 
lattice supersymmetry \cite{Fendley2019,FendleyPozsgay2024}, 
and exactly tractable dynamics\cite{VonaMestyanPozsgay2025,FukaiPozsgay2025,
SzaszSchagrinEtAl2026Circuits,SzaszSchagrinEtAl2026Dynamics}.

Despite these advances, existing exact constructions of HFF models often
depend either on model-specific algebraic constructions or on cancellation
patterns among local interaction terms.  Neither a common algebraic origin
of HFF solvability nor a complete framework encompassing its distinct
representations and boundary conditions has yet been established.  The
gap is sharpest under periodic closure: integrability survives, but the known
HFF constructions no longer determine the spectrum, leaving the periodic spectra
inaccessible \cite{Fendley2019,FendleyPozsgay2024}.  Two long-standing challenges therefore remain:
identifying the common integrable structure underlying HFF
solvability and determining the exact spectra under periodic
boundary conditions~\cite{Fendley2019,FendleyPozsgay2024}. 

In this Letter, we establish a unified algebraic and analytic
framework for HFF solvability through distinct Lax realizations
of a twisted multistate Perk--Schultz \(R\) matrix
\cite{Baxter1982,Faddeev1996,PerkSchultz1981}.
Its reflected degeneration generates a finite hierarchy of
mutually commuting transfer matrices, whose functional relations
yield closed equations for the many-body spectrum.
Different Lax representations specify the microscopic
Hamiltonians, while the couplings determine which fusion
levels decouple and which remain coupled.
The resulting hierarchy provides a common description of
HFF, nested HFF, and interacting spectra.
We demonstrate this framework through the exact solution
of the inhomogeneous free-fermions-in-disguise (FFD) chain.
Its different couplings realize all three spectral
structures, with the open, one-wrap, and periodic chains
as representative cases.
The homogeneous FFD chain exhibits an emergent cubic finite-size gap
under both one-wrap and periodic boundary conditions. 
In both cases, these cubic gaps coexist with the exact
\(z=3/2\) bulk behavior, yielding a multiscale excitation
spectrum.
The construction extends to higher-state and face-type
realizations, providing a unified approach to the exact
solution of HFF systems.

\prlhead{Common \(R\)   matrix framework} The Yang--Baxter equation provides the algebraic foundation for quantum
integrability in one dimension.  In lattice systems, its local realization
is encoded by the RLL relation
\begin{equation}
R_{0,\bar 0}(u,v)L_{0j}(u)L_{\bar 0j}(v)
=
L_{\bar 0j}(v)L_{0j}(u)R_{0,\bar 0}(u,v).
\label{eq:RLL-framework}
\end{equation}
Here \(R_{0,\bar 0}(u,v)\in
\operatorname{End}(V_0\otimes V_{\bar 0})\) is the \(R\) matrix,
\(L_{0j}(u)\in\operatorname{End}(V_0\otimes\mathcal H_j)\) is the local Lax
operator, \(V_0\) and \(V_{\bar 0}\) are \(k\)-dimensional auxiliary spaces,
and \(\mathcal H_j\) is the physical space at site \(j\). 

For a chain of \(N\) sites, the local Lax operators generate the monodromy
matrix
\begin{equation}
T_0(u)
=
L_{0N}(u)L_{0,N-1}(u)\cdots L_{01}(u).
\label{eq:monodromy-framework}
\end{equation}
The local RLL relation implies the corresponding RTT relation for
\(T_0(u)\).  Appropriate auxiliary closures of the monodromy matrix
therefore generate commuting transfer families, whose regular
expansions produce the Hamiltonian and its conserved charges. 
In this construction, the \(R\) matrix fixes the integrable structure,
whereas the choice of Lax operator selects the physical chain.

Our central observation is that a multistate \(R\) matrix can serve as
the common integrable parent of many HFF systems.
Writing \(E_{ab}=|a\rangle\langle b|\), the \(R\) matrix takes the form
\begin{align}
R_{0,\bar 0}(u,v)
&={}
(u+v)\sum_{a=1}^{k}E_{aa}\otimes E_{aa}
\nonumber\\
&+\sum_{a<b}\Big[
(v-u)E_{aa}\otimes E_{bb}
+(u-v)E_{bb}\otimes E_{aa}
\Big]
\nonumber\\
&+2u\sum_{a<b}E_{ab}\otimes E_{ba}
+2v\sum_{a<b}E_{ba}\otimes E_{ab}.
\label{eq:HFF-R}
\end{align}
This is a twisted \(k\)-state Perk--Schultz \(R\) matrix at the
root-of-unity point \(q=\ii\), and it satisfies the Yang--Baxter
equation~\cite{PerkSchultz1981}. 
For the \(R\) matrix~\eqref{eq:HFF-R}, let
\(L_{0j}^{[1]}(u)\equiv L_{0j}(u)\) denote a local Lax operator
satisfying the RLL relation~\eqref{eq:RLL-framework}, and assume
that it acts irreducibly on \(\mathcal H_j\).  The corresponding
monodromy and transfer matrices are denoted by
\(T_0^{[1]}(u)\equiv T_0(u)\) and
\(t^{[1]}(u)=\operatorname{tr}_0T_0^{[1]}(u)\), respectively.
As we now show, the reflected degeneration of \(R\) produces a
closed finite hierarchy of exact functional relations that solves
the transfer matrix problem. 

\prlhead{Finite transfer matrix hierarchy}
Our method is based on fusion
\cite{KulishReshetikhinSklyanin1981,KirillovReshetikhin1987,
KunibaNakanishiSuzuki1994,WangYangCaoShi2015}.
At the reflected point \(v=-u\), the \(R\) matrix factorizes as
\begin{equation}
R_{0\bar0}(u,-u)
=
-4u\sum_{a<b}|s_{ab}\rangle\langle\bar s_{ab}|,
\label{eq:R-reflected}
\end{equation}
where
\(s_{ab}=(|ab\rangle+|ba\rangle)/\sqrt2\) and
\(\bar s_{ab}=(|ab\rangle-|ba\rangle)/\sqrt2\).
Together with the diagonal states \(|aa\rangle\), these vectors form
an orthonormal basis of \(V_0\otimes V_{\bar0}\).
$R_{0\bar0}(u,-u)$~\eqref{eq:R-reflected} maps each \(\bar s_{ab}\) onto
\(s_{ab}\), while annihilating all \(s_{ab}\) and \(|aa\rangle\).
The doubled auxiliary space therefore decomposes into three
orthogonal layers of dimensions
\(\binom{k}{2}\), \(k\), and \(\binom{k}{2}\).

Let \(U_{0\bar0}\) denote the unitary transformation whose columns
are ordered as
\begin{equation}
U_{0\bar0}
=
\Bigl(
(s_{ab})_{a<b};
(|aa\rangle)_{a=1}^{k};
(\bar s_{ab})_{a<b}
\Bigr).
\label{eq:three-layer-basis}
\end{equation}
In this basis, the RLL relation~\eqref{eq:RLL-framework}
implies that the reflected local product takes the universal upper
block-triangular form 
\begin{equation}
\begin{aligned}
&\left(U_{0\bar0}^{\dagger}\otimes\mathbf1_j\right)
L_{0j}^{[1]}(u)L_{\bar0j}^{[1]}(-u)
\left(U_{0\bar0}\otimes\mathbf1_j\right)
\\
&\qquad=
\begin{pmatrix}
L_j^{[2]}(u) & * & *\\
0 & \mathsf S_{1,j}(u^2)\otimes\mathbf1_j & *\\
0 & 0 & L_j^{[2]}(-u)
\end{pmatrix}.
\end{aligned}
\label{eq:local-fusion-block}
\end{equation} 
Here,
\(L_j^{[2]}(u)\in
\operatorname{End}(V^{[2]}\otimes\mathcal H_j)\)
is the first fused Lax operator, with
\(\dim V^{[2]}=\binom{k}{2}\), defined by
\begin{equation}
L_j^{[2]}(u)
=
\left(F_+^\dagger\otimes\mathbf1_j\right)
L_{0j}^{[1]}(u)L_{\bar0j}^{[1]}(-u)
\left(F_+\otimes\mathbf1_j\right),
\label{eq:first-fused-Lax}
\end{equation}
where
\(F_+=\sum_{a<b}|s_{ab}\rangle\,{}_{[2]}\langle ab|\).
The middle block is described by
\(\mathsf S_{1,j}(u^2)\).  Writing
\(L_{0j}^{[1]}(u)=\sum_{a,b}E_{ab}^{(0)}
\otimes L_{ab,j}^{[1]}(u)\),
the RLL relation~\eqref{eq:RLL-framework} gives 
\begin{equation}
\left[
L_{ab,j}^{[1]}(u)L_{ab,j}^{[1]}(-u),
L_{cd,j}^{[1]}(v)
\right]
=0
\label{eq:middle-layer-centrality}
\end{equation}
for all \(a,b,c,d,u,v\).  Since the local \(L^{[1]}\)
representation on \(\mathcal H_j\) is irreducible, Schur's lemma
then gives
\begin{equation}
L_{ab,j}^{[1]}(u)L_{ab,j}^{[1]}(-u)
=
\left[\mathsf S_{1,j}(u^2)\right]_{ab}\mathbf1_j .
\label{eq:first-scalar-matrix-elements}
\end{equation}
Thus \(\mathsf S_{1,j}(u^2)\) is a \(k\times k\) c-number matrix,
and the middle block is scalar on the physical space.  A detailed
derivation of
Eqs.~\eqref{eq:local-fusion-block}--\eqref{eq:first-scalar-matrix-elements}
is given in the Supplemental Material~\cite{SM}.

The local three-layer decomposition lifts directly to the full
monodromy matrix.  Define the first fused transfer matrix by
$
t^{[2]}(u)
=
\operatorname{tr}_{V^{[2]}}
L_N^{[2]}(u)\cdots L_1^{[2]}(u).
$
The doubled auxiliary trace of
Eq.~\eqref{eq:local-fusion-block} then gives
\begin{equation}
t^{[1]}(u)t^{[1]}(-u)
=
\Phi_{1,N}(u^2)\mathbf1+ 
t^{[2]}(u)
+
t^{[2]}(-u),
\label{eq:first-fusion-relation}
\end{equation}
where
$
\Phi_{1,N}(u^2)
=
\operatorname{tr}
\left[
\mathsf S_{1,N}(u^2)\cdots
\mathsf S_{1,1}(u^2)
\right]. 
$ 

The same reflected analysis applied to
\(L_j^{[2]}(u)L_j^{[2]}(-u)\) introduces the next fused transfer
matrix \(t^{[3]}(u)\).  Iterating this construction generates
\(t^{[1]}(u),\ldots,t^{[k]}(u)\)~\cite{SM}.  The fused RLL relations place all transfer matrices in a single
commuting family,
\begin{equation}
\left[
t^{[a]}(u),t^{[b]}(v)
\right]
=0,
\qquad
a,b=1,\ldots,k .
\label{eq:fused-transfer-commutativity}
\end{equation}
The resulting transfer matrices satisfy
\begin{equation}
\begin{aligned}
t^{[a]}(u)t^{[a]}&(-u)
={}
\Phi_{a,N}(u^2)\mathbf1
\\
&+
\sum_{b=1}^{\min(a,k-a)}
(-1)^{b-1}
\Big[
t^{[a-b]}(u)t^{[a+b]}(-u)
\\
&\hspace{28mm}
+t^{[a-b]}(-u)t^{[a+b]}(u)
\Big],
\end{aligned}
\label{eq:general-fusion-hierarchy}
\end{equation}
for \(a=1,\ldots,k-1\), with \(t^{[0]}(u)=\mathbf1\).  Here
\begin{equation}
\Phi_{a,N}(u^2)
=
\operatorname{tr}_{V^{[a]}}
\left[
\mathsf S_{a,N}(u^2)\cdots
\mathsf S_{a,1}(u^2)
\right],
\label{eq:general-scalar-source}
\end{equation}
where \(\mathsf S_{a,j}(u^2)\) is the scalar middle block at fusion
level \(a\). 
The auxiliary space of \(t^{[a]}(u)\) has dimension
\(
\dim V^{[a]}=\binom{k}{a}.
\)
At the top level \(a=k\), \(V^{[k]}\) is one dimensional, so
\(t^{[k]}(u)\) is the known quantum determinant in each symmetry
sector.  The hierarchy consequently closes on the \(k-1\) nontrivial transfer matrices
\(t^{[1]}(u),\ldots,t^{[k-1]}(u)\) through \(k-1\) functional relations.
Its 
derivation requires no explicit form of the fundamental Lax operator.
Together with the polynomial and symmetry constraints fixed by the specific Lax operator, these relations
reduce the exact spectrum to a finite algebraic problem. 

\prlhead{Hidden free fermions from fusion relations}
The fusion relations~\eqref{eq:first-fusion-relation}--\eqref{eq:general-fusion-hierarchy} naturally generate HFF
systems.  Projecting Eq.~\eqref{eq:local-fusion-block} onto the
diagonal auxiliary states \(\lvert aa\rangle\) and \(\lvert bb\rangle\) of
its middle layer, we define
\(
\tau_{ab}^{[1]}(u):=\langle a|T^{[1]}(u)|b\rangle .
\)
For this projected transfer operator, we have
\begin{equation}
\tau_{ab}^{[1]}(u)\tau_{ab}^{[1]}(-u)
=
\left[
\mathsf S_{1,N}(u^2)\cdots
\mathsf S_{1,1}(u^2)
\right]_{ab}\mathbf1 .
\label{eq:HFF-inversion}
\end{equation}
This is the scalar inversion identity underlying HFF spectra
\cite{Fendley2019,FendleyPozsgay2024}. 
Thus any irreducible Lax
realization of the \(R\) matrix~\eqref{eq:HFF-R} that admits such
scalar-channel projections generates a family of HFF models.

The same scalar reduction is also realized in special Lax representations and
coupling patterns, where the corresponding higher fused transfer contributions
vanish identically.  At higher fusion levels, decoupled channels give nested
HFF spectra, whereas coupled channels produce interacting spectra.  

\prlhead{Exact solution of the inhomogeneous FFD chain}
We illustrate these mechanisms with the FFD chain~\cite{Fendley2019}. Its fundamental Lax operator is
\begin{equation}
L_{\mathrm F,j}^{[1]}(u)
=
\begin{pmatrix}
\mathbf1_j & 0 & -h_ju\sigma_j^x\\
\sigma_j^z & 0 & 0\\
0 & \sigma_j^z & 0
\end{pmatrix}.
\label{eq:L-FFD}
\end{equation}
Here \(\sigma^a_j\), \(a=x,y,z\), denotes the Pauli matrix acting on
site \(j\), and the \(h_j\) are site-dependent constants.  This Lax
operator satisfies the RLL relation~\eqref{eq:RLL-framework} with the
\(k=3\) specialization of Eq.~\eqref{eq:HFF-R}.  At the fundamental level, $t_{\mathrm F}^{[1]}(u)=
\mathbf1-uH_{\mathrm F}+O(u^2)$, where
\begin{equation}
\begin{aligned}
H_{\mathrm F}
&=
\sum_{j=1}^{N-2}
h_{j+2}\,\sigma_j^z\sigma_{j+1}^z\sigma_{j+2}^x
\\
&\quad
+h_2\,\sigma_{N}^z\sigma_{1}^z\sigma_{2}^x
+h_1\,\sigma_{N-1}^z\sigma_{N}^z\sigma_{1}^x . 
\end{aligned}
\label{eq:FFD-general-Hamiltonian}
\end{equation} 
Fusion procedure~\eqref{eq:local-fusion-block}--\eqref{eq:general-fusion-hierarchy}
gives the second level
\begin{equation}
L_{\mathrm F,j}^{[2]}(u)
=
\begin{pmatrix}
0 & \ii h_ju\sigma_j^y & 0\\
\sigma_j^z & 0 & \ii h_ju\sigma_j^y\\
\mathbf1_j & 0 & 0
\end{pmatrix},
\label{eq:FFD-fused-Lax}
\end{equation}
and the top level
\(L_{\mathrm F,j}^{[3]}(u)=-h_ju\,\sigma_j^x\), giving
\(t_{\mathrm F}^{[3]}(u)
=\mathcal X_N(-u)^N\prod_{j=1}^{N}h_j \) with
\(\mathcal X_N=\prod_{j=1}^{N}\sigma_j^x\). 
The hierarchy therefore closes on \(t_{\mathrm F}^{[1]}(u)\) and
\(t_{\mathrm F}^{[2]}(u)\).  In a sector \(\mathcal X_N=\chi=\pm1\),
denote their eigenvalues by \(\Lambda_\chi^{[1]}(u)\) and
\(\Lambda_\chi^{[2]}(u)\).  They obey
\begin{align}
\Lambda_\chi^{[1]}(u)\Lambda_\chi^{[1]}(-u)
={}&
\Phi_{1,N}^{\mathrm F}(u^2)
+\Lambda_\chi^{[2]}(u)
+\Lambda_\chi^{[2]}(-u),
\label{eq:FFD-eigenvalue-1}\\
\Lambda_\chi^{[2]}(u)\Lambda_\chi^{[2]}(-u)
={}&
\Phi_{2,N}^{\mathrm F}(u^2)
+\Big(\prod_{j}h_j\Big)\chi u^N
\nonumber\\
&\times
\left[
\Lambda_\chi^{[1]}(u)
+(-1)^N\Lambda_\chi^{[1]}(-u)
\right],
\label{eq:FFD-eigenvalue-2}
\end{align}
where the scalar functions are
\begin{equation}
\Phi_{a,N}^{\mathrm F}(u^2)
=
\operatorname{tr}
\big[
\mathsf S_a(h_N^2u^2)\cdots
\mathsf S_a(h_1^2u^2)
\big],
\qquad a=1,2,
\label{eq:FFD-scalar-functions}
\end{equation}
with
\begin{equation}
\mathsf S_1(x)=
\begin{pmatrix}
1&0&-x\\
1&0&0\\
0&1&0
\end{pmatrix},
\qquad
\mathsf S_2(x)=
\begin{pmatrix}
0&x&0\\
1&0&x\\
1&0&0
\end{pmatrix}.
\label{eq:FFD-scalar-matrices}
\end{equation} 

\begin{figure*}[t]
\centering
\includegraphics[width=\textwidth]{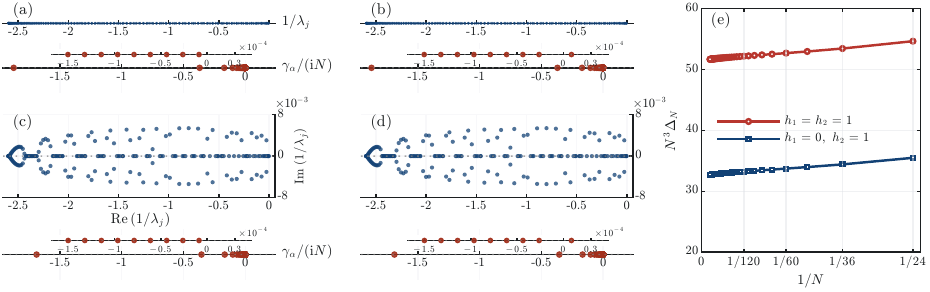}
\caption{Zero-root configurations at \(N=600\).
Panels (a) and (b) show the ground state and first excited state
of the one-wrap FFD chain,
\(h_1=0\) and \(h_2=\cdots=h_N=1\).
Panels (c) and (d) show the corresponding states of the
periodic chain, \(h_1=\cdots=h_N=1\).
Blue points show \(1/\lambda_j\), and red points show
\(\gamma_\alpha/(\ii N)\).
The enlarged windows resolve the \(\gamma\) roots closest
to the origin.
The coefficient \(\kappa\) is \(1\), \(-1\), \(2\), and \(-2\)
in panels (a)--(d), respectively.
Panel (e) compares \(N^3\Delta_N\) versus \(1/N\) for the
two closures at \(N=24,36,\ldots,600\), where
\(\Delta_N=E_1(N)-E_0(N)\).}
\label{fig:main} 
\end{figure*}

The couplings \(h_j\) determine whether the two fusion
relations scalarize.  Denote the non-scalar terms in
Eqs.~\eqref{eq:FFD-eigenvalue-1} and
\eqref{eq:FFD-eigenvalue-2} by
\begin{align}
 \mathcal C_{1}(u)
 &=
 \Lambda_\chi^{[2]}(u)+\Lambda_\chi^{[2]}(-u),\\
 \mathcal C_{2}(u)
 &=
 \Big(\prod_j h_j\Big)\chi u^N
 \left[
 \Lambda_\chi^{[1]}(u)+(-1)^N\Lambda_\chi^{[1]}(-u)
 \right].
\end{align}
When $\mathcal C_{a}(u)=0$, $a=$ 1,2, we have,
\begin{equation}
\Lambda_\chi^{[a]}(u)\Lambda_\chi^{[a]}(-u)
=
\Phi_{a,N}^{\mathrm F}(u^2),
\qquad a=1,2,
\end{equation}
and the spectrum is HFF.  For example, the open-boundary FFD chain with $h_1= 0$, $h_2 =0$ and $h_j =1$, $j>2$ realizes this case.
When \(\mathcal C_{1}(u)\neq0\) but
\(\mathcal C_{2}(u)=0\), the second level remains HFF.  Since
\(t_{\mathrm F}^{[2]}(u)\) is conserved, fixing its eigenvalue sector
scalarizes the first relation, while the resulting first-level energies
depend on that sector, giving nested HFF spectra. The system contains two distinct exact hidden-free-fermion
structures. For example, the one-wrap FFD chain
with \(h_1=0\) and $h_j =1$, $j>1$ realizes this case.  In the remaining
cases, Eqs.~\eqref{eq:FFD-eigenvalue-1} and
\eqref{eq:FFD-eigenvalue-2} remain coupled and the resulting spectrum
is interacting.

In the general case, both \(t_{\mathrm F}^{[1]}(u)\) and \(t_{\mathrm F}^{[2]}(u)\) are
polynomials in \(u\), whose explicit form follows from
Eqs.~\eqref{eq:L-FFD} and \eqref{eq:FFD-fused-Lax}~\cite{SM}.  Their
eigenvalues therefore take the form
\begin{align}
\Lambda_\chi^{[1]}(u)
&=
\prod_{j=1}^{\left\lfloor N/3\right\rfloor}
\left(1-\frac{u}{\lambda_j}\right),
\label{eq:FFD-fundamental-roots}\\
\Lambda_\chi^{[2]}(u)
&=
u^{\left\lfloor (N+1)/2\right\rfloor}
\kappa  \prod_{\alpha=1}^{
\left\lfloor 2N/3\right\rfloor
-\left\lfloor (N+1)/2\right\rfloor}
\left(1-\frac{u}{\gamma_\alpha}\right),
\label{eq:FFD-fused-roots}
\end{align}
Here \(\lambda_j,\gamma_\alpha\) are state-dependent zero roots and \(\kappa\) is a state-dependent prefactor. Substituting
Eqs.~\eqref{eq:FFD-fundamental-roots} and \eqref{eq:FFD-fused-roots}
into Eqs.~\eqref{eq:FFD-eigenvalue-1} and
\eqref{eq:FFD-eigenvalue-2} and matching powers of \(u\) gives a closed
finite system of zero root equations for \(\{\lambda_j\}\),
\(\{\gamma_\alpha\}\) and \(\kappa\).  Their solutions determine
\(\Lambda_\chi^{[1]}(u)\), from which the energy follows,
\begin{equation}
E = -\Lambda_\chi^{[1]\prime}(0)
=
\sum_{j=1}^{\left\lfloor N/3\right\rfloor}
\frac{1}{\lambda_j}.
\label{eq:FFD-energy}
\end{equation}
We have verified that the zero root solutions reproduce the complete exact diagonalization spectra at small \(N\)~\cite{SM}.

\prlhead{Thermodynamic properties of the nested structure}
We investigate the physical consequences of the nested structure
of the FFD chain. 
Setting \(h_1=0\) and \(h_j=1\) for \(j\geq2\) gives an open chain
with one wrap-around bond.
Fig.~\ref{fig:main}(a) and (b) show its ground
state and first excited state zero root patterns.
The \(\lambda_j\) are negative and real, whereas the
\(\gamma_\alpha\) are purely imaginary.
The first excited state is obtained by reversing the fused root closest to the
origin, \(\ii\gamma_{1,N}\) with \(\gamma_{1,N}>0\), together with the overall
coefficient, \((\gamma_{1,N},\kappa)\mapsto(-\gamma_{1,N},-\kappa)\), while the
\(\lambda\) roots rearrange collectively and the remaining \(\gamma\) roots
stay fixed.  The Eqs.~\eqref{eq:FFD-eigenvalue-1} and \eqref{eq:FFD-eigenvalue-2} give
\(\gamma_{1,N}\sim N^{-1}\), while the first-level energy response carries
an additional \(N^{-2}\) suppression~\cite{SM}.  In the thermodynamic limit, 
\[
\begin{aligned}
\Delta_N
&=
-\frac{9\pi^2}{2N^3}
\,\operatorname{Li}_2\!\left(-\frac{\sqrt3}{2}\right)
+o(N^{-3})\\
&=
\frac{32.32450299\ldots}{N^3}
+o(N^{-3}),
\end{aligned}
\]
for even \(N\), with twice the leading amplitude for odd \(N\). Here \(\operatorname{Li}_2(z)=\sum_{n\geq1}z^n/n^2\) is the
dilogarithm. 
Successive \(\gamma\)-root reflections generate a low-energy
tower with the same \(N^{-3}\) scaling at fixed excitation
index.

In the periodic chain, \(h_1=h_2=1\), Figs.~\ref{fig:main}(c) and (d)
show the ground state and first excited state zero root patterns, with the latter again
obtained by the excitation
\((\gamma_{1,N},\kappa)\mapsto(-\gamma_{1,N},-\kappa)\).  The \(\gamma\) and
\(\lambda\) roots now rearrange self-consistently, with the \(\lambda\) roots
forming complex-conjugate configurations.  The non-scalar term of Eq.\eqref{eq:FFD-eigenvalue-2}, however, is
weak at the scale of \(\gamma_{1,N}\): the factor \(u^N\) is negligible there,
while the induced rearrangement of the \(\lambda\) roots changes only the
response amplitude and does not alter its finite-size exponent.  The cubic
gap structure therefore persists in the periodic chain. 
Fig.~\ref{fig:main}(e) compares the rescaled gaps of the one-wrap and periodic chains
along the same sequence of chain lengths.  \(N^3\Delta_N\) approaches
\(32.3245\ldots\) for the one-wrap chain and approximately \(51.2\) for the
periodic chain, confirming the common cubic scaling.

Both closures also retain the \(z=3/2\) free-fermion
scale of the open FFD chain~\cite{Fendley2019}.
In the one-wrap chain, it is realized by the exact,
sector-resolved \(\lambda\) excitations.
For even \(N\), the periodic spectral equations likewise
admit an explicit free-fermion family built from paired
\(\lambda\) roots, whose lowest excitation energies scale
as \(N^{-3/2}\)~\cite{SM}.
Thus the smaller \(\gamma\)-induced splittings coexist
with the \(z=3/2\) free branch, rather than replacing it.

The boundary terms generate an anomalous thermodynamic response that is
distinct from an ordinary \(O(1)\) boundary correction.  In the one-wrap
chain, a single boundary bond partially lifts the exponentially large
open-chain degeneracy.  Although the lowest splittings scale as \(N^{-3}\),
occupying \(O(N)\) fused modes raises the energy by only \(O(N^{-1})\).
Their competition with the extensive entropy produces a collective
heat-capacity peak at
\(
T_{\rm p}\sim N^{-2},\qquad C_N(T_{\rm p})\sim N ,
\)
while the nonzero \(\lambda\) excitations remain frozen in this window
\cite{SM}.  Thus a local boundary modification produces an extensive
thermodynamic response.  With both wrap-around terms present, finite-size
continuation and entropy counting support a periodic collective band on the
same \(T\sim N^{-2}\) scale.  The periodic coupling dresses its energy
functional and modifies the peak amplitude.  At fixed positive temperature,
both closures retain the bulk behavior
\(c_{\mathrm{bulk}}(T)\propto T^{2/3}\), while the collective response is
confined to a window that shrinks towards zero as \(N\) grows.

\prlhead{Scope of the framework}
The Lax operator~\eqref{eq:L-FFD} is one realization of the \(R\)
matrix~\eqref{eq:HFF-R}.  Other irreducible Lax
representations, including higher-\(k\) realizations and inequivalent local
representations, generate integrable transfer families governed by the analogous
fusion hierarchy.  The chosen representation and its couplings
determine which fusion levels decouple and which remain coupled, thereby
producing HFF, nested HFF, or interacting 
spectra.  A single commuting transfer family may therefore contain several
distinct Hamiltonians with different spectral organizations.  A range of
such realizations is worked out in~\cite{SM}. 

The same \(R\) matrix also admits face-type realizations, in which the
Lax operator carries a face label in addition to the
auxiliary color index and satisfies a dynamical Yang--Baxter equation~\cite{Felder1994,EtingofVarchenko1998}.
A representative example is the periodic Fendley--Pozsgay
model~\cite{FendleyPozsgay2024}, 
\begin{equation}
H_{\rm FP}
=
\sum_j
\left(
\sigma_j^x\sigma_{j+1}^x\sigma_{j+2}^z
+b\,\sigma_{j-1}^z\sigma_{j+1}^z
+b^2\,\sigma_{j-1}^z\sigma_j^y\sigma_{j+1}^y
\right),
\label{eq:FP-Hamiltonian}
\end{equation} 
with \(b\neq0\) and all site indices understood periodically.  It
generates a closed two-level hierarchy, the face-type deformation of
 Eq.~\eqref{eq:general-fusion-hierarchy}.  Its exact solution is derived in~\cite{SM}.

\prlhead{Discussion and outlook}
We have identified the multistate Perk--Schultz \(R\) matrix as the
common integrable parent of a broad class of hidden-free-fermion
models.  The natural object of classification is therefore a Lax
operator realization together with its auxiliary closure, rather
than an isolated HFF Hamiltonian.  Under periodic closure the scalar
inversion relation of open chains is not lost but completed into a
finite fusion hierarchy fixed by the \(R\) matrix.

The FFD chain shows that closing a hidden-free-fermion chain is not a
small correction to its open spectrum.  The boundary terms reorganize
the spectrum into nested and interacting structures, in which an
emergent exact cubic gap coexists with the free \(z=3/2\) bulk behavior. 
Our result provides an analytic example of
\(N^{-3}\) finite-size scaling, complementing studies of cubic
low-energy dispersions in interacting Majorana and
lattice-supersymmetric chains 
~\cite{OBrienFendley2018,Rahmani2015,Chepiga2023}.  
Preliminary results indicate that periodic multispin FFD chains with
\(k>3\) retain an \(N^{-3}\) finite-size gap despite their
\(k\)-dependent bulk dispersion $k/2$~\cite{AlcarazPimenta2020,HigherKUnpublished}.
Boundary-dependent dynamics and further face-type realizations are left for future work.

\prlhead{Acknowledgments}
We are especially grateful to Yupeng Wang for illuminating discussions,
particularly on the fusion construction and the analysis of the polynomial
eigenvalue identities.  We also thank Wen-li Yang for helpful suggestions.
We thank Kohei Fukai and Istv\'an Vona for useful discussions.
M.Z. acknowledges financial support from the China Postdoctoral
Science Foundation (Grant No.~2025M783457) and the Postdoctoral Fellowship
Program of CPSF (Grant No.~GZC20252261).

\clearpage
\onecolumngrid
\setcounter{page}{1}
\setcounter{section}{0}
\setcounter{subsection}{0}
\setcounter{equation}{0}
\setcounter{figure}{0}
\setcounter{table}{0}
\setcounter{secnumdepth}{2}
\renewcommand{\theequation}{S\arabic{equation}}
\renewcommand{\thefigure}{S\arabic{figure}}
\renewcommand{\thetable}{S\arabic{table}}
\renewcommand{\thesection}{\arabic{section}}
\renewcommand{\thesubsection}{\arabic{section}.\arabic{subsection}}
\renewcommand{\theHequation}{S.\arabic{equation}}
\renewcommand{\theHfigure}{S.\arabic{figure}}
\renewcommand{\theHtable}{S.\arabic{table}}
\renewcommand{\theHsection}{SM.\arabic{section}}
\renewcommand{\theHsubsection}{SM.\arabic{section}.\arabic{subsection}}
\hypersetup{pageanchor=false}

\begin{center}
{\large\bf Supplemental Material for
\emph{Exact solutions of hidden free fermion models}}
\end{center}

\medskip
\begin{center}
\begin{minipage}{0.92\textwidth}
\textbf{Contents}\par\medskip
\begingroup
\small
\setlength{\parskip}{1pt}
\noindent
1.\ Universal Reflected Fusion for Irreducible \(L\)-Operators
\dotfill \pageref{sec:supp-kall-fusion}\par
\hspace*{1.5em}
1.1.\ First Reflected Fusion and Scalar Closure
\dotfill \pageref{subsec:supp-first-fusion}\par
\hspace*{1.5em}
1.2.\ Finite Reflected Hierarchy
\dotfill \pageref{subsec:supp-finite-hierarchy}\par
\noindent
2.\ Closure of the Inhomogeneous FFD Chain
\dotfill \pageref{sec:supp-twisted-closure}\par
\hspace*{1.5em}
2.1.\ Commuting Family and Hamiltonian
\dotfill \pageref{subsec:supp-ffd-commuting}\par
\hspace*{1.5em}
2.2.\ The Two Fusion Relations
\dotfill \pageref{subsec:supp-ffd-fusion-relations}\par
\hspace*{1.5em}
2.3.\ Eigenvalue Form
\dotfill \pageref{subsec:supp-ffd-eigenvalue-form}\par
\hspace*{1.5em}
2.4.\ Zero-Root Parametrization
\dotfill \pageref{subsec:supp-ffd-zero-roots}\par
\noindent
3.\ Cubic Gap of the FFD Chain
\dotfill \pageref{sec:supp-cubic-gap}\par
\hspace*{1.5em}
3.1.\ Exact Cubic Gap of the One-Wrap Closure
\dotfill \pageref{subsec:supp-gap-onewrap}\par
\hspace*{1.5em}
3.2.\ Arbitrary Chain Length
\dotfill \pageref{subsec:supp-gap-alllength}\par
\hspace*{1.5em}
3.3.\ Periodic Chain: the Response with the Inter-Level Coupling Frozen
\dotfill \pageref{subsec:supp-gap-periodic-scalar}\par
\hspace*{1.5em}
3.4.\ Dressing by the Inter-Level Coupling
\dotfill \pageref{subsec:supp-gap-dressing}\par
\hspace*{1.5em}
3.5.\ The \(z=3/2\) Branch in the Two Closures
\dotfill \pageref{subsec:supp-gap-freebranch}\par
\hspace*{1.5em}
3.6.\ Low-Temperature Thermodynamics of the One-Wrap Chain
\dotfill \pageref{subsec:supp-thermo-onewrap}\par
\hspace*{1.5em}
3.7.\ Arbitrary Chain Length
\dotfill \pageref{subsec:supp-thermo-alllength}\par
\hspace*{1.5em}
3.8.\ The Periodic Chain
\dotfill \pageref{subsec:supp-thermo-periodic}\par
\noindent
4.\ Further Examples
\dotfill \pageref{sec:supp-examples}\par
\hspace*{1.5em}
4.1.\ The Arbitrary-\(k\) Non-Hermitian Perk--Schultz Chain
\dotfill \pageref{subsec:supp-PSk}\par
\hspace*{1.5em}
4.2.\ Periodic FFD Logarithmic Charge with Claws
\dotfill \pageref{subsec:supp-claw}\par
\hspace*{1.5em}
4.3.\ An Irreducible \(D=4\) Range-Five Model
\dotfill \pageref{subsec:supp-D4}\par
\hspace*{1.5em}
4.4.\ The Fendley--Pozsgay Face Model
\dotfill \pageref{subsec:supp-FP}
\endgroup
\end{minipage}
\end{center}
\medskip

\section{Universal Reflected Fusion for Irreducible \texorpdfstring{$L$}{L}-Operators}
\label{sec:supp-kall-fusion}

This section derives the fusion structure without choosing a particular
physical realization of the Yang--Baxter algebra.  The argument has four
steps.  First, the reflected Perk--Schultz matrix fixes three explicit
auxiliary layers.  Second, the RLL relation makes the reflected local Lax
product upper block triangular in these layers and identifies the two outer
blocks as $L^{[2]}(u)$ and $L^{[2]}(-u)$.  Third, an additional auxiliary line shows
that every operator-valued matrix element of the middle block commutes with
the local RLL representation.  Schur's lemma then makes this layer scalar on an
irreducible physical space.  Finally, the same construction is iterated to
obtain the finite reflected hierarchy.

\subsection{First reflected fusion and scalar closure}
\label{subsec:supp-first-fusion}
Let \(V=\operatorname{span}\{|1\rangle,\ldots,|k\rangle\}\),
\(E_{ab}=|a\rangle\langle b|\), and let \(u,v,w\) denote spectral
parameters.  The \(R\) matrix acting on \(V_0\otimes V_{\bar 0}\) is
\begin{equation}
\begin{aligned}
R_{0\bar 0}(u,v)
={}&
(u+v)\sum_{a=1}^{k}E_{aa}\otimes E_{aa}
+(v-u)\sum_{a<b}E_{aa}\otimes E_{bb}
\\
&+
(u-v)\sum_{a>b}E_{aa}\otimes E_{bb}
+2u\sum_{a<b}E_{ab}\otimes E_{ba}
+2v\sum_{a<b}E_{ba}\otimes E_{ab}.
\end{aligned}
\label{eq:supp-kall-R}
\end{equation}
For arbitrary \(u,v,w\), it satisfies the Yang--Baxter equation
\begin{equation}
R_{12}(u,v)R_{13}(u,w)R_{23}(v,w)
=
R_{23}(v,w)R_{13}(u,w)R_{12}(u,v)
\label{eq:supp-YBE}
\end{equation}
on \(V_1\otimes V_2\otimes V_3\). For the \(R\) matrix above, let
\(L_{0j}^{[1]}(u)\) denote a local Lax operator satisfying
\begin{equation}
R_{0\bar 0}(u,v)
L_{0j}^{[1]}(u)L_{\bar 0j}^{[1]}(v)
=
L_{\bar 0j}^{[1]}(v)L_{0j}^{[1]}(u)
R_{0\bar 0}(u,v),
\label{eq:supp-framework-RLL}
\end{equation}
and assume that the local \(L^{[1]}\)-operator is irreducible on the
finite-dimensional space \(\mathcal H_j\).  Here irreducible means that no
nonzero proper subspace of \(\mathcal H_j\) is preserved by every matrix
element \(L_{ab,j}^{[1]}(v)\), for all \(a,b\) and all regular \(v\).  We also
assume that these matrix elements depend continuously on the spectral
parameter and remain finite at the values used below.

The fusion technique uses spectral points where the \(R\) matrix loses rank.
At such points, products of Lax operators acquire a block-triangular structure
in the doubled auxiliary space, generating fused transfer matrices and their
functional relations.
For the present Perk--Schultz \(R\) matrix, this occurs at the reflected point
\(v=-u\), where Eq.~\eqref{eq:supp-kall-R} becomes
\begin{equation}
 R_{0 \bar{0}}(u,-u)
 =
 2u\sum_{a<b}
 \left(
 -E_{aa}\otimes E_{bb}
 +E_{bb}\otimes E_{aa}
 +E_{ab}\otimes E_{ba}
 -E_{ba}\otimes E_{ab}
 \right).
 \label{eq:supp-kall-R-reflected}
\end{equation}
In each two-color sector
\(\operatorname{span}\{|ab\rangle,|ba\rangle\}\), \(a<b\), its matrix is
\begin{equation}
 R_{0\bar{0}}(u,-u)\big|_{\{|ab\rangle,|ba\rangle\}}
 =
 2u
 \begin{pmatrix}
 -1&1\\
 -1&1
 \end{pmatrix}.
 \label{eq:supp-kall-R-pair-block}
\end{equation}
Introducing
\(s_{ab}=(|ab\rangle+|ba\rangle)/\sqrt2\) and
\(\bar s_{ab}=(|ab\rangle-|ba\rangle)/\sqrt2\), this result can be written as
\begin{equation}
 R_{0\bar{0}}(u,-u)
 =
 -4u\sum_{a<b}|s_{ab}\rangle\langle\bar s_{ab}|.
 \label{eq:supp-kall-R-factorization}
\end{equation}
Thus, for \(u\neq0\), \(R_{0\bar{0}}(u,-u)\) has rank
\(k(k-1)/2\): it maps each \(\bar s_{ab}\) into \(s_{ab}\), while
annihilating all \(s_{ab}\) and \(|aa\rangle\).

The local RLL relation turns this rank degeneration into a block-triangular
structure in the doubled auxiliary space.  Setting \(v=-u\) in
Eq.~\eqref{eq:supp-framework-RLL} gives 
\begin{align}
 R_{0\bar{0}}(u,-u)
 L^{[1]}_{0j}(u)L^{[1]}_{\bar{0}j}(-u)
 &=
 L^{[1]}_{\bar{0}j}(-u)L^{[1]}_{0j}(u)
 R_{0\bar{0}}(u,-u),
 \label{eq:supp-kall-degenerate-RLL-left}
 \\
 L^{[1]}_{0j}(u)L^{[1]}_{\bar{0}j}(-u)
 R_{0\bar{0}}(u,-u)
 &=
 R_{0\bar{0}}(u,-u)
 L^{[1]}_{\bar{0}j}(-u)L^{[1]}_{0j}(u).
 \label{eq:supp-kall-degenerate-RLL-right}
\end{align} 
Here we use
\(R_{\bar{0}0}(-u,u)=R_{0\bar{0}}(u,-u)\). 

The reflected local product
\(
L^{[1]}_{0j}(u)L^{[1]}_{\bar{0}j}(-u)
\)
acts on
\(
(V_0\otimes V_{\bar 0})\otimes\mathcal H_j
\).
For \(u\neq0\), Eqs.~\eqref{eq:supp-kall-degenerate-RLL-left}
and \eqref{eq:supp-kall-R-factorization} give
\begin{equation}
\begin{aligned}
R_{0\bar 0}(u,-u)
L_{0j}^{[1]}(u)L_{\bar 0j}^{[1]}(-u)s_{ab}&=0,\\
R_{0\bar 0}(u,-u)
L_{0j}^{[1]}(u)L_{\bar 0j}^{[1]}(-u)|aa\rangle&=0.
\end{aligned}
\end{equation}
Hence both transformed states lie in the combined symmetric and
diagonal channels.  Moreover,
Eq.~\eqref{eq:supp-kall-degenerate-RLL-right} gives
\begin{equation}
\begin{aligned}
-4u\,
L_{0j}^{[1]}(u)L_{\bar 0j}^{[1]}(-u)s_{ab}
={}&
R_{0\bar 0}(u,-u)
L_{\bar 0j}^{[1]}(-u)L_{0j}^{[1]}(u)\bar s_{ab},
\end{aligned}
\end{equation}
so the symmetric channel is preserved separately. 
The allowed channel transitions are therefore
\begin{equation}
\begin{aligned}
 (s_{ab})_{a<b}
 &\longrightarrow (s_{ab})_{a<b},\\
 (|aa\rangle)_{a=1}^{k}
 &\longrightarrow
 (s_{ab})_{a<b}\oplus(|aa\rangle)_{a=1}^{k},\\
 (\bar s_{ab})_{a<b}
 &\longrightarrow
 (s_{ab})_{a<b}\oplus(|aa\rangle)_{a=1}^{k}
 \oplus(\bar s_{ab})_{a<b},
\end{aligned}
\label{eq:supp-kall-channel-rules}
\end{equation}
where each parenthesis denotes the span of the indicated states.
Consequently, the reflected local product is upper block triangular in this
ordered doubled-auxiliary basis.

Let \(V^{[2]}\) be the \(\binom{k}{2}\)-dimensional auxiliary space
with basis \(|ab\rangle_{[2]}\), \(a<b\).  The upper diagonal block
defines the first fused Lax operator through
\begin{equation}
{}_{[2]}\langle ab|L_j^{[2]}(u)|cd\rangle_{[2]}
=
\langle s_{ab}|
L_{0j}^{[1]}(u)L_{\bar 0j}^{[1]}(-u)
|s_{cd}\rangle .
\label{eq:supp-first-L2-definition}
\end{equation}
All matrix elements in this equation are operators on
\(\mathcal H_j\).

For \(u\neq0\), the lower diagonal block is fixed by the same fused
operator at \(-u\).  Indeed, taking matrix elements of
Eq.~\eqref{eq:supp-kall-degenerate-RLL-left} between
\(\langle s_{ab}|\) and \(|\bar s_{cd}\rangle\), and using
Eq.~\eqref{eq:supp-kall-R-factorization}, gives
\begin{align}
\langle\bar s_{ab}|
L_{0j}^{[1]}(u)L_{\bar 0j}^{[1]}(-u)
|\bar s_{cd}\rangle
&=
\langle s_{ab}|
L_{\bar 0j}^{[1]}(-u)L_{0j}^{[1]}(u)
|s_{cd}\rangle
\nonumber\\
&=
\langle s_{ab}|
L_{0j}^{[1]}(-u)L_{\bar 0j}^{[1]}(u)
|s_{cd}\rangle
\nonumber\\
&=
{}_{[2]}\langle ab|L_j^{[2]}(-u)|cd\rangle_{[2]} .
\label{eq:supp-first-lower-L2}
\end{align}
In the second equality we exchanged the two auxiliary spaces and
used the symmetry of \(s_{ab}\).
The equality extends to \(u=0\) whenever the Lax operator is regular
there.

It remains to determine the middle diagonal block.  Expanding
\begin{equation}
L_{0j}^{[1]}(u)
=
\sum_{a,b=1}^{k}
E_{ab}^{(0)}\otimes L_{ab,j}^{[1]}(u),
\end{equation}
its matrix elements are
\begin{equation}
\begin{aligned}
\langle aa|
L_{0j}^{[1]}(u)L_{\bar 0j}^{[1]}(-u)
|bb\rangle
=
L_{ab,j}^{[1]}(u)L_{ab,j}^{[1]}(-u).
\end{aligned}
\label{eq:supp-middle-elements}
\end{equation}
We therefore denote the middle block by
\begin{equation}
\mathcal C_j^{[1]}(u)
=
\sum_{a,b=1}^{k}
E_{ab}\otimes
L_{ab,j}^{[1]}(u)L_{ab,j}^{[1]}(-u).
\label{eq:supp-middle-block}
\end{equation}
Let \(U_{0\bar 0}\) denote the unitary change-of-basis matrix whose
columns are ordered as
\begin{equation}
U_{0\bar 0}
=
\Bigl(
(s_{ab})_{a<b};
(|aa\rangle)_{a=1}^{k};
(\bar s_{ab})_{a<b}
\Bigr).
\label{eq:supp-kall-U}
\end{equation} 
At this stage, the triangular structure derived above gives
\begin{equation}
\begin{aligned}
\left(U_{0\bar 0}^{\dagger}\otimes\mathbf1_j\right)
L_{0j}^{[1]}(u)L_{\bar 0j}^{[1]}(-u)
\left(U_{0\bar 0}\otimes\mathbf1_j\right)
=
\begin{pmatrix}
L_j^{[2]}(u)&*&*\\
0&\mathcal C_j^{[1]}(u)&*\\
0&0&L_j^{[2]}(-u)
\end{pmatrix}.
\end{aligned}
\label{eq:supp-first-ULLU}
\end{equation}

We now show that the middle block is scalar on the physical space.
Introducing a third auxiliary line, two applications of the RLL
relation give
\begin{align}
R_{03}(u,v)R_{\bar 0 3}(-u,v)
L_{0j}^{[1]}(u)L_{\bar 0j}^{[1]}(-u)L_{3j}^{[1]}(v)=
L_{3j}^{[1]}(v)
L_{0j}^{[1]}(u)L_{\bar 0j}^{[1]}(-u)
R_{03}(u,v)R_{\bar 0 3}(-u,v).
\label{eq:supp-three-line-RLL}
\end{align}
At the reflected point, the Yang--Baxter equation gives
\begin{equation}
\begin{aligned}
R_{0\bar 0}(u,-u)
R_{03}(u,v)R_{\bar 0 3}(-u,v) =
R_{\bar 0 3}(-u,v)R_{03}(u,v)
R_{0\bar 0}(u,-u).
\end{aligned}
\label{eq:supp-reflected-YBE-train}
\end{equation}
Together with its auxiliary-exchanged form, this relation gives the
same upper-triangular channel structure as the reflected RLL
relations above.   Its action on a middle-channel state is
\begin{equation}
\begin{aligned}
R_{03}(u,v)R_{\bar 0 3}(-u,v)|aac\rangle
=
(v^2-u^2)|aac\rangle
+
\begin{cases}
2v(u+v)(|caa\rangle+|aca\rangle),&a<c,\\
-2u(u+v)(|caa\rangle+|aca\rangle),&a>c,\\
0,&a=c.
\end{cases}
\end{aligned}
\label{eq:supp-middle-train-action}
\end{equation}
The additional terms lie in the upper symmetric channel.  Hence its
middle diagonal block is simply
\begin{equation}
(v^2-u^2)\mathbf1.
\label{eq:supp-middle-train-scalar}
\end{equation}
For \(v^2\neq u^2\), comparing the middle diagonal blocks of
Eq.~\eqref{eq:supp-three-line-RLL} and canceling the common factor
\(v^2-u^2\) gives
\begin{equation}
\mathcal C_j^{[1]}(u)L_{3j}^{[1]}(v)
=
L_{3j}^{[1]}(v)\mathcal C_j^{[1]}(u).
\label{eq:supp-middle-commutes}
\end{equation} 
By the continuity assumed above, the relation extends to the remaining regular
values of \(u\) and \(v\).  Using Eq.~\eqref{eq:supp-middle-block}, this is equivalent to
\begin{equation}
\left[
L_{ab,j}^{[1]}(u)L_{ab,j}^{[1]}(-u),
L_{cd,j}^{[1]}(v)
\right]
=0
\label{eq:supp-middle-entry-central}
\end{equation}
for all \(a,b,c,d,u,v\).  By irreducibility in the sense stated above,
Schur's lemma implies
\begin{equation}
L_{ab,j}^{[1]}(u)L_{ab,j}^{[1]}(-u)
=
\sigma_{ab,j}(u)\mathbf1_j .
\end{equation}
Taking the physical trace of this equation at \(u\) and \(-u\), and
using cyclicity of the trace, gives
\(\sigma_{ab,j}(u)=\sigma_{ab,j}(-u)\).  We may therefore define the
\(k\times k\) matrix
\begin{equation}
\left[\mathsf S_{1,j}(u^2)\right]_{ab}
=
\sigma_{ab,j}(u),
\end{equation}
so that
\begin{equation}
\mathcal C_j^{[1]}(u)
=
\mathsf S_{1,j}(u^2)\otimes\mathbf1_j .
\label{eq:supp-middle-Schur}
\end{equation}

Substituting Eq.~\eqref{eq:supp-middle-Schur} into
Eq.~\eqref{eq:supp-first-ULLU} finally gives
\begin{equation}
\begin{aligned}
\left(U_{0\bar 0}^{\dagger}\otimes\mathbf1_j\right)
L_{0j}^{[1]}(u)L_{\bar 0j}^{[1]}(-u)
\left(U_{0\bar 0}\otimes\mathbf1_j\right)
=
\begin{pmatrix}
L_j^{[2]}(u)&*&*\\
0&\mathsf S_{1,j}(u^2)\otimes\mathbf1_j&*\\
0&0&L_j^{[2]}(-u)
\end{pmatrix}.
\end{aligned}
\label{eq:supp-kall-first-block}
\end{equation}
This proves the first reflected local fusion relation and shows that its
middle block is scalar on the local physical space.  The full relation also
extends to \(u=0\) whenever the Lax operator is regular there.

\subsection{Finite reflected hierarchy}
\label{subsec:supp-finite-hierarchy}

For comparison, let \(t_m^{[a]}(x)\) denote the standard
\(A_{k-1}\) transfer matrix whose auxiliary representation is the
rectangle with \(a\) rows and \(m\) columns.  With the conventional
scalar normalization, the two-direction \(T\)-system reads
\begin{equation}
t_m^{[a]}(xq^{-1})t_m^{[a]}(xq)
=t_{m-1}^{[a]}(x)t_{m+1}^{[a]}(x)
+t_m^{[a-1]}(x)t_m^{[a+1]}(x)
\label{eq:supp-standard-T-system}
\end{equation}
\cite{KunibaNakanishiSuzuki1994}.
In particular, \(m=1\) gives
\begin{equation}
t_1^{[a]}(xq^{-1})t_1^{[a]}(xq)
=t_1^{[a-1]}(x)t_1^{[a+1]}(x)+t_2^{[a]}(x).
\label{eq:supp-standard-T-system-m1}
\end{equation}
At \(q=\ii\), set \(u=xq^{-1}=-\ii x\), so that
\(xq=-u\).  Equation~\eqref{eq:supp-standard-T-system-m1}
then compares \(t^{[a]}(u)t^{[a]}(-u)\) with the neighboring
single-column levels and the reparametrized width-two term
\(\mathcal T_2^{[a]}(u)\equiv t_2^{[a]}(x)\).  We will derive the
reflected relation directly at \(q=\ii\). The comparison with the
width-two term is made after the derivation.

\paragraph{Recursive single-column fusion.}
Put \(V^{[1]}=V\), \(F_1=\1_V\), and
\begin{equation}
u_\ell=(-1)^{\ell-1}u.
\label{eq:supp-alternating-string}
\end{equation}
Suppose that the level-\(a\) auxiliary space and its embedding
\(F_a:V^{[a]}\to V^{\otimes a}\) have already been obtained.  The
unrenormalized scattering of one fundamental auxiliary space through
this fused space is
\begin{equation}
\widetilde R^{[a,1]}(u,v)
=
(F_a^\dagger\otimes\1)
R_{1,a+1}(u_1,v)\cdots R_{a,a+1}(u_a,v)
(F_a\otimes\1).
\label{eq:supp-raw-Ra1}
\end{equation}
The factors are written in their operator order, so the rightmost one
acts first.  Projection onto the already fused space produces the
common scalar factor
\begin{equation}
\chi_a(u,v)
=\prod_{\ell=2}^{a}(v-u_\ell)
=(v+u)^{\lfloor a/2\rfloor}
(v-u)^{\lfloor(a-1)/2\rfloor}.
\label{eq:supp-Ra1-common-factor}
\end{equation}
We remove it and define
\begin{equation}
R^{[a,1]}(u,v)
=\frac{\widetilde R^{[a,1]}(u,v)}{\chi_a(u,v)}.
\label{eq:supp-normalized-Ra1}
\end{equation}
The quotient is taken after cancelling the common polynomial factor.
At \(q=\ii\), substituting a fusion value into
Eq.~\eqref{eq:supp-raw-Ra1} before this cancellation can make the
whole product vanish.  For \(a=3\) the cancellation amounts to taking
the first nonzero coefficient, whereas at higher levels the common
zero can have higher order.  Equation~\eqref{eq:supp-Ra1-common-factor} is
the general prescription.

We label an ordered color configuration by
\(I=(i_1<\cdots<i_a)\).  If \(x\) occurs in \(I\) and \(c\) does
not, \(I-x+c\) denotes the ordered list obtained by replacing \(x\)
with \(c\).  Similarly, \(K-y\) means that \(y\) is deleted from
\(K\), and \(I+c\) means that \(c\) is added and the result reordered.
Let \(n_I(c)\) be the number of colors in \(I\) smaller than \(c\),
and write \(|I,c\rangle=|I\rangle_{[a]}\otimes|c\rangle\).

\begin{lemma}[Fundamental fused \(R\) matrix]
\label{lem:supp-explicit-Ra1}
The recursive construction above gives
\begin{align}
R^{[a,1]}(u,v)|I,c\rangle
={}&(-1)^{a+n_I(c)+1}(u+v)|I,c\rangle,
&&c\in I,
\label{eq:supp-Ra1-repeated}\\
R^{[a,1]}(u,v)|I,c\rangle
={}&(-1)^{a+n_I(c)+1}
\bigg[
(u-v)|I,c\rangle
\nonumber\\[-1mm]
&\hspace{10mm}
-2v\sum_{\substack{x\in I\\x<c}}|I-x+c,x\rangle
+2u\sum_{\substack{x\in I\\x>c}}|I-x+c,x\rangle
\bigg],
&&c\notin I.
\label{eq:supp-Ra1-distinct}
\end{align}
The next fusion point is
\begin{equation}
v_*=u_{a+1}=(-1)^a u.
\label{eq:supp-next-fusion-point}
\end{equation}
At this point the two equations select precisely one state for every
choice of \(a+1\) distinct colors.
\end{lemma}

\begin{proof}
For \(a=1\), Eqs.~\eqref{eq:supp-Ra1-repeated} and
\eqref{eq:supp-Ra1-distinct} are Eq.~\eqref{eq:supp-kall-R}
written in the two-color basis.  Assume that they hold at level \(a\).

First let \(a\) be odd, so that \(v_*=-u\).  For
\(K=(k_1<\cdots<k_{a+1})\), define the normalized symmetric ket and
alternating bra
\begin{align}
|S_K^{(a+1)}\rangle
&=\frac{1}{\sqrt{a+1}}
\sum_{m=1}^{a+1}|K-k_m,k_m\rangle,\\
\langle A_K^{(a+1)}|
&=\frac{1}{\sqrt{a+1}}
\sum_{m=1}^{a+1}(-1)^{m-1}\langle K-k_m,k_m|.
\label{eq:supp-SK-AK}
\end{align}
The two formulas reduce to
\begin{equation}
R^{[a,1]}(u,-u)
=2u(a+1)\sum_{k_1<\cdots<k_{a+1}}
|S_K^{(a+1)}\rangle\langle A_K^{(a+1)}|,
\qquad a\ \text{odd}.
\label{eq:supp-Ra1-odd}
\end{equation}
For \(a=1\), this is exactly
Eq.~\eqref{eq:supp-kall-R-factorization}.  All repeated-color states
are annihilated, while the produced states are precisely the
\(|S_K^{(a+1)}\rangle\).  Since \(a+1\) is even,
\begin{equation}
\langle A_K^{(a+1)}|S_K^{(a+1)}\rangle
=\frac{1}{a+1}\sum_{m=1}^{a+1}(-1)^{m-1}=0,
\end{equation}
so Eq.~\eqref{eq:supp-Ra1-odd} squares to zero.

Now let \(a\) be even, so that \(v_*=u\).  Then
\begin{equation}
\frac{R^{[a,1]}(u,u)}{-2u}|I,c\rangle
=
\begin{cases}
(-1)^{n_I(c)}|I,c\rangle,&c\in I,\\[1mm]
(-1)^{n_I(c)}\left[
-\displaystyle\sum_{\substack{x\in I\\x<c}}|I-x+c,x\rangle
+\displaystyle\sum_{\substack{x\in I\\x>c}}|I-x+c,x\rangle
\right],&c\notin I.
\end{cases}
\label{eq:supp-Ra1-even}
\end{equation}
The repeated-color states have nonzero eigenvalues.  In a fixed
color sector \(K=(k_1<\cdots<k_{a+1})\), put
\(e_m=|K-k_m,k_m\rangle\).  If \(\sum_m z_me_m\) is annihilated,
the coefficient of \(e_n\) gives
\begin{equation}
\sum_{m<n}(-1)^{m-1}z_m
+\sum_{m>n}(-1)^m z_m=0.
\label{eq:supp-even-zero-equations}
\end{equation}
Subtracting the equation with index \(n\) from the one with index
\(n+1\) gives \(z_n=z_{n+1}\).  Thus all \(z_m\) are equal, and
because \(a+1\) is odd this equal-weight state satisfies
Eq.~\eqref{eq:supp-even-zero-equations}.  It is the only annihilated
state in this color sector.

Both parities therefore select the same one-step embedding,
\begin{equation}
C_{a+1}|K\rangle_{[a+1]}=|S_K^{(a+1)}\rangle,
\qquad
F_{a+1}=(F_a\otimes\1)C_{a+1}.
\label{eq:supp-Ca-recursion}
\end{equation}
The fused RLL relation preserves the produced channel in the odd case
and the annihilated channel in the even case, so both define a valid
level-\(a+1\) auxiliary space.

The product in Eq.~\eqref{eq:supp-raw-Ra1} now obeys
\begin{equation}
R^{[a+1,1]}(u,v)
=\frac{1}{v-u_{a+1}}
(C_{a+1}^\dagger\otimes\1)
R^{[a,1]}(u,v)R^{[1,1]}(u_{a+1},v)
(C_{a+1}\otimes\1).
\label{eq:supp-Ra1-recursion}
\end{equation}
If numerator and denominator vanish together, their common factor is
cancelled first.  Substituting the induction hypothesis and
Eq.~\eqref{eq:supp-kall-R} shows that, before division by
\(v-u_{a+1}\), the nonzero coefficients are
\begin{equation}
\begin{array}{c|c}
\text{final state}&\text{coefficient}\\ \hline
|I,c\rangle,\ c\in I
&(v-u_{a+1})(-1)^{a+n_I(c)+2}(u+v)\\
|I,c\rangle,\ c\notin I
&(v-u_{a+1})(-1)^{a+n_I(c)+2}(u-v)\\
|I-x+c,x\rangle,\ x<c
&(v-u_{a+1})(-1)^{a+n_I(c)+2}(-2v)\\
|I-x+c,x\rangle,\ x>c
&(v-u_{a+1})(-1)^{a+n_I(c)+2}(2u).
\end{array}
\label{eq:supp-Ra1-induction-table}
\end{equation}
Terms containing two different exchanges occur twice with opposite
signs and cancel.  Division by the common first factor gives
Eqs.~\eqref{eq:supp-Ra1-repeated}--\eqref{eq:supp-Ra1-distinct}
with \(a\) replaced by \(a+1\), and simultaneously proves
Eq.~\eqref{eq:supp-Ra1-common-factor}.  This closes the induction.
\end{proof}

Iterating Eq.~\eqref{eq:supp-Ca-recursion} gives
\begin{equation}
F_a|I\rangle_{[a]}
=\frac{1}{\sqrt{a!}}\sum_{\pi\in S_a}
|i_{\pi(1)},\ldots,i_{\pi(a)}\rangle,
\qquad I=(i_1<\cdots<i_a).
\label{eq:supp-higher-Fa}
\end{equation}
Here \(S_a\) denotes the \(a!\) permutations of the \(a\) positions.
Fixing the last color gives the \((a-1)!\) permutations of the
remaining colors, so the recursion produces every permutation once.
Consequently,
\begin{equation}
\dim V^{[a]}=\binom{k}{a},
\qquad
V^{[0]}=\mathbb C,
\qquad
V^{[a]}=0\quad(a>k).
\label{eq:supp-fused-dimensions}
\end{equation}

The alternating product of fundamental Lax operators preserves the
selected channel.  Its restriction defines
\begin{equation}
L_{0j}^{[a]}(u)
=(F_a^\dagger\otimes\1_j)
L_{0_1j}^{[1]}(u_1)\cdots L_{0_aj}^{[1]}(u_a)
(F_a\otimes\1_j).
\label{eq:supp-higher-La}
\end{equation}
The fused matrix \(R^{[a,b]}(u,v)\) is obtained by projecting the
ordered \(a\times b\) product of fundamental \(R\) matrices onto
\(V^{[a]}\otimes V^{[b]}\).  The Yang--Baxter equation then gives the
fused RLL and RTT relations.  Thus the
periodic fused transfer matrices
\[
t^{[a]}(u)
=\operatorname{tr}_{V^{[a]}}
L_{0N}^{[a]}(u)\cdots L_{01}^{[a]}(u)
\]
belong to one commuting family,
\begin{equation}
[t^{[a]}(u),t^{[b]}(v)]=0.
\label{eq:supp-general-L-fused-commutativity}
\end{equation}

\begin{lemma}[Fixed-boundary commutativity]
Let
\[
T_0^{[i]}(u)
=L_{0N}^{[i]}(u)\cdots L_{01}^{[i]}(u)
\]
be the level-\(i\) fused monodromy.  Choose auxiliary boundary vectors
\(\langle\alpha_i|\), \(|\beta_i\rangle\in V^{[i]}\) such that, for
every pair \(i,j\),
\begin{align}
(\langle\alpha_i|\otimes\langle\alpha_j|)
R^{[i,j]}(u,v)
&=
\rho_{ij}(u,v)
(\langle\alpha_i|\otimes\langle\alpha_j|),
\nonumber\\
R^{[i,j]}(u,v)
(|\beta_i\rangle\otimes|\beta_j\rangle)
&=
\rho_{ij}(u,v)
(|\beta_i\rangle\otimes|\beta_j\rangle),
\label{eq:supp-compatible-boundaries}
\end{align}
with the same generically nonzero character \(\rho_{ij}\) at the two
ends.  Then
\[
\tau_{\alpha_i\beta_i}^{[i]}(u)
=\langle\alpha_i|T^{[i]}(u)|\beta_i\rangle
\]
satisfies
\begin{equation}
\left[
\tau_{\alpha_i\beta_i}^{[i]}(u),
\tau_{\alpha_j\beta_j}^{[j]}(v)
\right]=0.
\label{eq:supp-fixed-boundary-commutativity}
\end{equation}
\end{lemma}

\begin{proof}
The fused local RLL relations give the fused RTT relation
\[
R^{[i,j]}(u,v)T_1^{[i]}(u)T_2^{[j]}(v)
=T_2^{[j]}(v)T_1^{[i]}(u)R^{[i,j]}(u,v).
\]
Sandwiching it between the boundary vectors in
Eq.~\eqref{eq:supp-compatible-boundaries} gives
\[
\rho_{ij}\,
\tau_{\alpha_i\beta_i}^{[i]}(u)
\tau_{\alpha_j\beta_j}^{[j]}(v)
=
\rho_{ij}\,
\tau_{\alpha_j\beta_j}^{[j]}(v)
\tau_{\alpha_i\beta_i}^{[i]}(u).
\]
Canceling the generic nonzero character proves
Eq.~\eqref{eq:supp-fixed-boundary-commutativity}. The identity extends
to its isolated zeros by continuity.
\end{proof}

For the compatible boundary family used in the Letter, we abbreviate
these operators as \(\tau_{ab}^{[i]}(u)\).  The lemma therefore gives
\[
\left[\tau_{ab}^{[i]}(u),\tau_{ab}^{[j]}(v)\right]=0,
\]
and similarly for any two compatible boundary characters.

\paragraph{Reflected pairing at fixed total level.}
Fix \(p+r=2a\), assign \(\xi_\ell=(-1)^{\ell-1}u\) to the first
fused space and \(\eta_s=(-1)^{s-1}v\) to the second, and define the
ordered product
\begin{equation}
\mathcal B_{p|r}(u,v)
=R_{p,p+r}(\xi_p,\eta_r)\cdots
R_{p,p+1}(\xi_p,\eta_1).
\label{eq:supp-Bpr}
\end{equation}
Project its input with \(F_p\otimes F_r\) and its reordered output
with \(F_{p-1}^\dagger\otimes F_{r+1}^\dagger\).  If the result
vanishes at \(v=-u\), take its first nonzero coefficient in
\(v+u\).  Removing the common nonzero scalar defines
\begin{equation}
f_{p,r}:V^{[p]}\otimes V^{[r]}\longrightarrow
V^{[p-1]}\otimes V^{[r+1]}.
\end{equation}

\begin{lemma}[Fusion-changing map]
\label{lem:supp-fusion-changing-map}
In the ordered-color basis,
\begin{equation}
f_{p,r}|I,J\rangle
=\sum_{\substack{c\in I\\c\notin J}}
(-1)^{n_J(c)}|I-c,J+c\rangle,
\label{eq:supp-color-transfer}
\end{equation}
where \(n_J(c)\) is the number of colors in \(J\) smaller than \(c\).
Moreover,
\begin{equation}
(f_{p,r}\otimes\1_j)
L_j^{[p]}(u)L_j^{[r]}(-u)
=L_j^{[p-1]}(u)L_j^{[r+1]}(-u)
(f_{p,r}\otimes\1_j),
\label{eq:supp-color-intertwining}
\end{equation}
and
\begin{equation}
f_{p-1,r+1}f_{p,r}=0.
\label{eq:supp-color-square-zero}
\end{equation}
\end{lemma}

\begin{proof}
Insert the expansions~\eqref{eq:supp-higher-Fa} on both sides of
Eq.~\eqref{eq:supp-Bpr}.  The last line of the first fused space
carries every \(c\in I\) with the same amplitude.  The final
\(F_{r+1}^\dagger\) projection removes \(c\in J\), since the right
color list would then repeat a color.  For \(c\notin J\), direct use
of Eq.~\eqref{eq:supp-kall-R} gives a common coefficient independent
of \(c\). After it is removed, reordering \(J+c\) contributes
\((-1)^{n_J(c)}\).  This proves Eq.~\eqref{eq:supp-color-transfer}.
Applying the fundamental RLL relation successively to the same
ordered product gives Eq.~\eqref{eq:supp-color-intertwining}. If a
common scalar vanishes, the equality is taken at its first nonzero
coefficient.  Finally, moving two distinct colors in the two possible
orders gives the same final lists with opposite signs, proving
Eq.~\eqref{eq:supp-color-square-zero}.
\end{proof}

We now show that these maps pair all neighboring levels except for one
central block.  Fix the colors \(K\) that occur on both sides and the
colors \(\Omega\) that occur on only one side.  Let \(A\) denote the
colors that occur only on the left.  When \(\Omega\) contains at least
one color, choose \(c_0\in\Omega\) and define, only inside this proof,
\begin{equation}
h|A\rangle=
\begin{cases}
0,&c_0\in A,\\[1mm]
w_{c_0}^{-1}(-1)^{n_A(c_0)}|A+c_0\rangle,&c_0\notin A,
\end{cases}
\qquad
w_c=(-1)^{n_K(c)+n_\Omega(c)}.
\label{eq:supp-backward-map}
\end{equation}
Direct substitution into Eq.~\eqref{eq:supp-color-transfer} gives
\begin{equation}
fh+hf=\1.
\label{eq:supp-two-way-transfer}
\end{equation}
Therefore every state annihilated by the next \(f\) is produced by
the preceding one.  The only unpaired groups have no color in
\(\Omega\), which means \(I=J\).  This can occur only at \(p=r=a\),
where the unpaired states span
\begin{equation}
\mathcal M_a
=\operatorname{span}\{|I,I\rangle:|I|=a\}
\subset V^{[a]}\otimes V^{[a]},
\qquad
\mathcal M_a\simeq V^{[a]}.
\label{eq:supp-locked-space}
\end{equation}

Set
\begin{equation}
X_{p,j}^{(a)}(u)
=L_j^{[p]}(u)L_j^{[2a-p]}(-u).
\label{eq:supp-Xpa}
\end{equation}

\begin{lemma}[Paired auxiliary trace]
\label{lem:supp-paired-trace}
Let \(Y_p\) act on \(V^{[p]}\otimes V^{[2a-p]}\) and satisfy
\begin{equation}
f_{p,2a-p}Y_p=Y_{p-1}f_{p,2a-p}.
\end{equation}
Then
\begin{equation}
\sum_{\substack{0\leq p\leq k\\0\leq2a-p\leq k}}
(-1)^{a-p}\operatorname{tr}Y_p
=\operatorname{tr}_{\mathcal M_a}Y_a^0,
\label{eq:supp-paired-trace}
\end{equation}
where \(Y_a^0\) is the diagonal block on \(\mathcal M_a\).
\end{lemma}

\begin{proof}
Equation~\eqref{eq:supp-color-intertwining} preserves the states
produced by the preceding \(f\), as well as the larger space of states
annihilated by the next \(f\).  Hence \(Y_p\) is upper block
triangular in the order
\[
(\text{produced},\text{ locked},\text{ remaining}).
\]
Choose the remaining block complementary to the annihilated states.
The map \(f_{p,2a-p}\) is one to one from this block onto the produced
block at \(p-1\), and the intertwining relation identifies their
diagonal matrices.  Their traces occur with opposite signs in
Eq.~\eqref{eq:supp-paired-trace} and cancel.  Only the locked block at
\(p=a\) remains.
\end{proof}

At \(p=a\), let
\begin{equation}
\mathcal C_j^{[a]}(u)
=P_{\mathcal M_a}X_{a,j}^{(a)}(u)P_{\mathcal M_a}
\label{eq:supp-locked-local-block}
\end{equation}
be the locked diagonal block.  The upper-triangular form implies that
the locked block of a product of local operators is the product of
their locked blocks.

\paragraph{Scalar closure.}
Introduce a third fundamental auxiliary space and set
\begin{equation}
\mathcal A_{p,r}(u,v)
=R_{13}^{[p,1]}(u,v)R_{23}^{[r,1]}(-u,v).
\label{eq:supp-Apr}
\end{equation}
For \(p+r=2a\), \(p,r\geq1\), the normalization factors satisfy
\begin{equation}
\chi_p(u,v)\chi_r(-u,v)=(v^2-u^2)^{a-1},
\label{eq:supp-A-common-normalization}
\end{equation}
which is independent of \(p\).  The fused Yang--Baxter equation and
Eq.~\eqref{eq:supp-color-intertwining} therefore give
\begin{equation}
(f_{p,r}\otimes\1_3)\mathcal A_{p,r}(u,v)
=\mathcal A_{p-1,r+1}(u,v)(f_{p,r}\otimes\1_3).
\label{eq:supp-A-intertwining}
\end{equation}
Thus \(\mathcal A_{a,a}\) has the same upper-triangular form.  Its
locked diagonal block is
\begin{equation}
\left[
R_{13}^{[a,1]}(u,v)R_{23}^{[a,1]}(-u,v)
\right]^0
=(v^2-u^2)\1_{\mathcal M_a}\otimes\1_3.
\label{eq:supp-higher-middle-scalar-train}
\end{equation}
This is the normalized convention~\eqref{eq:supp-normalized-Ra1}.
For the unrenormalized matrices~\eqref{eq:supp-raw-Ra1},
Eq.~\eqref{eq:supp-A-common-normalization} restores the factor
\((v^2-u^2)^a\).  Terms that change only one auxiliary color list lie
in the produced block and do not contribute to the locked diagonal
block.

Two fused RLL relations give
\begin{equation}
\mathcal A_{a,a}(u,v)X_{a,j}^{(a)}(u)L_{3j}^{[1]}(v)
=L_{3j}^{[1]}(v)X_{a,j}^{(a)}(u)\mathcal A_{a,a}(u,v).
\label{eq:supp-three-line-RLL-general}
\end{equation}
Taking the locked diagonal block and cancelling \(v^2-u^2\) gives,
for generic \(u,v\),
\begin{equation}
\mathcal C_j^{[a]}(u)L_{3j}^{[1]}(v)
=L_{3j}^{[1]}(v)\mathcal C_j^{[a]}(u).
\label{eq:supp-higher-middle-commutes}
\end{equation}
Continuity extends the identity to all regular spectral values.  By
local irreducibility and Schur's lemma,
\begin{equation}
\mathcal C_j^{[a]}(u)
=\mathsf S_{a,j}(u^2)\otimes\1_{\mathcal H_j}.
\label{eq:supp-higher-middle-Schur}
\end{equation}
Exchanging the two reflected auxiliary spaces sends \(u\) to \(-u\).
After Eq.~\eqref{eq:supp-higher-middle-Schur}, cyclicity of the
physical trace shows that \(\mathsf S_{a,j}(u)=\mathsf S_{a,j}(-u)\),
which justifies writing its argument as \(u^2\).

We use \(t^{[0]}(u)=\1\) and \(t^{[r]}(u)=0\) outside
\(0\leq r\leq k\).

\begin{theorem}[Finite reflected hierarchy]
\label{thm:supp-kall-finite-hierarchy}
For every finite-dimensional local realization of
Eq.~\eqref{eq:supp-framework-RLL} that is regular at the reflected
fusion points and irreducible on \(\mathcal H_j\), the fused transfer
matrices satisfy, for \(a=1,\ldots,k-1\),
\begin{equation}
\begin{aligned}
t^{[a]}(u)t^{[a]}(-u)
={}&\Phi_{a,N}(u^2)\1
+\sum_{b=1}^{\min(a,k-a)}(-1)^{b-1}
\Big[
t^{[a-b]}(u)t^{[a+b]}(-u)
+t^{[a-b]}(-u)t^{[a+b]}(u)
\Big],
\end{aligned}
\label{eq:supp-general-L-finite-hierarchy}
\end{equation}
where
\begin{equation}
\Phi_{a,N}(u^2)
=\operatorname{tr}_{\mathcal M_a}
\left[
\mathsf S_{a,N}(u^2)\cdots\mathsf S_{a,1}(u^2)
\right].
\label{eq:supp-general-scalar-source}
\end{equation}
\end{theorem}

\begin{proof}
Apply Lemma~\ref{lem:supp-paired-trace} to the doubled fused
monodromies
\[
Y_p
=\prod_{j=N}^{1}
\left[
L_j^{[p]}(u)L_j^{[2a-p]}(-u)
\right].
\]
Equation~\eqref{eq:supp-color-intertwining} intertwines these
monodromies, while the locked block is the product of the local
matrices in Eq.~\eqref{eq:supp-higher-middle-Schur}.  The doubled
auxiliary trace of \(Y_p\) is
\(t^{[p]}(u)t^{[2a-p]}(-u)\).  Hence
\begin{equation}
\sum_{\substack{0\leq p,q\leq k\\p+q=2a}}
(-1)^{a-p}t^{[p]}(u)t^{[q]}(-u)
=\Phi_{a,N}(u^2)\1 .
\label{eq:supp-master-reflected-identity}
\end{equation}
Isolating \(p=q=a\), pairing \(p=a-b\) with \(p=a+b\), and using
Eq.~\eqref{eq:supp-general-L-fused-commutativity} gives
Eq.~\eqref{eq:supp-general-L-finite-hierarchy}.
\end{proof}

Comparing Eq.~\eqref{eq:supp-general-L-finite-hierarchy} with the
standard relation~\eqref{eq:supp-standard-T-system-m1} shows explicitly
what happens to the width-two term at \(q=\ii\).  In the reflected
endpoint convention,
\begin{align}
\mathcal T_2^{[a]}(u)
={}&\Phi_{a,N}(u^2)\1
+t^{[a-1]}(u)t^{[a+1]}(-u)
+\sum_{b=2}^{\min(a,k-a)}(-1)^{b-1}
\Big[
t^{[a-b]}(u)t^{[a+b]}(-u)
+t^{[a-b]}(-u)t^{[a+b]}(u)
\Big].
\label{eq:supp-width-two-folding}
\end{align}
Thus the generic width direction folds into the single-column
hierarchy plus the scalar source.

At \(a=k\), Eq.~\eqref{eq:supp-higher-Fa} leaves the unique state
containing all \(k\) colors.  Standard top-level fusion identifies
the corresponding transfer matrix \(t^{[k]}(u)\) with the quantum
determinant~\cite{KulishReshetikhinSklyanin1981,
KunibaNakanishiSuzuki1994}, whose auxiliary space is one dimensional.  The hierarchy therefore closes after
\(k-1\) nontrivial levels.  For \(k=3\),
only \(t^{[1]}\) and \(t^{[2]}\) remain, giving the two coupled
relations used below.

\section{Closure of the inhomogeneous \texorpdfstring{$k=3$}{k=3} FFD chain}
\label{sec:supp-twisted-closure}

\subsection{Commuting family and Hamiltonian}
\label{subsec:supp-ffd-commuting}

In the Lax operator of Eq.~\eqref{eq:L-FFD} the spectral parameter
occupies a single entry, so the site-dependent coupling enters only
through the product \(h_ju\),
\begin{equation}
L^{[1]}_{\mathrm F,j}(u)=\mathcal L_j(h_ju),
\qquad
\mathcal L_j(w)=
\begin{pmatrix}
\1_j&0&-w\,\sigma^x_j\\
\sigma^z_j&0&0\\
0&\sigma^z_j&0
\end{pmatrix}.
\label{eq:supp-twist-G}
\end{equation}
The \(R\) matrix~\eqref{eq:HFF-R} is homogeneous of degree one,
\(R_{0\bar0}(cu,cv)=c\,R_{0\bar0}(u,v)\), so the RLL
relation~\eqref{eq:RLL-framework} at site \(j\) is the homogeneous one
evaluated at the rescaled arguments \(h_ju\) and \(h_jv\), and the same
\(R_{0\bar0}(u,v)\) intertwines every site.  The transfer matrices
\begin{equation}
t^{[a]}_{\mathrm F}(u)
=\operatorname{tr}_{V^{[a]}}
\big[L^{[a]}_{\mathrm F,N}(u)\cdots L^{[a]}_{\mathrm F,1}(u)\big],
\qquad a=1,2,3,
\label{eq:supp-twist-transfer}
\end{equation}
therefore form one commuting family for every choice of the couplings
\(\{h_j\}_{j=1}^{N}\).

At \(u=0\) the only closed auxiliary path of the monodromy is the one
that stays in the first auxiliary state, so
\(t^{[1]}_{\mathrm F}(0)=\1\), and expanding to first order in \(u\)
gives \(t^{[1]}_{\mathrm F}(u)=\1-uH_{\mathrm F}+O(u^2)\) with
\begin{equation}
H_{\mathrm F}
=\sum_{j=1}^{N}h_j\,\sigma^z_{j-2}\sigma^z_{j-1}\sigma^x_j ,
\qquad \text{indices mod }N,
\label{eq:supp-twist-H}
\end{equation}
which is Eq.~\eqref{eq:FFD-general-Hamiltonian} written cyclically:
\(h_j\) multiplies the bond whose \(\sigma^x\) factor sits on site
\(j\).

Every bond carries its own coupling.  All \(h_j=1\) is the periodic
chain, a single vanishing coupling is the one-wrap closure, and two
cyclically adjacent vanishing couplings, \(h_1=h_2=0\), remove both
wrap-around bonds and leave the open chain of
Ref.~\cite{Fendley2019}.  Two symmetries are used below.  Conjugation
with \(\sigma^x_m\) flips the sign of the two bonds whose \(\sigma^z\)
factors touch site \(m\), that is
\((h_{m+1},h_{m+2})\mapsto(-h_{m+1},-h_{m+2})\), and conjugation with
\(\prod_j\sigma^z_j\) sends \(H_{\mathrm F}\mapsto-H_{\mathrm F}\), so
the spectrum is symmetric about zero.  A common rescaling
\(h_j\mapsto c\,h_j\) is the change of variable \(u\mapsto cu\) and
multiplies all energies by \(c\).

\subsection{The two fusion relations}
\label{subsec:supp-ffd-fusion-relations}

The construction of Sec.~\ref{sec:supp-kall-fusion} is local in the site
index and nowhere uses homogeneity, so it applies without change with
the site-dependent scalar matrices
\(\mathsf S_{a,j}(u^2)=\mathsf S_a(h_j^2u^2)\).  At the first level,
reading Eq.~\eqref{eq:L-FFD} entry by entry gives
\(L^{[1]}_{\mathrm F,ab,j}(u)L^{[1]}_{\mathrm F,ab,j}(-u)
=[\mathsf S_1(h_j^2u^2)]_{ab}\1_j\) with \(\mathsf S_1\) the level-one
scalar matrix of Eq.~\eqref{eq:FFD-scalar-matrices}, and the trace of
the block-triangular reflected monodromy gives
\begin{equation}
t^{[1]}_{\mathrm F}(u)\,t^{[1]}_{\mathrm F}(-u)
=\Phi^{\mathrm F}_{1,N}(u^2)\,\1
+t^{[2]}_{\mathrm F}(u)+t^{[2]}_{\mathrm F}(-u),
\label{eq:supp-twist-rel1}
\end{equation}
where the locked trace is the ordered product
\begin{equation}
\Phi^{\mathrm F}_{1,N}(u^2)
=\operatorname{tr}\big[\mathsf S_1(h_N^2u^2)\cdots\mathsf S_1(h_1^2u^2)\big].
\label{eq:supp-twist-Phi1}
\end{equation}

At the second level the same three-layer decomposition of
\(L^{[2]}_{\mathrm F,j}(u)L^{[2]}_{\mathrm F,j}(-u)\) gives
\begin{equation}
t^{[2]}_{\mathrm F}(u)\,t^{[2]}_{\mathrm F}(-u)
=\Phi^{\mathrm F}_{2,N}(u^2)\,\1
+t^{[3]}_{\mathrm F}(u)\,t^{[1]}_{\mathrm F}(-u)
+t^{[3]}_{\mathrm F}(-u)\,t^{[1]}_{\mathrm F}(u),
\label{eq:supp-twist-rel2-general}
\end{equation}
with
\begin{equation}
\Phi^{\mathrm F}_{2,N}(u^2)
=\operatorname{tr}\big[\mathsf S_2(h_N^2u^2)\cdots\mathsf S_2(h_1^2u^2)\big].
\label{eq:supp-twist-Phi2}
\end{equation}
Here \(\mathsf S_2(h_j^2u^2)\) is the locked block of
\(L^{[2]}_{\mathrm F,j}(u)L^{[2]}_{\mathrm F,j}(-u)\) in the basis
\((s_{12},s_{13},s_{23})\), whose entries are the products
\([L^{[2]}_{\mathrm F,j}(u)]_{I'I}[L^{[2]}_{\mathrm F,j}(-u)]_{I'I}\)
read off from Eq.~\eqref{eq:FFD-fused-Lax}, and it is the matrix of
\(2\times2\) subdeterminants of \(\mathsf S_1(h_j^2u^2)\) in the same
basis.  The top level is one dimensional.  With
\(F_3|123\rangle_{[3]}=6^{-1/2}\sum_{\pi\in S_3}|\pi(1),\pi(2),\pi(3)\rangle\),
the projection of Eq.~\eqref{eq:supp-higher-La} gives
\begin{equation}
L^{[3]}_{\mathrm F,j}(u)
=(F_3^\dagger\otimes\1_j)
L^{[1]}_{\mathrm F,0_1j}(u)L^{[1]}_{\mathrm F,0_2j}(-u)L^{[1]}_{\mathrm F,0_3j}(u)
(F_3\otimes\1_j)
=-h_ju\,\sigma^x_j ,
\label{eq:supp-twist-L3}
\end{equation}
so that
\begin{equation}
t^{[3]}_{\mathrm F}(u)=\Big(\prod_{j=1}^{N}h_j\Big)(-u)^N\mathcal X_N,
\qquad
\mathcal X_N=\prod_{j=1}^N\sigma^x_j ,
\label{eq:supp-twist-top}
\end{equation}
and Eq.~\eqref{eq:supp-twist-rel2-general} becomes
\begin{equation}
t^{[2]}_{\mathrm F}(u)\,t^{[2]}_{\mathrm F}(-u)
=\Phi^{\mathrm F}_{2,N}(u^2)\,\1
+\Big(\prod_{j=1}^{N}h_j\Big)u^N\mathcal X_N
\Big[t^{[1]}_{\mathrm F}(u)+(-1)^N t^{[1]}_{\mathrm F}(-u)\Big].
\label{eq:supp-twist-rel2}
\end{equation}
Equations~\eqref{eq:supp-twist-rel1} and \eqref{eq:supp-twist-rel2} are
the closed \(k=3\) hierarchy of the inhomogeneous chain.  The couplings
enter in two places only, through the scalar sources
\(\Phi^{\mathrm F}_{a,N}\) and through the product \(\prod_jh_j\) in
front of the feedback term.  The remaining level follows the same rule,
\(\Phi^{\mathrm F}_{3,N}
=\det\big[\mathsf S_1(h_N^2u^2)\cdots\mathsf S_1(h_1^2u^2)\big]
=\big(\prod_jh_j^2\big)(-u^2)^N\), which reproduces
\(t^{[3]}_{\mathrm F}(u)t^{[3]}_{\mathrm F}(-u)\).

\subsection{Eigenvalue form}
\label{subsec:supp-ffd-eigenvalue-form}

All \(t^{[a]}_{\mathrm F}\) commute with \(\mathcal X_N\).  In a sector
\(\mathcal X_N=\chi\), the joint eigenvalues \(\Lambda^{[1]}_\chi(u)\)
and \(\Lambda^{[2]}_\chi(u)\) obey
\begin{align}
\Lambda^{[1]}_\chi(u)\Lambda^{[1]}_\chi(-u)
&=\Phi^{\mathrm F}_{1,N}(u^2)+\Lambda^{[2]}_\chi(u)+\Lambda^{[2]}_\chi(-u),
\label{eq:supp-twist-eigen1}\\
\Lambda^{[2]}_\chi(u)\Lambda^{[2]}_\chi(-u)
&=\Phi^{\mathrm F}_{2,N}(u^2)
+\Big(\prod_{j}h_j\Big)\chi\,u^N
\Big[\Lambda^{[1]}_\chi(u)+(-1)^N\Lambda^{[1]}_\chi(-u)\Big],
\label{eq:supp-twist-eigen2}
\end{align}
which are Eqs.~\eqref{eq:FFD-eigenvalue-1} and
\eqref{eq:FFD-eigenvalue-2} of the Letter.  The coefficient of the
lowest power of \(\Lambda^{[2]}_\chi\) follows from the lowest
coefficient of \(\Phi^{\mathrm F}_{2,N}\), which is controlled by the
two sublattice products
\begin{equation}
P_{\mathrm o}=\prod_{j\ \mathrm{odd}}h_j ,
\qquad
P_{\mathrm e}=\prod_{j\ \mathrm{even}}h_j ,
\qquad
P_{\mathrm o}P_{\mathrm e}=\prod_{j=1}^{N}h_j .
\label{eq:supp-sublattice-products}
\end{equation}

\begin{lemma}[Lowest coefficient of the fused scalar source]
\label{lem:supp-psi0}
For even \(N\), the lowest power of \(u\) in
\(\Phi^{\mathrm F}_{2,N}(u^2)\) is \(u^{N}\), with coefficient
\begin{equation}
P_{\mathrm o}^2+P_{\mathrm e}^2 .
\label{eq:supp-psi0-value}
\end{equation}
\end{lemma}

\begin{proof}
Read \(\mathsf S_2(h_j^2u^2)\) as the step-\(j\) transfer matrix of a
directed graph on the three states \(12\), \(13\), \(23\), with the
moves \(12\to13\) and \(13\to23\) of weight \(h_j^2u^2\) and the moves
\(13\to12\) and \(23\to12\) of weight \(1\).  A diagonal entry of the
ordered product in Eq.~\eqref{eq:supp-twist-Phi2} is the sum over closed
walks of length \(N\), the step taken at position \(j\) carrying
\(h_j^2u^2\) whenever it is one of the two weighted moves.  The graph
has two simple cycles, the two-cycle \(12\to13\to12\), which uses one
weighted move in two steps, and the three-cycle
\(12\to13\to23\to12\), which uses two weighted moves in three steps.  A
closed walk built from \(n_2\) two-cycles and \(n_3\) three-cycles has
\(2n_2+3n_3=N\) and carries \(u^{\,N+n_3}\), so for even \(N\) the
lowest power \(u^{N}\) requires \(n_3=0\), that is the alternating walk
\(12\to13\to12\to\cdots\).  There are exactly two such walks, one
starting at \(12\) and one starting at \(13\), and they take the
weighted move on the two complementary sublattices of positions, so
they contribute \(\prod_{j\ \mathrm{odd}}h_j^2\) and
\(\prod_{j\ \mathrm{even}}h_j^2\).  The state \(23\) lies only on the
three-cycle, so its diagonal entry starts at a higher power.  Summing
the three diagonal entries gives Eq.~\eqref{eq:supp-psi0-value}.
\end{proof}

For even \(N\) the lowest power of \(\Lambda^{[2]}_\chi\) is
\(u^{N/2}\) (this is shown below), so the lowest term on the left of
Eq.~\eqref{eq:supp-twist-eigen2} is \((-1)^{N/2}\kappa_\chi^2u^{N}\),
where \(\kappa_\chi\) is the coefficient of \(u^{N/2}\).  On the right,
\(\Phi^{\mathrm F}_{2,N}\) contributes
\((P_{\mathrm o}^2+P_{\mathrm e}^2)u^N\) by
Lemma~\ref{lem:supp-psi0} and the feedback term contributes
\(2P_{\mathrm o}P_{\mathrm e}\chi u^N\) because
\(\Lambda^{[1]}_\chi(0)=1\).  Hence
\begin{equation}
(-1)^{N/2}\kappa_\chi^2=\big(P_{\mathrm o}+\chi P_{\mathrm e}\big)^2 .
\label{eq:supp-twist-kappa}
\end{equation}
This holds already as an operator identity, the coefficient of
\(u^{N/2}\) in \(t^{[2]}_{\mathrm F}(u)\) squaring to
\((-1)^{N/2}(P_{\mathrm o}+P_{\mathrm e}\mathcal X_N)^2\), which we have
checked directly at \(N=6\) and \(N=8\) for random couplings.  For the
family used in Fig.~\ref{fig:main}, \(h_1=h\) and
\(h_2=\cdots=h_N=1\) with \(N\equiv0\pmod4\), it gives
\(P_{\mathrm o}=h\), \(P_{\mathrm e}=1\) and
\(\kappa_{+1}=\pm(1+h)\).  The zero-root solutions fix the sign: at
\(h=0\) the ground state has \(\kappa=+1\) with all fused roots in one
half plane and the lowest excitation has \(\kappa=-1\) with the root
closest to the origin reflected, and since \(\kappa^2=(1+h)^2\) never
vanishes for \(0\leq h\leq1\) the sign is preserved along each state,
which gives the values \(\kappa=1,-1,2,-2\) quoted in the caption of
Fig.~\ref{fig:main} (confirmed by exact diagonalization at \(N=8\) for
\(h=0,\tfrac12,1\)).  For \(\chi=-1\) and
\(P_{\mathrm o}=P_{\mathrm e}\) the coefficient vanishes and
\(\Lambda^{[2]}_{-1}\) starts at a higher power, which is visible in
Table~\ref{tab:N12-complete-zero-root}.

Two limits are visible directly in Eqs.~\eqref{eq:supp-twist-eigen1} and
\eqref{eq:supp-twist-eigen2}.  When \(\prod_jh_j=0\) the feedback term
vanishes, the second relation no longer involves \(\chi\) or
\(\Lambda^{[1]}_\chi\), and the two relations are solved one after the
other: the fused eigenvalue is fixed by
\(\Lambda^{[2]}(u)\Lambda^{[2]}(-u)=\Phi^{\mathrm F}_{2,N}(u^2)\) up to
the sign choices of its roots, and the fundamental eigenvalue is then
free-fermion like within each such choice.  This is the sector-wise free
structure of the one-wrap chain, and the two \(\chi\) sectors are
degenerate there.  When \(\prod_jh_j\neq0\) the feedback couples the two
levels, the signs of the fundamental roots are no longer free, and the
spectrum is that of an interacting chain.  

\subsection{Zero-root parametrization}
\label{subsec:supp-ffd-zero-roots}

The polynomial structure of the two eigenvalues is read off from the
auxiliary paths of the two Lax operators, and it does not depend on the
values of the couplings.  In \(L^{[1]}_{\mathrm F,j}\) of
Eq.~\eqref{eq:L-FFD} the only entry carrying \(u\) is the return
\(3\to1\), and the auxiliary state can only move along
\(1\to2\to3\to1\) or stay at \(1\).  A closed path of length \(N\)
therefore consists of \(m\) cycles of length three and \(N-3m\) rests,
and carries \(u^{m}\) together with \(m\) of the couplings, so
\(t^{[1]}_{\mathrm F}(u)\) is a polynomial of degree
\begin{equation}
d_N=\lfloor N/3\rfloor,
\qquad t^{[1]}_{\mathrm F}(0)=\1 .
\label{eq:supp-twist-degree1}
\end{equation}
In \(L^{[2]}_{\mathrm F,j}\) of Eq.~\eqref{eq:FFD-fused-Lax} the moves
are \(s_{12}\to s_{13}\) with weight \(\ii h_ju\sigma^y_j\),
\(s_{13}\to s_{12}\) with weight \(\sigma^z_j\), \(s_{13}\to s_{23}\)
with weight \(\ii h_ju\sigma^y_j\), and \(s_{23}\to s_{12}\) with weight
\(\1_j\).  The closed paths are built from the two-cycle
\(s_{12}\to s_{13}\to s_{12}\), which carries one power of \(u\), and
the three-cycle \(s_{12}\to s_{13}\to s_{23}\to s_{12}\), which carries
two.  A path with \(n_2\) two-cycles and \(n_3\) three-cycles has
\(2n_2+3n_3=N\) and carries \(u^{\,n_2+2n_3}=u^{(N+n_3)/2}\), so the
powers present in \(t^{[2]}_{\mathrm F}(u)\) run from the smallest
admissible \(n_3\), which is \(0\) for even \(N\) and \(1\) for odd
\(N\), to the largest,
\begin{equation}
s_N=\lfloor(N+1)/2\rfloor\;\leq\;\text{power}\;\leq\;f_N=\lfloor2N/3\rfloor .
\label{eq:supp-twist-degree2}
\end{equation}
These are operator statements, so every joint eigenvalue is of the form
\begin{equation}
\Lambda^{[1]}_\chi(u)=\prod_{j=1}^{d_N}\Big(1-\frac{u}{\lambda_j}\Big),
\qquad
\Lambda^{[2]}_\chi(u)=\kappa_\chi\,u^{s_N}\prod_{\alpha=1}^{n_N}\Big(1-\frac{u}{\gamma_\alpha}\Big),
\qquad
n_N=f_N-s_N ,
\label{eq:supp-twist-roots}
\end{equation}
which is the parametrization of Eqs.~\eqref{eq:FFD-fundamental-roots}
and \eqref{eq:FFD-fused-roots} of the Letter, with the understanding that
a vanishing top coefficient is a root at infinity.  Equation
\eqref{eq:supp-twist-roots} is understood as the regular root chart with
\(\kappa_\chi\neq0\): if the coefficient of \(u^{s_N}\) vanishes, we take
the first nonzero power \(s_\chi>s_N\), replace \(n_N\) by
\(f_N-s_\chi\), and use its nonzero coefficient as \(\kappa_\chi\), with
the identically vanishing eigenvalue treated separately and the same
convention applied to the FP parametrization below.
Since
\(t^{[1]}_{\mathrm F}(u)=\1-uH_{\mathrm F}+O(u^2)\), the energy is
\(E=\sum_{j}\lambda_j^{-1}\), Eq.~\eqref{eq:FFD-energy}.

On the regular chart, the unknowns are the \(d_N\) fundamental roots, the \(n_N\) fused roots,
and \(\kappa_\chi\), that is \(r_N=d_N+n_N+1\) numbers in each sector.
Comparing the coefficients of \(u\) in the two relations, both of which
are even in \(u\), Eq.~\eqref{eq:supp-twist-eigen1} gives \(d_N\)
equations from \(u^{2},\ldots,u^{2d_N}\) and
Eq.~\eqref{eq:supp-twist-eigen2} gives \(n_N+1\) equations from
\(u^{2s_N},\ldots,u^{2f_N}\), all other powers vanishing identically by
Eqs.~\eqref{eq:supp-twist-degree1} and \eqref{eq:supp-twist-degree2}.
The number of equations therefore equals the number of unknowns,
\(r_N\) in each sector, for every choice of the couplings
\(\{h_j\}\).
Table~\ref{tab:N8-twisted-zero-root} lists all solutions for \(N=8\) at
\(h_1=0.7\), \(h_2=0.5\), \(h_3=\cdots=h_8=1\), and
Table~\ref{tab:N12-complete-zero-root} all solutions for \(N=12\) at the
periodic point \(h_1=\cdots=h_{12}=1\), the size class used in
Fig.~\ref{fig:main}.  In both cases they coincide with the joint
eigenvalues of
\((t^{[1]}_{\mathrm F},t^{[2]}_{\mathrm F},\mathcal X_N)\) from exact
diagonalization, every branch being highly degenerate.

\begin{table*}[p]
\caption{Complete \(N=8\) zero-root branches of the inhomogeneous FFD
chain with \(h_1=0.7\), \(h_2=0.5\) and \(h_3=\cdots=h_8=1\).  The
fundamental eigenvalue is \(\Lambda^{[1]}_\chi(u)=\prod_{j=1}^{2}(1-u/\lambda_j)\)
and the fused eigenvalue is
\(\Lambda^{[2]}_\chi(u)=\kappa_\chi u^{4}(1-u/\gamma_1)\), with
\(\kappa_\chi=\pm(h_1+h_2\chi)\) by Eq.~\eqref{eq:supp-twist-kappa}.
Consecutive rows with \(\pm\gamma_1\) are the partners
\(\Lambda^{[2]}_\chi(u)\leftrightarrow\Lambda^{[2]}_\chi(-u)\) at equal
energy.  The energies and degeneracies of all rows agree with exact
diagonalization, and the degeneracies sum to \(256=2^{8}\).}
\label{tab:N8-twisted-zero-root}
\centering
\begingroup
\setlength{\tabcolsep}{3pt}
\renewcommand{\arraystretch}{0.86}
\begin{tabular}{c c c c c c c}
\hline
level & \(\chi\) & \(E\) & degeneracy & \(\{\lambda_j\}\) &
\multicolumn{1}{c}{\(\kappa_\chi\)} & \multicolumn{1}{c}{\(\{\gamma_\alpha/\ii\}\)}\\
\hline
1 & \(+1\) & \(-3.6411081\) & 8 & \(\{-0.6309969,-0.4863070\}\) & \(+1.2000000\) & \(\{-0.5508191\}\)\\
2 & \(+1\) & \(-3.6411081\) & 8 & \(\{-0.6309969,-0.4863070\}\) & \(+1.2000000\) & \(\{+0.5508191\}\)\\
3 & \(-1\) & \(-3.5513332\) & 8 & \(\{-0.7634606,-0.4461282\}\) & \(+0.2000000\) & \(\{+0.3124194\}\)\\
4 & \(-1\) & \(-3.5513332\) & 8 & \(\{-0.7634606,-0.4461282\}\) & \(+0.2000000\) & \(\{-0.3124194\}\)\\
5 & \(-1\) & \(-3.5118161\) & 8 & \(\{-0.8194147,-0.4364082\}\) & \(-0.2000000\) & \(\{-0.2807445\}\)\\
6 & \(-1\) & \(-3.5118161\) & 8 & \(\{-0.8194147,-0.4364082\}\) & \(-0.2000000\) & \(\{+0.2807445\}\)\\
7 & \(+1\) & \(-3.4007257\) & 8 & \(\{-0.9916370,-0.4180091\}\) & \(-1.2000000\) & \(\{+0.5887925\}\)\\
8 & \(+1\) & \(-3.4007257\) & 8 & \(\{-0.9916370,-0.4180091\}\) & \(-1.2000000\) & \(\{-0.5887925\}\)\\
9 & \(+1\) & \(-1.3838587\) & 8 & \(\{-0.4180091,+0.9916370\}\) & \(-1.2000000\) & \(\{-1.3619906\}\)\\
10 & \(+1\) & \(-1.3838587\) & 8 & \(\{-0.4180091,+0.9916370\}\) & \(-1.2000000\) & \(\{+1.3619906\}\)\\
11 & \(-1\) & \(-1.0710497\) & 8 & \(\{-0.4364082,+0.8194147\}\) & \(-0.2000000\) & \(\{-0.0951034\}\)\\
12 & \(-1\) & \(-1.0710497\) & 8 & \(\{-0.4364082,+0.8194147\}\) & \(-0.2000000\) & \(\{+0.0951034\}\)\\
13 & \(-1\) & \(-0.9316827\) & 8 & \(\{-0.4461282,+0.7634606\}\) & \(+0.2000000\) & \(\{+0.0940701\}\)\\
14 & \(-1\) & \(-0.9316827\) & 8 & \(\{-0.4461282,+0.7634606\}\) & \(+0.2000000\) & \(\{-0.0940701\}\)\\
15 & \(+1\) & \(-0.4715206\) & 8 & \(\{-0.4863070,+0.6309969\}\) & \(+1.2000000\) & \(\{+2.7989135\}\)\\
16 & \(+1\) & \(-0.4715206\) & 8 & \(\{-0.4863070,+0.6309969\}\) & \(+1.2000000\) & \(\{-2.7989135\}\)\\
17 & \(+1\) & \(+0.4715206\) & 8 & \(\{-0.6309969,+0.4863070\}\) & \(+1.2000000\) & \(\{-2.7989135\}\)\\
18 & \(+1\) & \(+0.4715206\) & 8 & \(\{-0.6309969,+0.4863070\}\) & \(+1.2000000\) & \(\{+2.7989135\}\)\\
19 & \(-1\) & \(+0.9316827\) & 8 & \(\{-0.7634606,+0.4461282\}\) & \(+0.2000000\) & \(\{-0.0940701\}\)\\
20 & \(-1\) & \(+0.9316827\) & 8 & \(\{-0.7634606,+0.4461282\}\) & \(+0.2000000\) & \(\{+0.0940701\}\)\\
21 & \(-1\) & \(+1.0710497\) & 8 & \(\{-0.8194147,+0.4364082\}\) & \(-0.2000000\) & \(\{-0.0951034\}\)\\
22 & \(-1\) & \(+1.0710497\) & 8 & \(\{-0.8194147,+0.4364082\}\) & \(-0.2000000\) & \(\{+0.0951034\}\)\\
23 & \(+1\) & \(+1.3838587\) & 8 & \(\{-0.9916370,+0.4180091\}\) & \(-1.2000000\) & \(\{+1.3619906\}\)\\
24 & \(+1\) & \(+1.3838587\) & 8 & \(\{-0.9916370,+0.4180091\}\) & \(-1.2000000\) & \(\{-1.3619906\}\)\\
25 & \(+1\) & \(+3.4007257\) & 8 & \(\{+0.4180091,+0.9916370\}\) & \(-1.2000000\) & \(\{+0.5887925\}\)\\
26 & \(+1\) & \(+3.4007257\) & 8 & \(\{+0.4180091,+0.9916370\}\) & \(-1.2000000\) & \(\{-0.5887925\}\)\\
27 & \(-1\) & \(+3.5118161\) & 8 & \(\{+0.4364082,+0.8194147\}\) & \(-0.2000000\) & \(\{+0.2807445\}\)\\
28 & \(-1\) & \(+3.5118161\) & 8 & \(\{+0.4364082,+0.8194147\}\) & \(-0.2000000\) & \(\{-0.2807445\}\)\\
29 & \(-1\) & \(+3.5513332\) & 8 & \(\{+0.4461282,+0.7634606\}\) & \(+0.2000000\) & \(\{+0.3124194\}\)\\
30 & \(-1\) & \(+3.5513332\) & 8 & \(\{+0.4461282,+0.7634606\}\) & \(+0.2000000\) & \(\{-0.3124194\}\)\\
31 & \(+1\) & \(+3.6411081\) & 8 & \(\{+0.4863070,+0.6309969\}\) & \(+1.2000000\) & \(\{+0.5508191\}\)\\
32 & \(+1\) & \(+3.6411081\) & 8 & \(\{+0.4863070,+0.6309969\}\) & \(+1.2000000\) & \(\{-0.5508191\}\)\\
\hline
\end{tabular}
\endgroup
\end{table*}

\begin{table*}[p]
\caption{Complete zero-root branches of the periodic chain
\(h_1=\cdots=h_{12}=1\) at \(N=12\), where \(d_{12}=4\), \(s_{12}=6\),
\(f_{12}=8\) and \(n_{12}=2\).  The fundamental eigenvalue is
\(\Lambda^{[1]}_\chi(u)=\prod_{j=1}^{4}(1-u/\lambda_j)\) and the fused
eigenvalue is \(\Lambda^{[2]}_\chi(u)=\kappa_\chi u^{s}\prod_\alpha(1-u/\gamma_\alpha)\),
the last column listing \(\gamma_\alpha/\ii\).  For \(\chi=+1\),
Eq.~\eqref{eq:supp-twist-kappa} gives \(\kappa_{+1}=\pm2\), \(s=6\), and
two fused roots.  For \(\chi=-1\) it gives \(\kappa_{-1}=0\), so
\(\Lambda^{[2]}_{-1}\) starts at \(u^{7}\), one fused root remains, and
the listed \(\kappa_\chi\) is the coefficient of \(u^{7}\).  The energies
and degeneracies of all \(45\) rows agree with exact diagonalization,
and the degeneracies sum to \(4096=2^{12}\).}
\label{tab:N12-complete-zero-root}
\centering
\begingroup
\setlength{\tabcolsep}{1.7pt}
\renewcommand{\arraystretch}{0.86}
\begin{tabular}{c c c c c c c}
\hline
level & \(\chi\) & \(E\) & degeneracy & \(\{\lambda_j\}\) &
\multicolumn{1}{c}{\(\kappa_\chi\)} & \multicolumn{1}{c}{\(\{\gamma_\alpha/\ii\}\)}\\
\hline
1 & \(+1\) & \(-6.1911471\) & 32 & \(\{-0.3879073,-0.5176381,-0.8593118,-1.9318517\}\) & \(-2.0000000\) & \(\{-0.2839679,-2.3476836\}\)\\
2 & \(+1\) & \(-6.1527560\) & 32 & \(\{-0.4283730,-0.4283730,-1.3477747,-1.3477747\}\) & \(+2.0000000\) & \(\{+0.2853547,-2.3362734\}\)\\
3 & \(-1\) & \(-5.9309391\) & 192 & \(\{-0.3941642,-0.4974415,-0.8174179,-6.2393120\}\) & \(+0.9077234\ii\) & \(\{+0.9077234\}\)\\
4 & \(+1\) & \(-5.6519100\) & 96 & \(\{+6.5597572,-0.3856742,-0.5476586,-0.7217417\}\) & \(-2.0000000\) & \(\{+6.6221502,-0.3020167\}\)\\
5 & \(+1\) & \(-5.5647668\) & 96 & \(\{+5.8994147,-0.4092293,-0.4573467,-0.9056883\}\) & \(+2.0000000\) & \(\{+6.5477730,+0.3054474\}\)\\
6 & \(-1\) & \(-5.0282739\) & 64 & \(\{+1.7482723,-0.3977508,-0.4793565,-1.0000000\}\) & \(+3.2736009\ii\) & \(\{-1.0912003\}\)\\
7 & \(-1\) & \(-4.1722604\) & 64 & \(\{+1.0000000,-0.3977508,-0.4793565,-1.7482723\}\) & \(+4.3118724\ii\) & \(\{-1.4372908\}\)\\
8 & \(+1\) & \(-3.6955181\) & 96 & \(\{+0.9056883,-0.4092293,-0.4573467,-5.8994147\}\) & \(+2.0000000\) & \(\{+5.0497516,+0.3960591\}\)\\
9 & \(+1\) & \(-3.1857247\) & 96 & \(\{+0.7217417,-0.3856742,-0.5476586,-6.5597572\}\) & \(-2.0000000\) & \(\{+4.6868698,-0.4267240\}\)\\
10 & \(+1\) & \(-3.1849009\) & 32 & \(\{+1.3477747,+1.3477746,-0.4283730,-0.4283730\}\) & \(+2.0000000\) & \(\{+0.4430585,-1.5046921\}\)\\
11 & \(-1\) & \(-3.1636619\) & 192 & \(\{+6.2393120,+0.8174179,-0.3941642,-0.4974415\}\) & \(+5.0981608\ii\) & \(\{+5.0981608\}\)\\
12 & \(+1\) & \(-2.8284271\) & 32 & \(\{+1.9318517,+0.8593118,-0.3879073,-0.5176381\}\) & \(-2.0000000\) & \(\{-0.4714045,-1.4142136\}\)\\
13 & \(+1\) & \(-2.3048893\) & 96 & \(\{+0.5476586,-0.3856742,-0.7217417,-6.5597572\}\) & \(-2.0000000\) & \(\{+4.1325899,-0.4839580\}\)\\
14 & \(-1\) & \(-2.0000000\) & 64 & \(\{+0.4793565,-0.3977508,-1.0000000,-1.7482723\}\) & \(+5.6568543\ii\) & \(\{-1.8856181\}\)\\
15 & \(-1\) & \(-1.5898173\) & 192 & \(\{+6.2393120,+0.4974415,-0.3941642,-0.8174179\}\) & \(+5.7855407\ii\) & \(\{+5.7855407\}\)\\
16 & \(+1\) & \(-1.5307337\) & 96 & \(\{+0.4573467,-0.4092293,-0.9056883,-5.8994147\}\) & \(+2.0000000\) & \(\{+3.7494830,+0.5334069\}\)\\
17 & \(+1\) & \(-1.2921676\) & 32 & \(\{+1.9318517,+0.5176381,-0.3879073,-0.8593118\}\) & \(-2.0000000\) & \(\{-0.6290599,-1.0597825\}\)\\
18 & \(-1\) & \(-1.1439865\) & 64 & \(\{+0.3977508,-0.4793565,-1.0000000,-1.7482723\}\) & \(+5.8899315\ii\) & \(\{-1.9633105\}\)\\
19 & \(+1\) & \(-1.0165483\) & 96 & \(\{+0.4092293,-0.4573467,-0.9056883,-5.8994147\}\) & \(+2.0000000\) & \(\{+3.5663535,+0.5607969\}\)\\
20 & \(+1\) & \(-0.7710746\) & 96 & \(\{+0.3856742,-0.5476586,-0.7217417,-6.5597572\}\) & \(-2.0000000\) & \(\{+3.5026445,-0.5709971\}\)\\
21 & \(-1\) & \(-0.5363635\) & 192 & \(\{+6.2393120,+0.3941642,-0.4974415,-0.8174179\}\) & \(+5.9759781\ii\) & \(\{+5.9759781\}\)\\
22 & \(+1\) & \(+0.0000000\) & 32 & \(\{+1.9318517,+0.3879073,-0.5176381,-0.8593118\}\) & \(-2.0000000\) & \(\{-0.8164968,-0.8164964\}\)\\
23 & \(+1\) & \(+0.0000000\) & 128 & \(\{+1.3477747,+0.4283730,-0.4283730,-1.3477747\}\) & \(+2.0000000\) & \(\{+0.8164966,-0.8164966\}\)\\
24 & \(+1\) & \(+0.0000000\) & 32 & \(\{+0.8593118,+0.5176381,-0.3879073,-1.9318517\}\) & \(-2.0000000\) & \(\{+0.8164967,+0.8164965\}\)\\
25 & \(-1\) & \(+0.5363635\) & 192 & \(\{+0.8174179,+0.4974415,-0.3941642,-6.2393120\}\) & \(-5.9759781\ii\) & \(\{-5.9759781\}\)\\
26 & \(+1\) & \(+0.7710746\) & 96 & \(\{+6.5597572,+0.7217417,+0.5476586,-0.3856742\}\) & \(-2.0000000\) & \(\{-3.5026445,+0.5709971\}\)\\
27 & \(+1\) & \(+1.0165483\) & 96 & \(\{+5.8994147,+0.9056883,+0.4573467,-0.4092293\}\) & \(+2.0000000\) & \(\{-3.5663535,-0.5607969\}\)\\
28 & \(-1\) & \(+1.1439865\) & 64 & \(\{+1.7482723,+1.0000000,+0.4793565,-0.3977508\}\) & \(-5.8899315\ii\) & \(\{+1.9633105\}\)\\
29 & \(+1\) & \(+1.2921676\) & 32 & \(\{+0.8593118,+0.3879073,-0.5176381,-1.9318517\}\) & \(-2.0000000\) & \(\{+0.6290599,+1.0597825\}\)\\
30 & \(+1\) & \(+1.5307337\) & 96 & \(\{+5.8994147,+0.9056883,+0.4092293,-0.4573467\}\) & \(+2.0000000\) & \(\{-3.7494830,-0.5334069\}\)\\
31 & \(-1\) & \(+1.5898173\) & 192 & \(\{+0.8174179,+0.3941642,-0.4974415,-6.2393120\}\) & \(-5.7855407\ii\) & \(\{-5.7855407\}\)\\
32 & \(-1\) & \(+2.0000000\) & 64 & \(\{+1.7482723,+1.0000000,+0.3977508,-0.4793565\}\) & \(-5.6568543\ii\) & \(\{+1.8856181\}\)\\
33 & \(+1\) & \(+2.3048893\) & 96 & \(\{+6.5597572,+0.7217417,+0.3856742,-0.5476586\}\) & \(-2.0000000\) & \(\{-4.1325899,+0.4839580\}\)\\
34 & \(+1\) & \(+2.8284271\) & 32 & \(\{+0.5176381,+0.3879073,-0.8593118,-1.9318517\}\) & \(-2.0000000\) & \(\{+0.4714045,+1.4142136\}\)\\
35 & \(-1\) & \(+3.1636619\) & 192 & \(\{+0.4974415,+0.3941642,-0.8174179,-6.2393120\}\) & \(-5.0981608\ii\) & \(\{-5.0981608\}\)\\
36 & \(+1\) & \(+3.1849009\) & 32 & \(\{+0.4283730,+0.4283730,-1.3477746,-1.3477748\}\) & \(+2.0000000\) & \(\{-0.4430585,+1.5046921\}\)\\
37 & \(+1\) & \(+3.1857247\) & 96 & \(\{+6.5597572,+0.5476586,+0.3856742,-0.7217417\}\) & \(-2.0000000\) & \(\{-4.6868698,+0.4267240\}\)\\
38 & \(+1\) & \(+3.6955181\) & 96 & \(\{+5.8994147,+0.4573467,+0.4092293,-0.9056883\}\) & \(+2.0000000\) & \(\{-5.0497516,-0.3960591\}\)\\
39 & \(-1\) & \(+4.1722604\) & 64 & \(\{+1.7482723,+0.4793565,+0.3977508,-1.0000000\}\) & \(-4.3118724\ii\) & \(\{+1.4372908\}\)\\
40 & \(-1\) & \(+5.0282739\) & 64 & \(\{+1.0000000,+0.4793565,+0.3977508,-1.7482723\}\) & \(-3.2736009\ii\) & \(\{+1.0912003\}\)\\
41 & \(+1\) & \(+5.5647668\) & 96 & \(\{+0.9056883,+0.4573467,+0.4092293,-5.8994147\}\) & \(+2.0000000\) & \(\{-6.5477730,-0.3054474\}\)\\
42 & \(+1\) & \(+5.6519100\) & 96 & \(\{+0.7217417,+0.5476586,+0.3856742,-6.5597572\}\) & \(-2.0000000\) & \(\{-6.6221502,+0.3020167\}\)\\
43 & \(-1\) & \(+5.9309391\) & 192 & \(\{+6.2393120,+0.8174179,+0.4974415,+0.3941642\}\) & \(-0.9077234\ii\) & \(\{-0.9077234\}\)\\
44 & \(+1\) & \(+6.1527560\) & 32 & \(\{+1.3477747,+1.3477747,+0.4283730,+0.4283730\}\) & \(+2.0000000\) & \(\{-0.2853547,+2.3362734\}\)\\
45 & \(+1\) & \(+6.1911471\) & 32 & \(\{+1.9318517,+0.8593118,+0.5176381,+0.3879073\}\) & \(-2.0000000\) & \(\{+0.2839679,+2.3476836\}\)\\
\hline
\end{tabular}
\endgroup
\end{table*}

\section{Cubic gap of the FFD chain}
\label{sec:supp-cubic-gap}

\subsection{Exact cubic gap of the one-wrap closure}
\label{subsec:supp-gap-onewrap}

An HFF inversion relation \(\Lambda(u)\Lambda(-u)=P_N(u^2)\) reduces the
spectral problem to the zeros of a recurrence polynomial.  If
\(P_N(u^2)=\prod_k(1-u^2/u_k^2)\), the inverse roots give the
single-particle energies and hence determine the many-body spectrum and
partition function.  Earlier HFF studies extracted these roots for
particular models and coupling patterns
~\cite{Fendley2019,AlcarazPimenta2020}.  Polynomial sequences generated
by cubic denominators are also a standard subject in the theory of
recursive polynomials~\cite{BerahaKahaneWeiss1975,Tran2014,TranZumba2018}.
Their zeros can be obtained by expressing the ratio of two conjugate
characteristic roots as \(e^{2\ii\theta}\), which turns the zero
condition into a trigonometric quantization equation
~\cite{Tran2014,TranZumba2018}.  We use this construction as a systematic
route from an HFF inversion relation to finite-size energies and
thermodynamics.

We demonstrate the method for the one-wrap FFD chain
\(h_1=0\), \(h_2=\cdots=h_N=1\).  The feedback term in
Eq.~\eqref{eq:supp-twist-eigen2} then vanishes and the fused eigenvalue
obeys a scalar inversion relation.  Its zeros label the nested branches,
while the first-level zeros fix the energies within each branch.  We
first take the simplest size class \(N=12L\), for which the root counting
is uniform and all parity phases disappear.  Arbitrary chain lengths
are treated in Sec.~\ref{subsec:supp-gap-alllength}.

Put \(t=u^2\).  At \(h_1=0\) the matrix \(\mathsf S_2(0)\) has a single
nonzero column, which collapses the ordered trace of
Eq.~\eqref{eq:FFD-scalar-functions} onto one diagonal entry.  At the
first level direct multiplication gives a two-term expression,
\begin{equation}
 \phi^{\mathrm w}_N(t)
 =\big[\mathsf S_1(t)^N\big]_{11}-t\big[\mathsf S_1(t)^{N-3}\big]_{11},
 \qquad
 \psi^{\mathrm w}_N(t)=\big[\mathsf S_2(t)^N\big]_{11},
 \label{eq:supp-1w-entries}
\end{equation}
where \(\phi^{\mathrm w}_N=\Phi^{\mathrm F}_{1,N}|_{h_1=0}\) and
\(\psi^{\mathrm w}_N=\Phi^{\mathrm F}_{2,N}|_{h_1=0}\).

Because \(\mathsf S_1(t)\) and \(\mathsf S_2(t)\) are three-dimensional,
the Cayley--Hamilton identity gives third-order recursions in the chain
length,
\begin{equation}
 \phi_N^{\mathrm w}=\phi_{N-1}^{\mathrm w}-t\phi_{N-3}^{\mathrm w},
 \qquad
 \psi_N^{\mathrm w}=t\psi_{N-2}^{\mathrm w}
 +t^2\psi_{N-3}^{\mathrm w}.
 \label{eq:supp-1w-recurrences}
\end{equation}
The corresponding characteristic polynomials are, by definition,
\begin{equation}
 \chi_1(\rho;t)=\det[\rho\mathbf1-\mathsf S_1(t)]
 =\rho^3-\rho^2+t,
 \qquad
 \chi_2(\rho;t)=\det[\rho\mathbf1-\mathsf S_2(t)]
 =\rho^3-t\rho-t^2.
 \label{eq:supp-1w-characteristic}
\end{equation}
Their three roots \(\rho_\ell(t)\) are the eigenvalues of the scalar
matrix.  They generate the dependence on \(N\), since every matrix
element of \(\mathsf S_a(t)^N\) is a linear combination of
\(\rho_\ell(t)^N\).  The physical zeros in the \(t\) plane arise when
these three contributions cancel.

Equivalently, Eq.~\eqref{eq:supp-1w-recurrences} gives the generating
functions
\begin{equation}
 \sum_{N\geq0}\phi^{\mathrm w}_N(t)\xi^N
 =\frac{2-\xi}{1-\xi+t\xi^3},
 \qquad
 \sum_{N\geq0}\psi^{\mathrm w}_N(t)\xi^N
 =\frac{1}{1-t\xi^2-t^2\xi^3}.
 \label{eq:supp-1w-generating}
\end{equation}
We now implement this construction for the fused polynomial
\(\psi_N^{\mathrm w}\).  Expanding the second generating function shows
that the powers of \(t\) present in \(\psi^{\mathrm w}_N\) run from \(t^{N/2}\) to \(t^{2N/3}\) with both
extreme coefficients equal to one, so for \(N=12L\)
\begin{equation}
 \psi^{\mathrm w}_N(t)=t^{6L}\,\widehat\psi_N(t),
 \qquad
 \deg\widehat\psi_N=2L,
 \qquad
 \widehat\psi_N(0)=1,
 \label{eq:supp-1w-psi-reduced}
\end{equation}
in accordance with the degree window of
Eq.~\eqref{eq:supp-twist-degree2}.

Choose the conjugate characteristic roots so that
\(\rho_+(t)/\rho_-(t)=e^{2\ii\theta}\).  In the present normalization,
put
\begin{equation}
 q=2\cos\theta,
 \qquad
 t(\theta)=\frac{q^2}{(q^2-1)^3}
 =-\gamma(\theta)^2,
 \qquad
 \gamma(\theta)=\frac{q}{(1-q^2)^{3/2}},
 \label{eq:supp-1w-angle-map}
\end{equation}
where \(\pi/3<\theta<\pi/2\) and \(\gamma\) decreases strictly from
\(+\infty\) to \(0\).  We also write \(x=\pi/2-\theta\), so that
\(0<x<\pi/6\) and
\(\gamma(x)=2\sin x/(1-4\sin^2x)^{3/2}\) increases from \(0\) to
\(\infty\).  The fused roots of Eq.~\eqref{eq:FFD-fused-roots} are
written as \(\gamma_\alpha=\ii\gamma_{a,N}\), with
\(\gamma_{a,N}>0\).  Partial fractions of the second
generating function give the exact identity
\begin{equation}
 \psi^{\mathrm w}_N(-\gamma^2)
 =\left[\frac{q}{(1-q^2)^2}\right]^{N}
 \Big\{a_0q^N+2(-1)^N\operatorname{Re}\big[a_+e^{\ii N\theta}\big]\Big\},
 \quad
 a_0=\frac{q^2}{1+2q^2},
 \quad
 a_+=\frac{e^{2\ii\theta}}{-2\sin^2\theta+3\ii\sin2\theta}.
 \label{eq:supp-1w-psi-angle}
\end{equation}
Writing \(a_+=|a_+|e^{\ii\delta(\theta)}\) with
\(|a_+|=[2\sin\theta\sqrt{1+2q^2}\,]^{-1}\) and the branch fixed by
\(\delta(\pi/2)=0\), the positive coordinates \(\gamma_{a,N}\) of the
fused roots are exactly the solutions of
\begin{equation}
 \cos\!\big[N\theta+\delta(\theta)\big]
 +(-1)^N\frac{q^{N+2}\sin\theta}{\sqrt{1+2q^2}}=0 .
 \label{eq:supp-1w-quantization}
\end{equation}
The second term has a fixed sign and modulus strictly below one
throughout the interval, and it is exponentially small away from
\(\theta=\pi/3\).

\begin{lemma}[Fused roots]
\label{lem:supp-1w-roots}
For \(N=12L\),
\begin{equation}
 \widehat\psi_N(t)=\prod_{a=1}^{2L}\big(t+\gamma_{a,N}^2\big),
 \qquad
 0<\gamma_{1,N}<\cdots<\gamma_{2L,N},
 \qquad
 \prod_{a=1}^{2L}\gamma_{a,N}=1 .
 \label{eq:supp-1w-root-product}
\end{equation}
Moreover, \(\gamma_{1,N}\) corresponds to the fused root closest to the
origin, and
\begin{equation}
 \gamma_{1,N}=\frac{\pi}{N-1}+O(N^{-3}),
 \qquad
 N\gamma_{a,N}\longrightarrow(2a-1)\pi
 \quad(a\ \text{fixed}).
 \label{eq:supp-1w-first-root}
\end{equation}
\end{lemma}

\begin{proof}
The phase \(N\theta+\delta(\theta)\) is strictly increasing, while the
second term of Eq.~\eqref{eq:supp-1w-quantization} has a fixed sign and
modulus below one.  Each of \(2L\) consecutive half-lobes therefore
contains a root.  Since \(\deg\widehat\psi_N=2L\), these are all the
roots and they are simple.  Their product is one from the constant and
leading coefficients of \(\widehat\psi_N\).  For \(x\to0\),
\(\gamma=2x+O(x^3)\) and
\(\delta=x+O(x^3)\), so for \(N=12L\) the quantization reads
\(\cos[(N-1)x+O(x^3)]=O[(2x)^N]\), whose \(a\)-th positive solution is
\((N-1)x=(2a-1)\pi/2+O(N^{-2})\).
\end{proof}

With \(\prod_jh_j=0\) the second relation of
Eq.~\eqref{eq:supp-twist-eigen2} reduces to
\(\Lambda^{[2]}(u)\Lambda^{[2]}(-u)=\psi^{\mathrm w}_N(u^2)\), whose
polynomial solutions pick one member of each reflected pair
\(\{\pm\ii\gamma_{a,N}\}\) together with one overall sign,
\begin{equation}
 \Lambda^{[2]}_{\boldsymbol\sigma,\kappa}(u)
 =\kappa\,u^{6L}\prod_{a=1}^{2L}
 \left(1-\frac{u}{\ii\sigma_a\gamma_{a,N}}\right),
 \qquad
 \sigma_a=\pm1,
 \qquad
 \kappa^2=1,
 \label{eq:supp-1w-fused-branches}
\end{equation}
the value of \(\kappa^2\) being Eq.~\eqref{eq:supp-twist-kappa} at
\(P_{\mathrm o}=0\) and \(P_{\mathrm e}=1\).  Only the even combination
\(\Lambda^{[2]}(u)+\Lambda^{[2]}(-u)\) enters
the response below, so the global flip
\(\boldsymbol\sigma\mapsto-\boldsymbol\sigma\) is a symmetry.

For the one-wrap chain, let \(\alpha=\mathrm g,1\) label the aligned
branch and the branch that reflects \(\gamma_{1,N}\), and define from the first functional
relation
\begin{equation}
 Q_{N,\alpha}(u)
 :=\Phi^{\mathrm F}_{1,N}(u^2)
 +\Lambda^{[2]}_\alpha(u)+\Lambda^{[2]}_\alpha(-u)
 =\Lambda^{[1]}_\alpha(u)\,\Lambda^{[1]}_\alpha(-u).
 \label{eq:supp-gap-response-Q}
\end{equation}
Its zeros are the selected roots \(\{\lambda_{\alpha,j}\}\) and their
reflections.  Write \(\Delta_N=E_1-E_{\mathrm g}\) for the energy
difference between the two branches, which is the quantity computed in
this subsection.  With \(E_\alpha=\sum_j1/\lambda_{\alpha,j}\), a contour
\(\Gamma_u\) enclosing \(\{\lambda_{\alpha,j}\}\) but not the origin
gives
\begin{equation}
 \Delta_N=E_1-E_{\mathrm g}
 =\frac{1}{2\pi\ii}\oint_{\Gamma_u}\frac{\dd u}{u}\,
 \partial_u\log
 \frac{Q_{N,1}(u)}{Q_{N,\mathrm g}(u)} .
 \label{eq:supp-gap-exact-contour}
\end{equation}
Deforming \(\Gamma_u\) onto the imaginary axis gives the form used
below.

Figure~\ref{fig:supp-onewrap-contour} illustrates the contour for a
representative one-wrap spectrum.  The box surrounds the selected
\(\lambda_{\alpha,j}\) roots and excludes the origin and their reflected
partners.

\begin{figure}[!t]
\centering
\includegraphics[width=0.44\columnwidth]{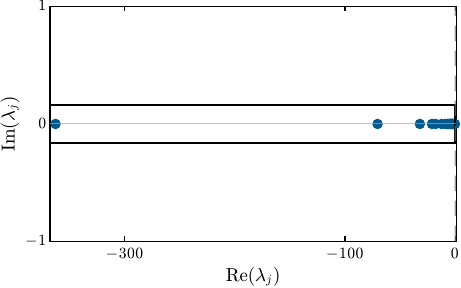}
\caption{Representative one-wrap root configuration at \(N=300\).
Filled points denote the selected \(\lambda_j\) roots, while open points
denote their reflected partners \(-\lambda_j\).  The rectangle is the
contour \(\Gamma_u\): it encloses all selected roots and excludes both
the origin and the reflected roots.}
\label{fig:supp-onewrap-contour}
\end{figure}

The branches are ordered by an exact factorization.  Label a branch by
the set \(A\) of reflected fused roots, \(\sigma_a=-1\) for \(a\in A\), and
by a further sign \(\xi=\pm1\), the two being tied to \(\kappa\) through
\begin{equation}
 \kappa_{\xi,A}=\xi\,(-1)^{3L+|A|} .
 \label{eq:supp-1w-kappa-rule}
\end{equation}
On the imaginary axis \(u=\ii y\) with \(y>0\), write
\(r_a=(y-\gamma_{a,N})/(y+\gamma_{a,N})\in(-1,1)\),
\(r_A=\prod_{a\in A}r_a\), and let \(A^{\rm c}\) be the complementary
set.  Collecting the two products in
Eq.~\eqref{eq:supp-1w-fused-branches} gives
\begin{equation}
 \frac{Q_{N,\xi,A}(\ii y)}{\phi^{\mathrm w}_N(-y^2)}
 =1+\xi\,\mathcal A_N(y)\,\big(r_A+r_{A^{\rm c}}\big),
 \qquad
 \mathcal A_N(y)=\frac{y^{6L}\prod_a\big(1+y/\gamma_{a,N}\big)}
 {\phi^{\mathrm w}_N(-y^2)} ,
 \label{eq:supp-1w-branch-value}
\end{equation}
and therefore, with \(\mathrm g\) the branch \(\xi=+1\),
\(A=\emptyset\),
\begin{equation}
 \frac{Q_{N,\mathrm g}(\ii y)-Q_{N,\xi,A}(\ii y)}
      {\phi^{\mathrm w}_N(-y^2)}
 =\mathcal A_N(y)\,(1-\xi r_A)(1-\xi r_{A^{\rm c}})\geq0 .
 \label{eq:supp-1w-branch-order}
\end{equation}
Every coefficient of \(\phi^{\mathrm w}_N(-y^2)\) is positive and every
factor of the numerator of \(\mathcal A_N\) is positive, so
\(\mathcal A_N>0\), and both factors in
Eq.~\eqref{eq:supp-1w-branch-order} are nonnegative because
\(|r_A|<1\).  The branch \(\mathrm g\) itself obeys
\(Q_{N,\mathrm g}/\phi^{\mathrm w}_N=1+\mathcal A_N(1+\prod_ar_a)>1\),
so \(Q_{N,\mathrm g}(\ii y)>\phi^{\mathrm w}_N(-y^2)>0\), while every
other branch has
\(Q_{N,\xi,A}(\ii y)=|\Lambda^{[1]}_{\xi,A}(\ii y)|^2\geq0\) because
\(\Lambda^{[1]}\) has real coefficients.  Hence \(\mathrm g\) is the
lowest branch, and it is the aligned configuration
\(\sigma_1=\cdots=\sigma_{2L}\).  The branch obtained by reflecting the
root closest to the origin has \(\xi=+1\) and \(A=\{1\}\).  Both
selected branches have their \(\lambda\) roots on
the negative real axis, so neither \(Q\) vanishes on the imaginary axis.

\begin{proposition}[Exact finite-size gap]
\label{prop:supp-1w-gap}
\begin{equation}
 \Delta_N
 =-\frac1\pi\int_0^\infty\frac{\dd y}{y^2}\,
 \log\frac{Q_{N,1}(\ii y)}{Q_{N,\mathrm g}(\ii y)} .
 \label{eq:supp-1w-exact-gap}
\end{equation}
\end{proposition}

\begin{proof}
Neither of the two branches has a zero on the imaginary axis, so
\(Q_{N,\alpha}(\ii y)>0\) there and the selected factor
\(\Lambda^{[1]}_\alpha\) can be taken with all its zeros in the left
half plane.  Writing
\(Q_{N,\alpha}(\ii y)=\prod_j(1+y^2/\lambda_{\alpha,j}^2)\) and using
\(\pi^{-1}\int_0^\infty y^{-2}\log(1+y^2/\lambda^2)\dd y
=-1/\lambda\) for \(\operatorname{Re}\lambda<0\), the energy
\(E_\alpha=\sum_j\lambda_{\alpha,j}^{-1}\) equals
\(-\pi^{-1}\int_0^\infty y^{-2}\log Q_{N,\alpha}(\ii y)\dd y\).
Subtract the two branches.
\end{proof}

Introduce the two ratios
\begin{equation}
 \mathcal R_N(y)=\frac{Q_{N,\mathrm g}(\ii y)}{\phi^{\mathrm w}_N(-y^2)}-1,
 \qquad
 \mathcal S_N(y)=\frac{y}{2\gamma_{1,N}}\,
 \frac{Q_{N,\mathrm g}(\ii y)-Q_{N,1}(\ii y)}
      {Q_{N,\mathrm g}(\ii y)-\phi^{\mathrm w}_N(-y^2)} ,
 \label{eq:supp-1w-RS}
\end{equation}
both positive on the axis.  Expanding the logarithm in
Eq.~\eqref{eq:supp-1w-exact-gap} to first order in the difference
\(Q_{N,\mathrm g}-Q_{N,1}\) turns it into
\begin{equation}
 \Delta_N=\frac{2\gamma_{1,N}}{\pi}\,\mathrm I_N+\delta_N,
 \qquad
 \mathrm I_N=\int_0^\infty\frac{\dd y}{y^3}\,
 \mathcal S_N(y)\,\frac{\mathcal R_N(y)}{1+\mathcal R_N(y)} .
 \label{eq:supp-gap-endpoint-response}
\end{equation}
The neglected part is fixed by the relative difference
\(d_N=1-Q_{N,1}(\ii y)/Q_{N,\mathrm g}(\ii y)\), through the factor
\(-\log(1-d_N)/d_N=1+d_N/2+\cdots\).  At the scale \(y=N/p\) with \(p\)
of order one, which is the scale that dominates the integral, the
explicit product below gives \(d_N=O(N^{-2})\), so
\(\delta_N=O(N^{-5})\).  The same two functions govern the periodic
chain treated later.

Comparing the two branches factor by factor gives the closed form
\begin{equation}
 \mathcal S_N(N/p)
 =\frac{1+X_N(p)}{1-X_N(p)+(p\gamma_{1,N}/N)[1+X_N(p)]},
 \qquad
 X_N(p)=\prod_{a\neq1}
 \frac{1-N/(p\gamma_{a,N})}{1+N/(p\gamma_{a,N})}
 =-\prod_{a\neq1}r_a .
 \label{eq:supp-1w-S-product}
\end{equation}
For \(\theta=\pi/3+\eta\),
\(\gamma(\theta)=(2\sqrt3\,\eta)^{-3/2}[1+O(\eta)]\),
so the window in which \(\gamma_{a,N}/N\) stays between two fixed
positive numbers contains of order \(N^{1/3}\) roots, each
contributing a factor of modulus bounded by a common constant below
one, while every other factor has modulus at most one.  Hence
\(|X_N(p)|\leq s^{\,cN^{1/3}}\) with \(s<1\), uniformly on every
fixed interval \(0<p_-\leq p\leq p_+<\infty\), and since
\(p\gamma_{1,N}/N=O(N^{-2})\),
\begin{equation}
 \mathcal S_N(N/p)\longrightarrow1
 \qquad\text{uniformly on }[p_-,p_+].
 \label{eq:supp-1w-S-limit}
\end{equation}

For \(\mathcal R_N\) let \(q_N(p)\in(0,1)\) solve \(\gamma=N/p\), so
that \(1-q_N^2=(p/N)^{2/3}[1+o(1)]\).  The first generating function in
Eq.~\eqref{eq:supp-1w-generating} gives
\(\phi^{\mathrm w}_N(t)=\sum_\rho\frac{2\rho-1}{3\rho-2}\rho^N\) over
the three roots of \(z^3-z^2+t\), which at \(t=-\gamma^2\) are
\(\rho_0=(1-q^2)^{-1}\) and
\(\rho_\pm=-q(1-q^2)^{-1}e^{\pm\ii\theta}\).  Therefore
\begin{equation}
 \Big(\frac pN\Big)^{2N/3}
 \phi^{\mathrm w}_N\!\Big(-\frac{N^2}{p^2}\Big)
 =q_N^{-2N/3}\left[\frac{1+q_N^2}{1+2q_N^2}+O(q_N^N)\right]
 =\frac23\,q_N^{-2N/3}\,[1+o(1)] .
 \label{eq:supp-1w-phi-hard}
\end{equation}
The product over \(\gamma_{a,N}\) supplies the remaining constant.

\begin{lemma}[Asymptotic product of the fused roots]
\label{lem:supp-1w-hard}
Uniformly on every fixed interval \(0<p_-\leq p\leq p_+<\infty\),
\begin{equation}
 q_N(p)^{2N/3}\prod_{a=1}^{2L}
 \Big(1+\frac{p\,\gamma_{a,N}}{N}\Big)
 \longrightarrow
 \frac{1}{\sqrt3}\exp\!\Big(-\frac{2p}{3\pi}\Big).
 \label{eq:supp-1w-hard-product}
\end{equation}
\end{lemma}

\begin{proof}
In the variable \(x\) of Eq.~\eqref{eq:supp-1w-angle-map},
Eq.~\eqref{eq:supp-1w-quantization} reads \(\cos\Psi_N(x)=-R_N(x)\)
with \(\Psi_N(x)=Nx-\varphi(x)\),
\(\varphi(x)=\arctan(3\tan x)-2x\) and
\(R_N(x)=(2\sin x)^{N+2}\cos x/\sqrt{1+8\sin^2x}\leq\tfrac12\).  Here
\(\varphi(0)=\varphi(\pi/6)=0\) and \(|\varphi'|\leq1\), so \(\Psi_N\)
increases from \(0\) to \(N\pi/6=2L\pi\) and the roots obey the
half-integer quantization
\begin{equation}
 \Psi_N(x_a)=\big(a-\tfrac12\big)\pi+\epsilon_a,
 \qquad a=1,\ldots,2L,
 \qquad |\epsilon_a|\leq2R_N(x_a).
 \label{eq:supp-1w-lattice}
\end{equation}
The displacements \(\epsilon_a\) are exponentially small outside the
region \(\pi/6-x=O(N^{-2/3})\) and contribute \(o(1)\) below.  Using
\(\prod_a\gamma_{a,N}=1\) and \(y=N/p\),
\begin{equation}
 \log\prod_a\Big(1+\frac{p\gamma_{a,N}}{N}\Big)
 =\sum_{a=1}^{2L}F_y(x_a)-\frac N6\log y,
 \qquad
 F_y(x)=\log\Big(1+\frac{y}{\gamma(x)}\Big).
 \label{eq:supp-1w-sum-form}
\end{equation}
Split \(F_y=-\log x+g_y\).  The remainder \(g_y\) is smooth with
\(\int_0^{\pi/6}|g_y''|\leq C(1+y^{2/3})\), obtained by treating the
layer \(\pi/6-x\lesssim y^{-2/3}\) separately, so midpoint summation
over the lattice \eqref{eq:supp-1w-lattice} replaces
\(\sum_ag_y(x_a)\) by
\(N\pi^{-1}\int g_y\dd x-\pi^{-1}\int g_y\varphi'\dd x\) up to
\(O[(1+y^{2/3})/N]\).  For the singular part
\(x_a=[(a-\tfrac12)\pi+O(1)]/N\) at small \(x\), and
\(\sum_{a=1}^n\log(a-\tfrac12)
=\log\Gamma(n+\tfrac12)-\log\Gamma(\tfrac12)\), so Stirling's formula
makes the midpoint sum of \(-\log x\) exceed its integral by exactly
\(-\tfrac12\log2\).  The two correction integrals combine, after one
integration by parts, into \(-I_\gamma(y)\) with
\begin{equation}
 I_\gamma(y)=\frac1\pi\int_0^{\pi/6}\varphi(x)\,
 \frac{y}{\gamma(x)+y}\,\dd\log\gamma(x),
 \qquad
 I_\gamma(\infty)=\tfrac12\log\tfrac32 .
 \label{eq:supp-1w-phase-integral}
\end{equation}
Finally the three terms proportional to \(N\), namely
\((2N/3)\log q_N\), \(N\pi^{-1}\int_0^{\pi/6}F_y\dd x\) and
\(-(N/6)\log y\), combine into \(-2p/(3\pi)+o(1)\).  Collecting the
constants, \(-\tfrac12\log2-\tfrac12\log\tfrac32=-\tfrac12\log3\),
which is Eq.~\eqref{eq:supp-1w-hard-product}.
\end{proof}

Combining Eqs.~\eqref{eq:supp-1w-phi-hard} and
\eqref{eq:supp-1w-hard-product} with the aligned orientation of
Eq.~\eqref{eq:supp-1w-fused-branches},
\begin{equation}
 \mathcal R_N(N/p)\longrightarrow
 \frac{\sqrt3}{2}\exp\!\Big(-\frac{2p}{3\pi}\Big),
 \label{eq:supp-1w-R-limit}
\end{equation}
the amplitude \(\sqrt3/2\) being \(1/\sqrt3\) from the large-\(\gamma\)
root product divided by \(2/3\) from the fundamental scalar source.

Changing variables \(y=N/p\) in
Eq.~\eqref{eq:supp-gap-endpoint-response} gives
\begin{equation}
 N^3\Delta_N
 =\frac{2N\gamma_{1,N}}{\pi}\int_0^\infty p\,
 \mathcal S_N(N/p)\,
 \frac{\mathcal R_N(N/p)}{1+\mathcal R_N(N/p)}\,\dd p
 +N^3\delta_N .
 \label{eq:supp-1w-scaled-gap}
\end{equation}
Lemma~\ref{lem:supp-1w-hard} is uniform only on fixed intervals, so the
majorant has to be established separately, and the integrand behaves
differently in three ranges of \(p\).  For \(p\leq1\) it is at most
\(Cp\), because \(0<\mathcal R_N/(1+\mathcal R_N)<1\) and
Eq.~\eqref{eq:supp-1w-S-product} bounds \(\mathcal S_N\) by a constant.
For \(1\leq p\leq N\), keeping the last cell near \(x=\pi/6\) in
Eq.~\eqref{eq:supp-1w-hard-product} unexpanded instead of passing to
the limit gives \(\mathcal S_N\mathcal R_N/(1+\mathcal R_N)\leq
Ce^{-cp}\) with \(C,c>0\) independent of \(N\) and \(p\), so the
integrand is at most \(Cpe^{-cp}\).  Beyond \(p=N\), that is for
\(y\leq1\), the exponential form is no longer the right one.  There
\(\mathcal R_N\) is simply a polynomial vanishing at the origin like
\(y^{6L}\) and \(\mathcal S_N(N/p)=O(N/p)\), so the integrand decays
like the power \(p^{-6L}\), and the whole range contributes
\(O\big(N^3\mathcal R_N(1)/L\big)\) to Eq.~\eqref{eq:supp-1w-scaled-gap},
which vanishes because \(\mathcal R_N(1)\) is of order
\(e^{-2N/(3\pi)}\).  The three pieces together give an integrable
majorant independent of \(N\).  Dominated convergence together with
Eqs.~\eqref{eq:supp-1w-first-root}, \eqref{eq:supp-1w-S-limit} and
\eqref{eq:supp-1w-R-limit} then proves the following.

\begin{theorem}[Cubic gap of the one-wrap chain]
\label{thm:supp-1w-gap}
For \(h_1=0\), \(h_2=\cdots=h_N=1\) and \(N=12L\), the energy
difference between the aligned branch and the branch that reflects the
fused root \(\ii\gamma_{1,N}\) closest to the origin satisfies
\begin{equation}
 \lim_{N\to\infty}N^3\Delta_N
 =2\int_0^\infty p\,
 \frac{(\sqrt3/2)e^{-2p/(3\pi)}}{1+(\sqrt3/2)e^{-2p/(3\pi)}}\,\dd p
 =\frac{9\pi^2}{2}
 \left[-\operatorname{Li}_2\!\Big(-\frac{\sqrt3}{2}\Big)\right]
 =32.3245029925\ldots ,
 \label{eq:supp-1w-amplitude}
\end{equation}
where \(\operatorname{Li}_2(z)=\sum_{n\geq1}z^n/n^2\) is the
dilogarithm.
\end{theorem}

\begin{proof}
Use \(\int_0^\infty p\,ae^{-bp}(1+ae^{-bp})^{-1}\dd p
=-\operatorname{Li}_2(-a)/b^2\) with \(a=\sqrt3/2\) and
\(b=2/(3\pi)\), together with \(N\gamma_{1,N}\to\pi\).
\end{proof}

The cubic power combines \(\gamma_{1,N}\simeq\pi/N\) with the
\(N^{-2}\) measure in Eq.~\eqref{eq:supp-1w-scaled-gap}.  Excitations
within a fixed branch retain the open-chain dispersion
\(\varepsilon(\vartheta)=(4\cos^2\vartheta-1)^{3/2}/(2\cos\vartheta)\)
for \(0<\vartheta<\pi/3\), so their lowest scale is \(N^{-3/2}\).
Reflecting the \(a\)-th fused root instead gives
\((2a-1)\Delta_N\) at leading order.

Direct evaluation of Eq.~\eqref{eq:supp-1w-exact-gap} from the
quantized roots gives
\begin{center}
\begin{tabular}{lcccc}
\hline
\(N\) & \(36\) & \(48\) & \(60\) & \(120\)\\
\hline
\(N^3\Delta_N\) & \(34.4109659\) & \(33.9698623\) & \(33.7232567\)
& \(33.2324842\)\\
\hline
\end{tabular}
\end{center}
consistent with Eq.~\eqref{eq:supp-1w-amplitude} up to a correction of
order \(N^{-2/3}\) from the region \(\theta\to\pi/3\).  Equation
\eqref{eq:supp-1w-branch-order} establishes that the aligned branch is
the ground state.  Complete branch enumeration at \(N=36\) and exact
diagonalization at accessible sizes identify the first root-reflection
branch with the first distinct excited level.

\subsection{Arbitrary chain length}
\label{subsec:supp-gap-alllength}

The choice \(N=12L\) fixes the simplest root counting and phase
conventions.  For general \(N\), the limiting function
\(\mathcal R_N\) remains the same, while the number of fused roots and
the factor \(N\gamma_{1,N}\) depend on the length class.  The result is
\begin{equation}
 N^3\Delta_N\longrightarrow
 \begin{cases}
 \displaystyle
 \frac{9\pi^2}{2}
 \left[-\operatorname{Li}_2\!\Big(-\frac{\sqrt3}{2}\Big)\right],
 &N\ \text{even},\\[6pt]
 \displaystyle
 9\pi^2
 \left[-\operatorname{Li}_2\!\Big(-\frac{\sqrt3}{2}\Big)\right],
 &N\ \text{odd},
 \end{cases}
 \label{eq:supp-al-result}
\end{equation}
The extension changes the root
counting, the quantization phase, and the large-\(\gamma\) normalization.

\emph{Counting.}  For general \(N\) the fused eigenvalue reads
\begin{equation}
\begin{aligned}
 \Lambda^{[2]}_{\boldsymbol\sigma,\kappa}(u)
 &={}\ii^{\,s_N}\kappa\sqrt{c_N}\,u^{s_N}
 \prod_{a=1}^{\nu_N}
 \left(1+\frac{\ii\sigma_au}{\gamma_{a,N}}\right),
 \qquad
 \nu_N=f_N-s_N,
 \\
 c_N&=
 \begin{cases}
 1,&N\ \text{even},\\
 (N-1)/2,&N\ \text{odd},
 \end{cases}
 \qquad
 \sigma_a=\pm1,
 \qquad
 \kappa=\pm1.
\end{aligned}
 \label{eq:supp-al-branches}
\end{equation}
with \(s_N\) and \(f_N\) as in Eq.~\eqref{eq:supp-twist-degree2}, and the
coordinates \(\gamma_{a,N}\) are obtained from the positive roots in
\(z=\gamma^2\) of
\begin{equation}
 \sum_{j\geq0}(-z)^j\binom{m-j}{2j}\quad(N=2m),
 \qquad
 \sum_{j\geq0}(-z)^j\binom{m-j}{2j+1}\quad(N=2m+1).
 \label{eq:supp-al-height-polynomials}
\end{equation}
The normalization \(\prod_a\gamma_{a,N}=1\) of
Eq.~\eqref{eq:supp-1w-root-product} is special to even multiples of
three.  For the other length classes, the leading fused coefficient
and the displacement of \(f_N\) from \(2N/3\) enter through the single
combination
\begin{equation}
 \frac{\sqrt{c_N}}{\prod_a\gamma_{a,N}}
 \Big(\frac Np\Big)^{f_N-2N/3}
 q_N^{2N/3}\prod_{a}\Big(1+\frac{p\gamma_{a,N}}{N}\Big),
 \label{eq:supp-al-combination}
\end{equation}
which contains no residual power of \(N\).  Writing \(N=6M+\nu\), the
root counting is
\begin{equation}
\begin{array}{c|ccc}
\nu & s_N & \nu_N & c_N\\ \hline
0 & 3M & M & 1\\
1 & 3M+1 & M-1 & 3M\\
2 & 3M+1 & M & 1\\
3 & 3M+2 & M & 3M+1\\
4 & 3M+2 & M & 1\\
5 & 3M+3 & M & 3M+2
\end{array}
 \label{eq:supp-al-table}
\end{equation}

\emph{Quantization.}  In the variable \(x\) of
Eq.~\eqref{eq:supp-1w-angle-map} the phase \(\delta\) of
Eq.~\eqref{eq:supp-1w-psi-angle} is
\(\varphi(x)=\arctan(3\tan x)-2x\), and
Eq.~\eqref{eq:supp-1w-quantization} becomes, exactly,
\begin{equation}
 \cos\!\big[Nx-\varphi(x)\big]+(-1)^{N/2}R_N(x)=0
 \quad(N\ \text{even}),
 \qquad
 \sin\!\big[Nx-\varphi(x)\big]-(-1)^{(N-1)/2}R_N(x)=0
 \quad(N\ \text{odd}),
 \label{eq:supp-al-quantization}
\end{equation}
with
\(R_N(x)=(2\sin x)^{N+2}\cos x/\sqrt{1+8\sin^2x}\leq\tfrac12\).  At
small \(x\) the roots therefore obey half-integer quantization for
even \(N\) and integer quantization for odd \(N\),
\begin{equation}
 N\gamma_{a,N}\longrightarrow
 \begin{cases}
 (2a-1)\pi,& N\ \text{even},\\
 2a\pi,& N\ \text{odd},
 \end{cases}
 \label{eq:supp-al-ladder}
\end{equation}
so that \(N\gamma_{1,N}\to\pi\) for even \(N\) and \(2\pi\) for odd
\(N\).

\emph{The region \(\theta\to\pi/3\).}  The only limit that could depend
on the chain-length class is the constant contribution from this
region in Lemma~\ref{lem:supp-1w-hard}.  It is carried by the sum of
\(F_y(x_a)\) in Eq.~\eqref{eq:supp-1w-sum-form}, reached from the product
over \(\gamma_{a,N}\) by the
normalization above, and its splitting \(F_y=-\log x+g_y\).  The bulk of
the sum and the contribution of \(g_y\) do not depend on the parity of
the quantization condition, but the logarithmic part does.  Summing
\(-\log x\) over the half-integer sequence gives \(-\tfrac12\log2\),
whereas over the integer lattice it leaves \(-\tfrac12\log(Ny)\) with
\(y=N/p\).  For odd \(N\) the extra factor \(\sqrt{c_N}\) in
Eq.~\eqref{eq:supp-al-branches} contributes
\(+\tfrac12\log[(N-1)y/2]\), which cancels the \(\log N\) and restores
\(-\tfrac12\log2\).  The two parities therefore give the same constant,
Eq.~\eqref{eq:supp-1w-hard-product} holds unchanged, and since
Eq.~\eqref{eq:supp-1w-phi-hard} is controlled by the dominant root
\(\rho_0\) for any \(N\),
\begin{equation}
 \mathcal R_N(N/p)\longrightarrow
 \frac{\sqrt3}{2}\exp\!\Big(-\frac{2p}{3\pi}\Big)
 \label{eq:supp-al-R-limit}
\end{equation}
at every chain length.  The limiting function \(\mathcal R_N\) is
therefore the same for both parities.  The remaining difference in
Eq.~\eqref{eq:supp-1w-scaled-gap} is the prefactor
\(N\gamma_{1,N}\), which gives Eq.~\eqref{eq:supp-al-result}.
A representative value from each of the six length classes is
\begin{center}
\begin{tabular}{c c c c c c c}
\hline
\(N\bmod6\) & \(0\) & \(1\) & \(2\) & \(3\) & \(4\) & \(5\)\\
\(N\) & \(12000\) & \(12001\) & \(12002\) & \(12003\) & \(12004\) & \(12005\)\\
\(N^3\Delta_N\) & \(32.3991\) & \(64.4444\) & \(32.3808\) & \(64.7982\) & \(32.2222\) & \(64.7617\)\\
\hline
\end{tabular}
\end{center}
The even and odd sequences approach \(32.3245\ldots\) and
\(64.6490\ldots\), respectively.

\subsection{Periodic chain: the response with the inter-level coupling frozen}
\label{subsec:supp-gap-periodic-scalar}

For the periodic chain \(h_1=\cdots=h_N=1\) with \(N=12L\), exact
diagonalization at small sizes identifies the ground-state and
lowest-excitation root configurations, whose continuation to \(N=600\)
is shown in Figs.~\ref{fig:main}(c) and (d).  As in the one-wrap chain,
the lowest excitation is associated with the reflection of the fused
root closest to the origin and a reversal of \(\kappa\).  Under periodic
closure, both root sets readjust self-consistently.  The fused roots
remain on the imaginary axis, while the \(\lambda_j\) form
complex-conjugate configurations with imaginary parts of order
\(N^{-1}\) in the region where \(\lambda_j=O(1)\).

We first consider the scalar inversion problem obtained by omitting
the inter-level coupling term in the second functional relation, and then compare
it with the fully coupled periodic solution.

Keep \(N=12L\) and the sector \(\chi=+1\), and write
\(\phi_N=\Phi^{\mathrm F}_{1,N}\) and \(\psi_N=\Phi^{\mathrm F}_{2,N}\).
In the angle of Eq.~\eqref{eq:supp-1w-angle-map} the two traces are
\begin{equation}
 \phi_N=\frac{(-1)^N+2q^N\cos N\theta}{(q^2-1)^N},
 \qquad
 \psi_N=\frac{q^{2N}+2(-1)^Nq^N\cos N\theta}{(q^2-1)^{2N}},
 \label{eq:supp-per-traces}
\end{equation}
so after the fixed zero at the origin is removed the positive root
coordinates \(\gamma_{a,N}\) are the \(2L\) solutions of
\begin{equation}
 q^N+2\cos(N\theta)=0,
 \qquad \frac\pi3<\theta<\frac\pi2 ,
 \label{eq:supp-per-quantization}
\end{equation}
one per half-lobe exactly as before, with
\(N\gamma_{1,N}=\pi+O(N^{-2})\).  Two constants differ from the
one-wrap chain.  Their product is
\begin{equation}
 \prod_{a=1}^{2L}\gamma_{a,N}=\sqrt{\frac23},
 \label{eq:supp-per-normalization}
\end{equation}
instead of one, and the fused prefactor is \(\kappa=2\) instead of
\(\kappa=\pm1\), by Eq.~\eqref{eq:supp-twist-kappa} at
\(P_{\mathrm o}=P_{\mathrm e}=1\).  On the other side of the ratio, the
trace \(\phi_N\) is dominated by the same characteristic root
\(\rho_0\), but with amplitude \(A_\phi=1\) rather than the \(2/3\) of
Eq.~\eqref{eq:supp-1w-phi-hard}.

The large-\(\gamma\) product asymptotic is unchanged.  The periodic
quantization condition carries no phase shift, so the term \(-I_\gamma\) of
Lemma~\ref{lem:supp-1w-hard} is absent, and it is compensated exactly
by the product entering through
\(\log\prod_a\gamma_{a,N}\), leaving
\begin{equation}
 q_N(p)^{2N/3}\prod_{a=1}^{2L}
 \Big(1+\frac{p\,\gamma_{a,N}}{N}\Big)
 \longrightarrow\frac{1}{\sqrt3}\exp\!\Big(-\frac{2p}{3\pi}\Big)
 \label{eq:supp-per-hard-product}
\end{equation}
as before.  Nothing else in Sec.~\ref{subsec:supp-gap-onewrap} is used
beyond the root quantization and this dominant root, so the
calculation carries over and the closure enters only through the three
constants just listed,
\begin{equation}
 \mathcal R_N(N/p)\longrightarrow
 \frac{\kappa}{A_\phi\prod_a\gamma_{a,N}}\,
 \frac{1}{\sqrt3}\exp\!\Big(-\frac{2p}{3\pi}\Big).
 \label{eq:supp-per-R-general}
\end{equation}
For the one-wrap chain this returns the \(\sqrt3/2\) of
Eq.~\eqref{eq:supp-1w-R-limit}.  For the periodic chain it gives
\(\mathcal R^{(0)}_N(N/p)\to\sqrt2\,e^{-2p/(3\pi)}\), and the integral of
Theorem~\ref{thm:supp-1w-gap} gives
\begin{equation}
 N^3\Delta^{(0)}_N\longrightarrow
 D^{(0)}=\frac{9\pi^2}{2}
 \left[-\operatorname{Li}_2\!\big(-\sqrt2\big)\right]
 =48.6081007479\ldots
 \label{eq:supp-per-bare}
\end{equation}
This is the response the periodic chain would have if the inter-level coupling term
of Eq.~\eqref{eq:supp-twist-eigen2} were dropped.  We call it the bare
response.  The next subsection shows, numerically up to \(N=600\), that restoring
the inter-level coupling changes the amplitude but not the power of
\(N\).

\subsection{Dressing by the inter-level coupling}
\label{subsec:supp-gap-dressing}

Write the two selected eigenvalues in terms of their even and odd parts,
\(t=u^2\),
\begin{equation}
 \Lambda^{[1]}(u)=U(t)+uV(t),
 \qquad
 \Lambda^{[2]}(u)=u^{6L}\big[C(t)+\ii uG(t)\big],
 \label{eq:supp-per-UVCG}
\end{equation}
so that Eqs.~\eqref{eq:supp-twist-eigen1} and
\eqref{eq:supp-twist-eigen2} become two real polynomial identities,
\begin{align}
 U^2-tV^2&=\phi_N+2t^{3L}C,
 \label{eq:supp-per-first}\\
 C^2+tG^2&=b_N+2U,
 \qquad \psi_N=t^{6L}b_N .
 \label{eq:supp-per-second}
\end{align}
The inter-level coupling is the single term \(2U\) on the right of
Eq.~\eqref{eq:supp-per-second}.  Without it the two relations are solved
one after the other.  With it they form the self-consistent loop
\begin{equation}
 C\ \longmapsto\ (U,V)\ \longmapsto\ (C,G),
 \label{eq:supp-per-loop}
\end{equation}
and the periodic eigenvalues are fixed points of one full turn.

The scalar problem of Sec.~\ref{subsec:supp-gap-periodic-scalar}
already contains the two scales behind the cubic response.  Its fused
root closest to the origin has \(\gamma_{1,N}=O(N^{-1})\), while the contour
measure at the scale \(y\sim N\) contributes a further \(N^{-2}\).
Restoring the inter-level coupling replaces the bare ratio
\(\mathcal R_N^{(0)}\) by the self-consistent ratio \(\mathcal R_N\).
It can change the amplitude of the response.  A different exponent would
require the inter-level coupling to generate an additional power of \(N\).

We solve the exact coupled identities~\eqref{eq:supp-per-first} and
\eqref{eq:supp-per-second} for every integer \(16\leq N\leq600\),
following the aligned ground branch and the branch obtained by reflecting
the fused root closest to the origin.  Exact diagonalization for
\(N\leq36\) identifies this root-reflection branch with the first distinct
excited level at these sizes.  The finite-size comparison below shows
that the inter-level coupling changes the amplitude without introducing a new
finite-size scale.

For the periodic chain the same contour construction must enclose the
corresponding complex-conjugate \(\lambda_j\) roots.  In this case the
contour cannot be chosen as a narrow interval around the real axis.

\emph{Periodic-gap comparison.}
Figure~\ref{fig:supp-periodic-gap-parity} shows \(N^3\Delta_N\) for the
periodic FFD chain with \(60\leq N\leq600\).
In panel (a), both parities form three smooth curves indexed by
\(N\bmod3\), supporting the cubic scaling \(\Delta_N\sim N^{-3}\).
Panel (b) compares \(N^3\Delta_N/2\) for odd lengths with
\(N^3\Delta_N\) for even lengths.
Their increasingly close agreement at large \(N\) supports a factor of
two between the odd- and even-length cubic-gap amplitudes.

\begin{figure}[!htbp]
\centering
\includegraphics[width=\textwidth]{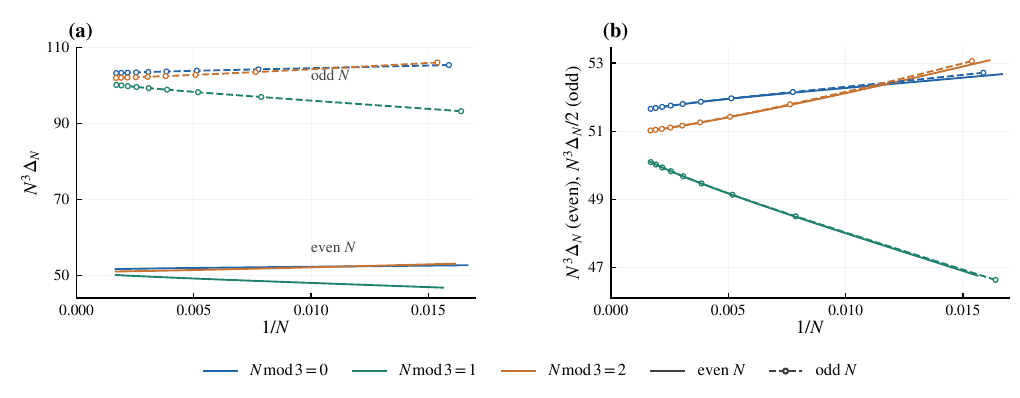}
\caption{Periodic FFD gap for \(60\leq N\leq600\).
(a) \(N^3\Delta_N\) versus \(1/N\).
(b) Comparison of \(N^3\Delta_N/2\) for odd lengths and
\(N^3\Delta_N\) for even lengths.
Colors distinguish \(N\bmod3=0,1,2\).
Solid and dashed lines denote even and odd lengths, respectively.}
\label{fig:supp-periodic-gap-parity}
\end{figure}

\clearpage

\subsection{The \texorpdfstring{$z=3/2$}{z=3/2} branch in the two closures}
\label{subsec:supp-gap-freebranch}

Both one-wrap and periodic closures preserve the \(z=3/2\) scaling
of the open FFD chain~\cite{Fendley2019}.
For the one-wrap chain, the spectrum remains free fermionic
within each fused branch. In the periodic chain the feedback removes that freedom, and
what survives is a smaller family in which the fundamental eigenvalue is
a perfect square.

\emph{One-wrap chain.}  Once the fused branch is fixed,
\(Q_{N,\xi,A}\) is a fixed polynomial and the only remaining freedom is
the sign of each of its \(d_N=\lfloor N/3\rfloor\) zero pairs
\(\pm\lambda_j\).  Every choice is a solution, so the branch carries
\(2^{d_N}\) levels with
\begin{equation}
 E=\sum_{j=1}^{d_N}\sigma_j\,\varepsilon_j,
 \qquad
 \varepsilon_j=\frac{1}{|\lambda_j|},
 \qquad \sigma_j=\pm1 .
 \label{eq:supp-fb-onewrap}
\end{equation}
This is the sector-wise free structure announced in the Letter.  The
one-particle energies are those of the open chain.  Continuing the
angle map of Eq.~\eqref{eq:supp-1w-angle-map} to positive \(t\), where
\(q=2\cos\theta\) runs over \((1,2)\),
\begin{equation}
 \varepsilon(q)=\frac{(q^2-1)^{3/2}}{q},
 \label{eq:supp-fb-dispersion}
\end{equation}
which vanishes as \((q-1)^{3/2}\) at the band edge \(q\to1\).  The
quantization spaces the angles by \(O(N^{-1})\), so the smallest
one-particle energy is of order \(N^{-3/2}\).

The band edge itself is common to all the closures.  Writing
\(r=N(\pi/3-\theta)\) for the scaled distance from the edge, the
quantization of the level-one scalar becomes, in every case,
\begin{equation}
 e^{-\sqrt3\,r}+2\cos r=0,
 \qquad
 r_1=1.6019848763\ldots,
 \label{eq:supp-fb-edge}
\end{equation}
so that
\begin{equation}
 \varepsilon_{\min}\simeq\Big(\frac{2\sqrt3\,r_1}{N}\Big)^{3/2},
 \qquad
 N^{3/2}\varepsilon_{\min}\longrightarrow
 \big(2\sqrt3\,r_1\big)^{3/2}=13.0737\ldots
 \label{eq:supp-fb-edge-scale}
\end{equation}
The open chain of Ref.~\cite{Fendley2019} realizes this directly,
because there \(t^{[2]}\) vanishes identically and \(Q\) is the
level-one scalar alone.  In the one-wrap chain the fused term of
Eq.~\eqref{eq:supp-gap-response-Q} is of the same order as
\(\phi^{\mathrm w}_N\) at large \(t\), so it displaces the edge root and
the approach is slow,
\begin{center}
\begin{tabular}{lccccc}
\hline
\(N\) & \(36\) & \(120\) & \(240\) & \(480\) & \(960\)\\
\hline
\(N^{3/2}\varepsilon_{\min}\) & \(15.543\) & \(14.751\) & \(14.409\)
& \(14.152\) & \(13.963\)\\
\hline
\end{tabular}
\end{center}
still above the value in Eq.~\eqref{eq:supp-fb-edge-scale} at the
largest size we have reached.

\emph{Periodic chain.}  Here the sign of an individual \(\lambda_j\)
cannot be flipped, because the two relations
\eqref{eq:supp-per-first} and \eqref{eq:supp-per-second} are coupled.
One family nevertheless survives for every even \(N=2M\), and it comes
from an identity between the scalar polynomials of the chain and those
of the half chain,
\begin{equation}
 \phi_M(t)^2=\phi_N(t)+2\psi_M(t) .
 \label{eq:supp-fb-identity}
\end{equation}
This is the elementary relation
\(({\rm tr}A)^2={\rm tr}A^2+2\,e_2(A)\) applied to
\(A=\mathsf S_1(t)^M\), together with the fact that \(\mathsf S_2\) is
the matrix of \(2\times2\) subdeterminants of \(\mathsf S_1\), so that
\(e_2(\mathsf S_1^M)={\rm tr}\,\mathsf S_2^M=\psi_M\).

Comparing Eq.~\eqref{eq:supp-fb-identity} with the first functional
relation shows that the choice
\begin{equation}
 \Lambda^{[1]}(u)\Lambda^{[1]}(-u)=\phi_M(u^2)^2,
 \qquad
 \Lambda^{[2]}(u)+\Lambda^{[2]}(-u)=2\psi_M(u^2)
 \label{eq:supp-fb-choice}
\end{equation}
solves it identically.  Writing
\(\phi_M(t)=\prod_a(1-\varepsilon_a^2t)\), every spectral factor of
\(\phi_M\) gives an admissible fundamental eigenvalue, and the family is
\begin{equation}
 \Lambda^{[1]}_{\boldsymbol\eta}(u)
 =\prod_{a=1}^{\lfloor N/6\rfloor}
 \big(1-\eta_a\varepsilon_au\big)^2,
 \qquad
 E_{\boldsymbol\eta}=2\sum_a\eta_a\varepsilon_a,
 \qquad \eta_a=\pm1 .
 \label{eq:supp-fb-periodic}
\end{equation}
Each \(\eta_a\) reverses a coincident pair of \(\lambda\) roots, so the
level spacing is \(4\varepsilon_a\) rather than \(2\varepsilon_a\), and
the number of modes is halved.  The \(\varepsilon_a\) are the energies
of the half chain, so they obey the same dispersion
\eqref{eq:supp-fb-dispersion} and the same edge condition
\eqref{eq:supp-fb-edge}, but with \(M=N/2\) in place of \(N\).  The
coefficient is therefore larger by \(2^{3/2}\),
\begin{equation}
 \varepsilon_{\min}\simeq
 \Big(\frac{4\sqrt3\,r_1}{N}\Big)^{3/2},
 \qquad
 N^{3/2}\varepsilon_{\min}\longrightarrow
 2^{3/2}\big(2\sqrt3\,r_1\big)^{3/2}=36.9759 .
 \label{eq:supp-fb-periodic-scale}
\end{equation}
Direct evaluation gives
\begin{center}
\begin{tabular}{lcccccc}
\hline
\(N\) & \(24\) & \(48\) & \(96\) & \(192\) & \(384\) & \(768\)\\
\hline
\(N^{3/2}\varepsilon_{\min}\) & \(33.643\) & \(35.203\) & \(36.056\)
& \(36.507\) & \(36.739\) & \(36.857\)\\
\hline
\end{tabular}
\end{center}
in agreement with Eq.~\eqref{eq:supp-fb-periodic-scale}.  We have checked
Eqs.~\eqref{eq:supp-fb-identity} and \eqref{eq:supp-fb-periodic}
against the complete physical spectrum at \(N=6,8,10,12,16,18\) and
\(20\), where every predicted square factor is realized and the
energies agree with exact diagonalization.

\subsection{Low-temperature thermodynamics of the one-wrap chain}
\label{subsec:supp-thermo-onewrap}

The tower of Sec.~\ref{subsec:supp-gap-onewrap} reverses one fused root at
a time.  We call the coordinate \(\gamma_{a,N}>0\) of a fused root its
height.  Reversing a finite fraction of them at once produces a band of
\(2^{2L}\) states whose width is \(O(N^{-1})\), and the competition of
that width against the entropy of the band fixes a temperature scale
\(T\propto N^{-2}\).  This subsection works out the resulting limit for
\(N=12L\).
For brevity in the thermodynamic formulas below, denote the
coefficient in Eq.~\eqref{eq:supp-1w-amplitude} by
\begin{equation}
 D_{\mathrm{1w}}
 :=\lim_{N=12L\to\infty}N^3\Delta_N
 =\frac{9\pi^2}{2}
 \left[-\operatorname{Li}_2\!\Big(-\frac{\sqrt3}{2}\Big)\right].
 \label{eq:supp-thermo-D1w}
\end{equation}

\emph{Energy of a many-root reversal.}  Let \(A\) be the set of reversed
heights.  On the imaginary axis at \(y=N/p\), the fused product of the
aligned branch is multiplied by
\begin{equation}
 R_{A,N}(p)=\prod_{a\in A}
 \frac{N/p-\gamma_{a,N}}{N/p+\gamma_{a,N}},
 \label{eq:supp-th-ratio}
\end{equation}
which is \(\prod_{a\in A}r_a\) in the notation of
Eq.~\eqref{eq:supp-1w-branch-order}.  Introduce the budget
\begin{equation}
 x_N=\frac1N\sum_{a\in A}\gamma_{a,N} .
 \label{eq:supp-th-budget}
\end{equation}
If \(x_N\to x\) with the largest reversed height still \(o(N)\), then
\(\log R_{A,N}(p)=-2px_N+o(1)\) at fixed \(p\), the remainder being
controlled by \(\sum_A(\gamma_{a,N}/N)^3\).  Inserting this and the
limit \(\mathcal R_N(N/p)\to a\,e^{-cp}\) of
Eq.~\eqref{eq:supp-1w-R-limit}, with
\begin{equation}
 a=\frac{\sqrt3}{2},
 \qquad
 c=\frac{2}{3\pi},
 \label{eq:supp-th-ac}
\end{equation}
into the energy integral \eqref{eq:supp-1w-exact-gap} gives
\begin{equation}
 N\big(E_A-E_{\mathrm g}\big)\longrightarrow
 e(x)=\frac1\pi\int_0^\infty
 \log\frac{1+a\,e^{-cp}}{1+a\,e^{-(c+2x)p}}\,\dd p
 =\frac{D_{\mathrm{1w}}}{\pi}\,\frac{x}{1+3\pi x} ,
 \label{eq:supp-th-energy}
\end{equation}
the integral being the same dilogarithm as in
Theorem~\ref{thm:supp-1w-gap}.  Within the controlled regime
\(\max_{a\in A}\gamma_{a,N}=o(N)\), the continuum energy functional
therefore saturates at
\begin{equation}
 e(\infty)=\frac{D_{\mathrm{1w}}}{3\pi^2}
 =1.0917189679\ldots\equiv J,
 \label{eq:supp-th-saturation}
\end{equation}
which identifies \(J/N\) as the saturation scale of this mesoscopic
reversal functional.  A bound on the width of the complete finite-size
collective band, including endpoint roots with \(\gamma_{a,N}=O(N)\) or
larger, requires a separate endpoint estimate, which we leave for future work.

Equation~\eqref{eq:supp-th-energy} ties the three regimes together.  A
single reversal of a fixed height, where
\(\gamma_{a,N}\simeq(2a-1)\pi/N\), costs
\(D_{\mathrm{1w}}(2a-1)/N^3\), which is the tower of
Theorem~\ref{thm:supp-1w-gap}.  A single reversal on a background
\(x\) costs \(e'(x)\gamma/N^2\).  Reversing \(O(N)\) heights at fixed
\(x\) costs \(e(x)/N\).  The three cannot be obtained from one another
by addition, because
\begin{equation}
 e'(x)=\frac{D_{\mathrm{1w}}/\pi}{(1+3\pi x)^2},
 \qquad
 e''(x)=-\frac{6D_{\mathrm{1w}}}{(1+3\pi x)^3}<0,
 \label{eq:supp-th-derivatives}
\end{equation}
so two reversals at finite heights have a negative connected energy
\(e''(x)\gamma_{a,N}\gamma_{b,N}/N^3\).  At fixed heights this is the
\(N^{-5}\) correction that made the tower look additive.  With \(O(N)\)
reversals it is a leading effect.

\emph{Entropy and the temperature scale.}  Label the heights by the
angle of Eq.~\eqref{eq:supp-1w-angle-map} and let \(n(\theta)\in[0,1]\)
be the local reversed fraction.  With
\begin{equation}
 \dd\mu=\frac{\dd\theta}{\pi},
 \qquad
 \int_{\pi/3}^{\pi/2}\dd\mu=\frac16,
 \qquad
 x[n]=\int\gamma(\theta)\,n(\theta)\,\dd\mu ,
 \label{eq:supp-th-measure}
\end{equation}
the number of sign words with profile \(n\) is \(\exp[N\,s_\gamma[n]]\)
to leading exponential order, where
\begin{equation}
 s_\gamma[n]=\int
 \big[-n\log n-(1-n)\log(1-n)\big]\dd\mu,
 \qquad
 \sup_n s_\gamma=\frac{\log2}{6} .
 \label{eq:supp-th-entropy}
\end{equation}
An entropy of order \(N\) competing with an energy of order \(N^{-1}\)
fixes \(T=\tau/N^2\), and the rate function is
\begin{equation}
 \psi(\tau)=\sup_{0\leq n\leq1}
 \left\{s_\gamma[n]-\frac{e(x[n])}{\tau}\right\}.
 \label{eq:supp-th-rate}
\end{equation}
At fixed \(x\) the entropy is maximized by a Lagrange multiplier
\(t\), and the stationary profile is
\begin{equation}
 n_\tau(\theta)=\frac{1}{1+e^{t_\tau\gamma(\theta)}},
 \qquad
 x_\tau=\int\frac{\gamma\,\dd\mu}{1+e^{t_\tau\gamma}},
 \qquad
 t_\tau=\frac{e'(x_\tau)}{\tau} .
 \label{eq:supp-th-saddle}
\end{equation}
The profile has the shape of a Fermi factor, but \(t_\tau\) is fixed by
the occupation itself through \(e'\), so these are not independent
fermions with a fixed set of one-particle energies.  The map
\(t\mapsto\tau\) is strictly monotone, so the solution is unique and
there is no transition at finite \(\tau\).

\emph{Result.}  With
\begin{equation}
 \mathcal V(t)=\int\gamma^2\,
 \frac{e^{t\gamma}}{(1+e^{t\gamma})^2}\,\dd\mu ,
 \label{eq:supp-th-V}
\end{equation}
the entropy and the specific heat converge uniformly on every interval
\(0<\tau_-\leq\tau\leq\tau_+<\infty\) to
\begin{equation}
 \frac{S_N}{N}\longrightarrow\frac12\log2+s_\gamma(\tau),
 \qquad
 \frac{C_N(\tau/N^2)}{N}\longrightarrow
 c_\gamma(\tau)=
 \frac{t_\tau^2\,\mathcal V(t_\tau)}
 {1+e''(x_\tau)\mathcal V(t_\tau)/\tau} ,
 \label{eq:supp-th-result}
\end{equation}
the denominator staying above \(1/3\).  The \(\tfrac12\log2\) is the
residual entropy of the one-wrap chain, and the collective band adds up
to \(\tfrac16\log2\) on top of it.  The two ends and the maximum are
\begin{equation}
 c_\gamma(\tau)\simeq\frac{\pi^2}{12D_{\mathrm{1w}}}\,\tau
 =0.0254440736\,\tau
 \quad(\tau\to0),
 \qquad
 c_\gamma(\tau)\simeq\frac{K}{\tau^2}
 \quad(\tau\to\infty),
 \label{eq:supp-th-ends}
\end{equation}
with
\begin{equation}
 K=\frac{J^2}{9\pi^2A^3}=16.2428224120\ldots,
 \qquad
 A=\frac{\Gamma(1/3)\big(1-2^{2/3}\big)\zeta(1/3)}{3\sqrt3\,\pi}
 =0.0938294824\ldots,
 \label{eq:supp-th-K}
\end{equation}
and
\begin{equation}
 \tau_{\rm p}=4.6049355402\ldots,
 \qquad
 c_\gamma(\tau_{\rm p})=0.0313367021\ldots
 \label{eq:supp-th-peak}
\end{equation}
The peak of the specific heat therefore sits at
\(T_{\rm p}=\tau_{\rm p}N^{-2}\) and is extensive,
\(C_{\rm p}=0.0313367\,N\).  A single boundary bond thus produces a
peak whose height grows with the volume.

The free branch of Sec.~\ref{subsec:supp-gap-freebranch} is separated from
this window.  At \(T=\tau/N^2\) its lowest mode obeys
\(\varepsilon_{\min}/T\simeq(c_\varepsilon/\tau)\sqrt N\) with
\(c_\varepsilon\) the coefficient of
Eq.~\eqref{eq:supp-fb-edge-scale}, so the \(\lambda\) sector is exponentially
suppressed as \(e^{-(c_\varepsilon/\tau)\sqrt N}\) in
Eq.~\eqref{eq:supp-th-result}.  Conversely, at fixed \(T\) the
collective band is saturated and the \(\lambda\) sector gives the bulk
\(T^{2/3}\) law.  The two scales are separated by a window that widens
with \(N\).  Figure~\ref{fig:supp-periodic-specific-heat} shows this at
\(N=120\): the one-wrap peak is confined to the \(T\sim N^{-2}\) window,
while at larger \(T\) the open and one-wrap curves approach the same bulk
specific heat.

\begin{figure}[!htbp]
\centering
\includegraphics[width=0.38\columnwidth]{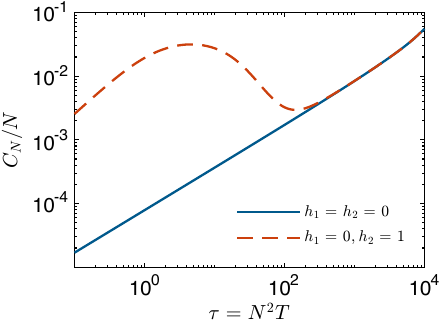}
\caption{Specific heat of the $N=120$ chain on a double-logarithmic scale.
The solid curve has $h_1=h_2=0$ and the dashed curve has $h_1=0$, $h_2=1$.
The latter shows the one-wrap collective peak, while both curves approach
the same bulk tail at large $\tau=N^2T$.}
\label{fig:supp-periodic-specific-heat}
\end{figure}

\subsection{Arbitrary chain length}
\label{subsec:supp-thermo-alllength}

The extension to all chain lengths uses the root counting of
Sec.~\ref{subsec:supp-gap-alllength} and one uniform endpoint bound, stated
at the end of this subsection.  Under that bound the same limit function
\(c_\gamma(\tau)\) is obtained along all integer lengths.

The counting of Sec.~\ref{subsec:supp-gap-alllength} replaces \(2L\) by
\(\nu_N\) and puts back the normalization \(\sqrt{c_N}\), which is what
makes the odd chains share the same amplitude.  The angular labels are
those of Eq.~\eqref{eq:supp-al-quantization}, which shift the lattice by
half a cell between even and odd \(N\) but leave the limiting measure
\(\dd\mu=\dd\theta/\pi\) and its total weight \(1/6\) unchanged.  Since
\(\dd\mu\) and \(e(x)\) are the only inputs of
Eqs.~\eqref{eq:supp-th-rate}--\eqref{eq:supp-th-result}, the limit
function is independent of the size class,
\begin{equation}
 \sup_{\tau_-\leq\tau\leq\tau_+}
 \left|\frac{C_N(\tau/N^2)}{N}-c_\gamma(\tau)\right|
 \longrightarrow0
 \qquad\text{along all integer lengths for which the uniform endpoint
 bound holds.}
 \label{eq:supp-th-alllength}
\end{equation}
In particular \(N^2T_{\rm p}\to\tau_{\rm p}\) and
\(C_N(T_{\rm p})/N\to c_\gamma(\tau_{\rm p})\) with the constants of
Eq.~\eqref{eq:supp-th-peak}.

The freezing of the \(\lambda\) sector has to be checked class by
class, since the residue of \(N\) modulo three changes the endpoint
equation.  The bound needed is that every nonzero one-particle energy
satisfies \(\varepsilon\geq0.1\,N^{-3/2}\) in every physical sector.  It
is supported by the three endpoint laws and by the root-sum estimates in
the corresponding classes.

\subsection{The periodic chain}
\label{subsec:supp-thermo-periodic}

Finite-size continuation suggests that the periodic chain may possess a
collective band with the same sign-word combinatorics and temperature scale,
but with a distinct energy functional.  The exact energy kernel keeps the
same form.  Write \(\boldsymbol\eta\) for the sign word of the second
level and \(\boldsymbol\eta_0\) for its ground configuration.  Since
\(Q_{N,\boldsymbol\eta}(u)\) is an even polynomial, we write
\begin{equation}
 Q_{N,\boldsymbol\eta}(u)
 =\widehat Q_{N,\boldsymbol\eta}(u^2),
 \label{eq:supp-th-per-Qhat}
\end{equation}
where \(\widehat Q_{N,\boldsymbol\eta}(t)\) is a polynomial in \(t=u^2\).
Then, along the imaginary axis \(u=\ii N/p\),
\begin{equation}
 e_N(\boldsymbol\eta)
 :=N\big(E_{N,\boldsymbol\eta}-E_{N,\boldsymbol\eta_0}\big)
 =\frac1\pi\int_0^\infty
 \log\frac{\widehat Q_{N,\boldsymbol\eta_0}(-N^2/p^2)}
 {\widehat Q_{N,\boldsymbol\eta}(-N^2/p^2)}\,\dd p .
 \label{eq:supp-th-per-kernel}
\end{equation}
The difference from Eq.~\eqref{eq:supp-th-energy} is that each
\(\widehat Q_{N,\boldsymbol\eta}\) must now be obtained from the fixed
point of the loop~\eqref{eq:supp-per-loop}, so the periodic feedback
produces a different functional.  This defines a formal periodic energy
functional \(\mathcal E_{\rm PBC}[n]\); its explicit closed form, and the
existence and uniqueness of the corresponding fixed point for every
macroscopic profile, are left for future work.  Conditional on these properties,
the rate function is
\begin{equation}
 \psi_{\rm PBC}(\tau)=\sup_n
 \left\{s_\gamma[n]-\frac{\mathcal E_{\rm PBC}[n]}{\tau}\right\},
 \label{eq:supp-th-per-rate}
\end{equation}
and, if the periodic fixed-point branches converge uniformly on every interval
\(0<\tau_-\leq\tau\leq\tau_+<\infty\) and the states outside the collective band are
uniformly suppressed, then
\begin{equation}
 \frac{C_N(\tau/N^2)}{N}\longrightarrow
 c_{\rm PBC}(\tau)=2\tau\psi_{\rm PBC}'(\tau)
 +\tau^2\psi_{\rm PBC}''(\tau) .
 \label{eq:supp-th-per-C}
\end{equation}

The finite-size relations already determine two robust thermodynamic
consequences.  First, the band contains \(2^{\nu_N}\) sign words, giving
\begin{equation}
 \int_0^\infty\frac{c_{\rm PBC}(\tau)}{\tau}\,\dd\tau
 =\frac{\log2}{6} ,
 \label{eq:supp-th-per-sumrule}
\end{equation}
which fixes a nonzero integrated response.  Second, if two fixed values
\(0<\tau_-<\tau_+<\infty\) capture half of that entropy, then
\begin{equation}
 c_{\rm PBC}(\tau_{\rm p})\geq
 \frac{\log2}{12\log(\tau_+/\tau_-)}>0
 \label{eq:supp-th-per-peak}
\end{equation}
for some \(\tau_{\rm p}\) in that interval.  Under the fixed-point and
band-counting assumptions above, the periodic chain therefore has an
extensive collective peak at \(T\propto N^{-2}\).  Its detailed position
and shape are determined by the periodic energy functional.

\section{Further examples}
\label{sec:supp-examples}

We finally illustrate the range of the construction beyond the three-site
FFD Hamiltonian.  The examples probe three directions: higher-rank
auxiliary spaces, irreducible local representations, and face-type Lax
operators.  A scalar auxiliary closure gives the hidden free-fermion
reduction, whereas retaining the coupled fusion levels gives nested or
interacting spectral equations.  We first consider the arbitrary-\(k\)
Perk--Schultz realization, then an FFD logarithmic charge with claws, an
irreducible \(D=4\) representation, and the Fendley--Pozsgay model.

\subsection{The arbitrary-\texorpdfstring{$k$}{k} non-Hermitian
Perk--Schultz chain}
\label{subsec:supp-PSk}

The generic twisted and inhomogeneous Perk--Schultz chain is already
solved by the nested algebraic Bethe ansatz (ABA) and its fused transfer-matrix
relations~\cite{KulishReshetikhinSklyanin1981,KirillovReshetikhin1987,
KunibaNakanishiSuzuki1994,WangYangCaoShi2015}.  Our question here is
different: we reconstruct the same spectrum directly from the reflected
fusion hierarchy of Secs.~\ref{sec:supp-kall-fusion} and
\ref{sec:supp-twisted-closure}, without using a Bethe reference state.
This places the arbitrary-$k$ Perk--Schultz family inside the same
ABA-free polynomial framework as the HFF realizations studied above and
makes the root-of-unity and Jordan-block limits part of the same derivation.

Take the local Lax operator to be
\begin{equation}
 L_{0j}^{[1]}(u)=R_{0j}^{(k)}(u,\xi_j),
 \label{eq:supp-PSk-L}
\end{equation}
with the $R$ matrix of Eq.~\eqref{eq:supp-kall-R}.  For a diagonal auxiliary
boundary twist $\Omega=\operatorname{diag}(\omega_1,\ldots,\omega_k)$,
let $\Omega^{[a]}$ denote the twist induced on $V^{[a]}$ and define
\begin{equation}
 t_{\Omega}^{[a]}(u)=\operatorname{tr}_{V^{[a]}}
 \left[\Omega^{[a]}T^{[a]}(u)\right],
 \qquad a=1,\ldots,k,
 \label{eq:supp-PSk-fused-transfer}
\end{equation}
where $T^{[a]}$ is the fused monodromy constructed by the recursive
embeddings of Sec.~\ref{subsec:supp-finite-hierarchy}.  In the basis labelled by
$ I=\{i_1<\cdots<i_a\}$, the induced twist has eigenvalue
$\omega_I=\prod_{i\in I}\omega_i$.  Diagonal twists preserve every reflected
fusion channel, so the paired-trace proof applies without modification.

Writing $\Lambda_a(u)$ for a joint eigenvalue of $t_{\Omega}^{[a]}(u)$,
we obtain the finite reflected hierarchy
\begin{equation}
\begin{aligned}
\Lambda_a(u)\Lambda_a(-u)
={}&\Phi_{a,N}^{\Omega}(u^2)
+\sum_{b=1}^{\min(a,k-a)}(-1)^{b-1}\\[-1mm]
&\times\big[\Lambda_{a-b}(u)\Lambda_{a+b}(-u)
+\Lambda_{a-b}(-u)\Lambda_{a+b}(u)\big],
\end{aligned}
\label{eq:supp-PSk-fusion-spectrum}
\end{equation}
with $\Lambda_0(u)=1$.  The scalar source is
\begin{equation}
 \Phi_{a,N}^{\Omega}(u^2)=
 \operatorname{tr}_{\mathcal M_a}\!\left[
 (\Omega^{[a]})^2\,\mathsf S_{a,N}(u^2)\cdots
 \mathsf S_{a,1}(u^2)\right].
 \label{eq:supp-PSk-source}
\end{equation}
For the fundamental level one finds directly
$\mathsf S_{1,j}(u^2)=(\xi_j^2-u^2)\mathbf1_k$, and hence
\begin{equation}
 \Phi_{1,N}^{\Omega}(u^2)=
 \left(\sum_{a=1}^{k}\omega_a^2\right)
 \prod_{j=1}^{N}(\xi_j^2-u^2).
 \label{eq:supp-PSk-source1}
\end{equation}
The higher $\mathsf S_{a,j}$ are obtained recursively from the locked blocks
of the fused local products.

The top level has a one-dimensional auxiliary space.  We denote its central eigenvalue by
$d_\chi(u)$, where $\chi$ labels the corresponding color/central sector,
and set $\Lambda_k(u)=d_\chi(u)$.  Keeping $d_\chi$ explicit is essential at
$q=\ii$, where the quantum determinant may carry a nontrivial color-diagonal
central factor.  The hierarchy therefore closes on the $k-1$ unknown
polynomials $\Lambda_1,\ldots,\Lambda_{k-1}$ in each $\chi$ sector.
Their degrees, normalization at the regular point, and large-$u$ coefficients,
together with Eq.~\eqref{eq:supp-PSk-fusion-spectrum}, give a finite polynomial
system.  Equivalently, at every zero of a selected polynomial the right-hand
side must vanish.  The energy of the homogeneous untwisted chain is read off
from the first level as
\begin{equation}
 E=2\xi\,\frac{\Lambda_1'(\xi)}{\Lambda_1(\xi)}-N,
 \qquad (\Omega=\mathbf1,\ \xi_j=\xi).
 \label{eq:supp-PSk-fusion-energy}
\end{equation}
This gives an ABA-free finite polynomial formulation of the joint spectrum.

\paragraph{The $k=4$ specialization.}
For $k=4$ the nontrivial levels are $a=1,2,3$, and the hierarchy reads
\begin{align}
\Lambda_1(u)\Lambda_1(-u)
 &=\Phi_1^\Omega(u^2)+\Lambda_2(u)+\Lambda_2(-u),
 \label{eq:supp-PS4-fusion1}\\
\Lambda_2(u)\Lambda_2(-u)
 &=\Phi_2^\Omega(u^2)+\Lambda_1(u)\Lambda_3(-u)
 +\Lambda_1(-u)\Lambda_3(u)
 -\Lambda_4(u)-\Lambda_4(-u),
 \label{eq:supp-PS4-fusion2}\\
\Lambda_3(u)\Lambda_3(-u)
 &=\Phi_3^\Omega(u^2)+\Lambda_2(u)\Lambda_4(-u)
 +\Lambda_2(-u)\Lambda_4(u).
 \label{eq:supp-PS4-fusion3}
\end{align}
Here $\Lambda_4=d_\chi$ is the central top-level eigenvalue.  The first
source is
$\Phi_1^\Omega=(\omega_1^2+\omega_2^2+\omega_3^2+\omega_4^2)
\prod_j(\xi_j^2-u^2)$.  The sources $\Phi_2^\Omega$ and $\Phi_3^\Omega$ follow
from the same locked-block recursion.  Equations~\eqref{eq:supp-PS4-fusion1}--
\eqref{eq:supp-PS4-fusion3}, together with the polynomial degree and
regularity constraints,
provide a finite algebraic formulation for the $k=4$ chain.

For the $k=4$ hierarchy we parameterize the three nontrivial joint
transfer-matrix eigenvalues by
\begin{equation}
\Lambda_1(u)=c_1\prod_{j=1}^{d_1}\left(1-\frac{u}{\lambda_j}\right),\qquad
\Lambda_2(u)=c_2\prod_{\alpha=1}^{d_2}\left(1-\frac{u}{\gamma^{(1)}_\alpha}\right),\qquad
\Lambda_3(u)=c_3\prod_{\beta=1}^{d_3}\left(1-\frac{u}{\gamma^{(2)}_\beta}\right),
\label{eq:supp-PS4-root-param}
\end{equation}
where the generic degrees satisfy $d_a\leq aN$ and are fixed by the
large-$u$ normalization.  Degree-deficient branches are retained.  

\paragraph{Weight, regularity, and completeness constraints.}
The fused equations are supplemented by three independent sets of
constraints.  First, the Perk--Schultz chain preserves the diagonal color
charges
\begin{equation}
 {\cal Q}_a=\sum_{j=1}^{N}E_{aa}^{(j)},\qquad
 [H_{\rm PS}^{(k)},{\cal Q}_a]=0,\qquad
 \sum_{a=1}^{k}{\cal Q}_a=N\mathbf1 .
 \label{eq:supp-PSk-weight-charges}
\end{equation}
We therefore solve the fused polynomial system separately in each weight
sector
\[
 {\bf n}=(n_1,\ldots,n_k),\qquad
 n_a\in{\mathbb Z}_{\geq0},\qquad
 \sum_{a=1}^{k}n_a=N ,
\]
whose Hilbert-space dimension is
\[
 \dim{\cal H}_{\bf n}=\frac{N!}{n_1!\cdots n_k!}.
\]
The weight sector fixes the corresponding asymptotic coefficients of the
fused eigenvalues and excludes solutions belonging to a different color
content.

Second, regularity at \(u=\xi\) fixes the translation quantum numbers.  For
the homogeneous untwisted chain,
\begin{equation}
 t^{[1]}(\xi)=(2\xi)^N{\cal U}^{-1},
 \qquad
 \Lambda_1(\xi)=(2\xi)^N e^{-\ii P},
 \label{eq:supp-PSk-regular-point}
\end{equation}
where \({\cal U}\) is the one-site translation operator.  The fused regular
values \(\Lambda_a(\xi)\) are fixed analogously by the translation eigenvalues
in the fused auxiliary representations.  These conditions determine the
remaining normalization constants \(c_a\) and remove spurious polynomial
branches.

Finally, the fused relations are imposed as polynomial identities after
substitution of Eq.~\eqref{eq:supp-PS4-root-param}.  Completeness is tested
sector by sector: the number of admissible solutions, counted with
algebraic multiplicity, must reproduce
\[
 \sum_{\bf n}\dim{\cal H}_{\bf n}=k^N .
\]
At the root-of-unity point, degree-deficient solutions and coalescing roots
are retained by continuation from generic twists and inhomogeneities. 

\subsection{Periodic FFD logarithmic charge with claws}
\label{subsec:supp-claw}

Let
\begin{equation}
 h_j=\sigma_j^z\sigma_{j+1}^z\sigma_{j+2}^x,
 \qquad j=1,\ldots,N,
 \qquad h_{j+N}=h_j ,
 \label{eq:supp-claw-h}
\end{equation}
and expand the same fundamental periodic FFD transfer matrix used in the
Letter as
\begin{equation}
 t_{\mathrm F}^{[1]}(u)
 =\1-uQ_1+u^2Q_2-u^3Q_3+\cdots ,
 \qquad Q_1=H_{\mathrm F}=\sum_{j=1}^{N}h_j .
 \label{eq:supp-claw-transfer-expansion}
\end{equation}
Its logarithmic charges are defined by
\begin{equation}
 \log t_{\mathrm F}^{[1]}(u)
 =-\sum_{r\geq1}\frac{u^r}{r}\mathcal J_r .
 \label{eq:supp-claw-logarithmic-charges}
\end{equation}
The first nontrivial odd charge is
\begin{equation}
 \mathcal J_3
 =Q_1^3-3Q_1Q_2+3Q_3
 =-\frac12
 \left.\partial_u^3\log t_{\mathrm F}^{[1]}(u)\right|_{u=0}.
 \label{eq:supp-claw-J3-Newton}
\end{equation}
Using \(h_j^2=\1\), the anticommutation of generators at separations one
and two, and commutativity at larger separations gives the local periodic
form
\begin{equation}
\begin{aligned}
 \mathcal J_3
 ={}&5\sum_{j=1}^{N}h_j
 \\
 &+\sum_{j=1}^{N}\left(
 h_jh_{j+1}h_{j+3}
 {}+h_jh_{j+2}h_{j+3}
 {}+h_jh_{j+2}h_{j+4}
 \right).
\end{aligned}
\label{eq:supp-claw-J3-local}
\end{equation}
All products in Eq.~\eqref{eq:supp-claw-J3-local} are ordered cyclically.
They are Hermitian, and their local frustration graph contains the claw
substructures of Ref.~\cite{FukaiVonaPozsgay2025}.  A convenient
one-parameter stitching with unit coefficient for the fundamental terms is
\begin{equation}
\begin{aligned}
 H_{\rm claw}(g)
 &=
 \mathcal J_1+g(\mathcal J_3-5\mathcal J_1)
 \\
 &=
 \sum_{j=1}^{N}h_j
 +g\sum_{j=1}^{N}\left(
 h_jh_{j+1}h_{j+3}
 {}+h_jh_{j+2}h_{j+3}
 {}+h_jh_{j+2}h_{j+4}
 \right).
\end{aligned}
\label{eq:supp-claw-stitched-H}
\end{equation}
The charge is generated by \(t_{\mathrm F}^{[1]}(u)\), so it commutes with
the complete fused transfer family and uses the same periodic functional
relations and root selection rules as \(H_{\mathrm F}\).  The claw model
therefore provides a higher conserved charge within the same interacting
periodic FFD realization.

Indeed, on a physical joint transfer character,
\begin{equation}
 \Lambda_\chi^{[1]}(u)
 =\prod_{j=1}^{d_N}\left(1-\frac{u}{\lambda_j}\right)
 \quad\Longrightarrow\quad
 E_{\mathcal J_r}=\sum_{j=1}^{d_N}\lambda_j^{-r},
 \qquad d_N=\left\lfloor\frac N3\right\rfloor .
 \label{eq:supp-claw-logarithmic-spectrum}
\end{equation}
Consequently, the strict periodic solution of the pure claw charge is
\begin{equation}
 H_{\rm claw}=\mathcal J_3,
 \qquad
 E_{\rm claw}=\sum_{j=1}^{d_N}\lambda_j^{-3},
 \label{eq:supp-claw-periodic-solution}
\end{equation}
where the \(\lambda_j\), the fused roots, and their physical selection rules
are precisely those of the periodic FFD solution.  For the stitched
Hamiltonian in Eq.~\eqref{eq:supp-claw-stitched-H}, the same roots give
\begin{equation}
 E_{\rm claw}(g)
 =\sum_{j=1}^{d_N}
 \left[(1-5g)\lambda_j^{-1}+g\lambda_j^{-3}\right].
 \label{eq:supp-claw-stitched-energy}
\end{equation}

\subsection{An irreducible \texorpdfstring{$D=4$}{D=4} range-five model}
\label{subsec:supp-D4}

We next construct a different solution of the same RLL equation, with an
irreducible four-dimensional quantum space and the same three-dimensional
auxiliary space.

Let the local quantum space be
\(\mathcal H_j\simeq\mathbb C^2\otimes\mathbb C^2\), and introduce
\begin{equation}
 P=\sigma^z\otimes\1,\qquad
 D_2=\sigma^x\otimes\1,\qquad
 Q=\1\otimes\sigma^z,\qquad
 D_3=\1\otimes\sigma^y,\qquad
 C=\sigma^y\otimes\sigma^x .
 \label{eq:supp-D4-generators}
\end{equation}
The first four matrices already generate
\(\operatorname{End}(\mathbb C^4)\), so the physical representation is
irreducible.  The \(k=3\) Lax operator
\begin{equation}
{
 L_j^{[1]}(u)=
 \begin{pmatrix}
 \rho_{1,j}\1_j&0&-u\alpha_j C_j\\
 z_{1,j}P_j&\rho_{2,j}D_{2,j}&0\\
 0&z_{2,j}Q_j&\rho_{3,j}D_{3,j}
 \end{pmatrix}}
\label{eq:supp-D4-L}
\end{equation}
satisfies the RLL relation
\eqref{eq:supp-framework-RLL} with the same three-state
Perk--Schultz \(R\) matrix.

A model using both diagonal channels is obtained from the two-site unit
cell
\begin{equation}
\begin{array}{c|ccccc}
 &\rho_1&\rho_2&\rho_3&z_1&z_2\\
\hline
 S&1&1&1&0&0\\
 T&1&0&0&1&1
\end{array},
\qquad \alpha_S=\alpha_T=1 .
\label{eq:supp-D4-cells}
\end{equation}
The two local tensors are therefore
\begin{equation}
 L_S(u)=
 \begin{pmatrix}
 \1&0&-uC\\
 0&D_2&0\\
 0&0&D_3
 \end{pmatrix},
 \qquad
 L_T(u)=
 \begin{pmatrix}
 \1&0&-uC\\
 P&0&0\\
 0&Q&0
 \end{pmatrix}.
 \label{eq:supp-D4-ST-L}
\end{equation}
Place \(S\) on odd sites and \(T\) on even sites, and choose the fixed
auxiliary color \(|1\rangle\):
\begin{equation}
 \tau_4(u)
 =
 \langle1|
 L_1^{[1]}(u)L_2^{[1]}(u)\cdots L_N^{[1]}(u)
 |1\rangle,
 \qquad
 H_4=-\tau_4'(0).
 \label{eq:supp-D4-transfer}
\end{equation}
There are two possible closed auxiliary paths containing one
\((-uC)\) vertex.  They give
\begin{align}
 h_{2n-1}
 &=
 C_{2n-1}Q_{2n}D_{2,2n+1}P_{2n+2},
 \label{eq:supp-D4-h-odd}\\
 h_{2n}
 &=
 C_{2n}D_{3,2n+1}Q_{2n+2}
 D_{2,2n+3}P_{2n+4},
 \label{eq:supp-D4-h-even}
\end{align}
and hence
\begin{equation}
{
 H_4
 =
 \sum_{\substack{n\geq1\\2n+2\leq N}}h_{2n-1}
 +
 \sum_{\substack{n\geq1\\2n+4\leq N}}h_{2n}.}
\label{eq:supp-D4-H}
\end{equation}
The two families have ranges four and five, respectively.  Because both
\(D_2\) and \(D_3\) occur, this representation combines the sitewise
sectors of the three-site realization into a single irreducible model.

The boundary pair \(\langle11|\cdots|11\rangle\) selects the scalar
diagonal channel of the reflected product derived above.  Therefore
\begin{equation}
 [\tau_4(u),\tau_4(v)]=0,\qquad
 \tau_4(u)\tau_4(-u)=P_N^{(4)}(u^2)\1 .
 \label{eq:supp-D4-inversion}
\end{equation}
For orientation, the first even-length scalar polynomials are
\begin{equation}
 P_4^{(4)}(x)=1-x,\qquad
 P_6^{(4)}(x)=1-3x,\qquad
 P_8^{(4)}(x)=1-5x+x^2,\qquad
 P_{10}^{(4)}(x)=1-7x+5x^2 .
 \label{eq:supp-D4-P-examples}
\end{equation}
Writing
\begin{equation}
 P_N^{(4)}(u^2)
 =\prod_{a=1}^{d_4}(1-u^2\epsilon_{a,N}^2)
 \label{eq:supp-D4-P-factor}
\end{equation}
gives the joint eigenvalues
\begin{equation}
 \Lambda_{\boldsymbol\eta}^{(4)}(u)
 =\prod_{a=1}^{d_4}(1-u\eta_a\epsilon_{a,N}),
 \qquad
 E_{\boldsymbol\eta}=\sum_{a=1}^{d_4}\eta_a\epsilon_{a,N},
 \qquad \eta_a=\pm1 .
\label{eq:supp-D4-solution}
\end{equation}
We now close the same staggered row periodically.  Let \(N=2M\), read all
physical indices modulo \(N\), and define
\begin{equation}
 t_{4,\mathrm{per}}^{[a]}(u)
 =\operatorname{tr}_{V^{[a]}}
 L_{1}^{[a]}(u)L_{2}^{[a]}(u)\cdots L_{N}^{[a]}(u),
 \qquad a=1,2,3.
 \label{eq:supp-D4-periodic-transfers}
\end{equation}
The fundamental trace is regular and generates the periodic Hamiltonian,
\begin{equation}
 t_{4,\mathrm{per}}^{[1]}(0)=\1,\qquad
 H_{4}^{\mathrm{per}}
 =-\left.\partial_u t_{4,\mathrm{per}}^{[1]}(u)\right|_{u=0}
 =\sum_{n=1}^{M}\bigl(h_{2n-1}+h_{2n}\bigr),
 \label{eq:supp-D4-periodic-H}
\end{equation}
where Eqs.~\eqref{eq:supp-D4-h-odd} and
\eqref{eq:supp-D4-h-even} are understood cyclically.  Coincident wrapped
terms at the shortest sizes are combined.

The first fused local tensors obtained from
Eq.~\eqref{eq:supp-higher-La} are
\begin{equation}
 L_{4,S}^{[2]}(u)=
 \begin{pmatrix}
 D_2&0&\ii u\,\sigma^z\otimes\sigma^x\\
 0&D_3&0\\
 0&0&\sigma^x\otimes\sigma^y
 \end{pmatrix},
 \qquad
 L_{4,T}^{[2]}(u)=
 \begin{pmatrix}
 0&-\ii u\,\sigma^x\otimes\sigma^x&0\\
 Q&0&\ii u\,\sigma^y\otimes\sigma^y\\
 \sigma^z\otimes\sigma^z&0&0
 \end{pmatrix}.
 \label{eq:supp-D4-periodic-L2}
\end{equation}
The one-dimensional top channel is especially simple:
\begin{equation}
 L_{4,S}^{[3]}(u)=\sigma^x\otimes\sigma^y,\qquad
 L_{4,T}^{[3]}(u)=-u\sigma^x\otimes\sigma^y,\qquad
 t_{4,\mathrm{per}}^{[3]}(u)=(-u)^M\mathcal W_N,
 \label{eq:supp-D4-periodic-top}
\end{equation}
where
\begin{equation}
 \mathcal W_N=\prod_{j=1}^{N}(\sigma^x\otimes\sigma^y)_j,\qquad
 \mathcal W_N^2=\1.
 \label{eq:supp-D4-periodic-closing-operator}
\end{equation}
Thus \(\mathcal W_N\) is the global closing operator of the periodic fusion
hierarchy.

The two scalar returns can be evaluated over one \(ST\) cell.  In a common
locked-channel gauge they are generated by
\begin{equation}
 \mathsf A_1^{(4)}(x)=
 \begin{pmatrix}
 1&0&-2x\\
 1&0&-x\\
 0&1&0
 \end{pmatrix},\qquad
 \mathsf A_2^{(4)}(x)=
 \begin{pmatrix}
 0&x&0\\
 1&0&2x\\
 1&0&x
 \end{pmatrix},\qquad x=u^2,
 \label{eq:supp-D4-periodic-scalar-matrices}
\end{equation}
so that
\begin{equation}
 \Phi_{a,N}^{(4)}(u^2)
 =\operatorname{tr}\left[\mathsf A_a^{(4)}(u^2)^M\right],
 \qquad a=1,2.
 \label{eq:supp-D4-periodic-sources}
\end{equation}
Equivalently, the two finite scalar recursions are fixed by
\begin{align}
 \det\left[z\1-\mathsf A_1^{(4)}(x)\right]
 &=z^3-z^2+xz+x,
 \nonumber\\
 \det\left[z\1-\mathsf A_2^{(4)}(x)\right]
 &=z^3-xz^2-xz-x^2.
 \label{eq:supp-D4-periodic-spectral-curves}
\end{align}

Substitution into the finite \(k=3\) hierarchy gives the closed operator
relations
\begin{align}
 t_{4,\mathrm{per}}^{[1]}(u)t_{4,\mathrm{per}}^{[1]}(-u)
 ={}&\Phi_{1,N}^{(4)}(u^2)\1
 +t_{4,\mathrm{per}}^{[2]}(u)
 +t_{4,\mathrm{per}}^{[2]}(-u),
 \label{eq:supp-D4-periodic-first-fusion}\\
 t_{4,\mathrm{per}}^{[2]}(u)t_{4,\mathrm{per}}^{[2]}(-u)
 ={}&\Phi_{2,N}^{(4)}(u^2)\1
 +u^M\mathcal W_Nt_{4,\mathrm{per}}^{[1]}(u)
 \nonumber\\
 &+(-u)^M\mathcal W_Nt_{4,\mathrm{per}}^{[1]}(-u).
 \label{eq:supp-D4-periodic-second-fusion}
\end{align}
All three transfer matrices commute.  On a simultaneous eigenstate with
\(\mathcal W_N=\chi=\pm1\), their two nontrivial eigenvalue polynomials
satisfy
\begin{align}
 \Lambda_{\chi}^{[1]}(u)\Lambda_{\chi}^{[1]}(-u)
 ={}&\Phi_{1,N}^{(4)}(u^2)
 +\Lambda_{\chi}^{[2]}(u)+\Lambda_{\chi}^{[2]}(-u),
 \label{eq:supp-D4-periodic-eigenvalue-1}\\
 \Lambda_{\chi}^{[2]}(u)\Lambda_{\chi}^{[2]}(-u)
 ={}&\Phi_{2,N}^{(4)}(u^2)
 +\chi u^M\left[
 \Lambda_{\chi}^{[1]}(u)
 +(-1)^M\Lambda_{\chi}^{[1]}(-u)
 \right].
 \label{eq:supp-D4-periodic-eigenvalue-2}
\end{align}
The loop lengths in Eq.~\eqref{eq:supp-D4-periodic-L2} fix the polynomial
supports.  In the same parametrization as the FFD chain, write
\begin{equation}
 \begin{aligned}
 \Lambda_{\chi}^{[1]}(u)
 &=\prod_{j=1}^{d_{4,N}}\left(1-\frac{u}{\lambda_j}\right),\\
 \Lambda_{\chi}^{[2]}(u)
 &=\kappa_\chi u^{s_{4,N}}
 \prod_{\alpha=1}^{n_{4,N}}
 \left(1-\frac{u}{\gamma_\alpha}\right),\\
 d_{4,N}&=\left\lfloor\frac N4\right\rfloor,
 \qquad s_{4,N}=\left\lceil\frac N4\right\rceil,
 \qquad n_{4,N}=\frac N2-s_{4,N}.
 \end{aligned}
 \label{eq:supp-D4-periodic-polynomials}
\end{equation}
Matching all powers of \(u\) in
Eqs.~\eqref{eq:supp-D4-periodic-eigenvalue-1} and
\eqref{eq:supp-D4-periodic-eigenvalue-2} is therefore a finite polynomial
system for the periodic spectrum.  Finally,
\begin{equation}
 E=\sum_{j=1}^{d_{4,N}}\frac1{\lambda_j}
 \label{eq:supp-D4-periodic-energy}
\end{equation}
gives the energy of \(H_4^{\mathrm{per}}\).  The fixed-boundary spectrum is
the HFF scalar closure of this irreducible \(D=4\) representation.  The
periodic construction retains the coupled fusion levels and gives the
corresponding interacting polynomial system.

\subsection{The Fendley--Pozsgay face model}
\label{subsec:supp-FP}

The Perk--Schultz \(R\) matrix also admits
face-type (dynamical) extensions~\cite{Felder1994,EtingofVarchenko1998}.
This enlarges the class of integrable chains by encoding the local face
configuration through an invertible operator \(G_j(x)\) acting on the quantum
space.  The resulting local intertwining relation is
\begin{equation}
 R_{12}(u,v)L^{\rm face}_{1j}(u;x)G_j(x)^{-1}L^{\rm face}_{2j}(v;x)
 =L^{\rm face}_{2j}(v;x)G_j(x)^{-1}L^{\rm face}_{1j}(u;x)R_{12}(u,v).
 \label{eq:supp-FP-RLGL-transition}
\end{equation}
Here \(R_{12}\) is the Perk--Schultz \(R\) matrix defined in
Eq.~\eqref{eq:HFF-R}.  The face Lax operator is related to an ordinary Lax
operator by
\begin{equation}
 L^{\rm face}_{aj}(u;x)=G_j(x)L_{aj}(u),
 \label{eq:supp-FP-face-L-general}
\end{equation}
where \(L_{aj}(u)\) satisfies the ordinary relation
\begin{equation}
 R_{12}(u,v)L_{1j}(u)L_{2j}(v)
 =L_{2j}(v)L_{1j}(u)R_{12}(u,v).
 \label{eq:supp-FP-RLL-general}
\end{equation}
The operators \(G_j(x)\) act only on the quantum space and generally do not
commute for neighboring sites.  Their composition is constrained by
\begin{equation}
 G_j(x)G_{j+1}(z)G_j(y)
 =
 G_{j+1}(y)G_j(z)G_{j+1}(x),
 \label{eq:supp-FP-G-consistency-transition}
\end{equation}
where \(z\) is the intermediate face parameter.  Together with the local
relation above, this allows the color lines to be propagated through the
whole chain.  Writing
\(\mathbb T(u)\) for the resulting face monodromy matrix, the periodic
transfer matrix and its Hamiltonian are defined by
\begin{equation}
 t^{\rm face}(u)=\operatorname{tr}_{\rm aux}^{\rm face}\mathbb T(u),
 \qquad
 H=-\left.\partial_u t^{\rm face}(u)\right|_{u=0},
 \label{eq:supp-FP-general-face-transfer}
\end{equation}
 up to an additive normalization.  The transfer matrix \(t^{\rm face}(u)\)
 can also serve as a parent transfer family for HFF Hamiltonians.  We illustrate this
 construction with the Fendley--Pozsgay realization.  Its closed two-level
 polynomial hierarchy
is derived below from the equivalent doubled-vertex realization.

For the Fendley--Pozsgay model~\cite{FendleyPozsgay2024}, define
\begin{equation}
 A_j=\sigma^x_{j-1}\sigma^x_j\sigma^z_{j+1},\qquad
 B_j=b^2\sigma^z_{j-2}\sigma^y_{j-1}\sigma^y_j,\qquad
 C_j=b\sigma^z_{j-1}\sigma^z_{j+1},
 \qquad b\neq0 .
 \label{eq:supp-FP-tiles}
\end{equation}
With all site indices understood periodically, the Hamiltonian is
\begin{equation}
 H_{\rm FP}
 =
 \sum_j
 \left(
 \sigma^x_j\sigma^x_{j+1}\sigma^z_{j+2}
 +b\sigma^z_{j-1}\sigma^z_{j+1}
 +b^2\sigma^z_{j-1}\sigma^y_j\sigma^y_{j+1}
 \right),
\label{eq:supp-FP-H}
\end{equation}
Here we give an RLL-based integrable construction of this model together
with a complete exact solution of its periodic spectrum.
The local identity introduced in Ref.~\cite{FendleyPozsgay2024},
\begin{equation}
 A_jB_j=-C_{j-1}C_j
 \label{eq:supp-FP-local-identity}
\end{equation}
allows the local terms to be factorized.  Introduce operators
\(\widehat p_j,\widehat q_j\) with
\begin{equation}
\begin{gathered}
 \widehat p_j^2=\widehat q_j^2=\1,\qquad
 [\widehat p_j,\widehat q_j]
 =
 [\widehat p_j,\widehat p_{j+1}]
 =
 [\widehat q_j,\widehat q_{j+1}]=0,\\
 \{\widehat p_j,\widehat q_{j+1}\}
 =
 \{\widehat q_j,\widehat p_{j+1}\}=0 ,
\end{gathered}
\label{eq:supp-FP-half-edge-algebra}
\end{equation}
Choose them such that
\begin{equation}
 A_j=\widehat p_j\widehat p_{j+1},
 \qquad
 b^{-2}B_j=\widehat q_j\widehat q_{j+1},
 \qquad
 b^{-1}C_j=\widehat p_j\widehat q_j.
 \label{eq:supp-FP-half-edge}
\end{equation}
For the physical normalization set
\(p_j=\widehat p_j\), \(q_j=b\widehat q_j\), and \(c_j=p_jq_j\).
The corresponding Lax operator is
\begin{equation}
 L_j^{\rm FP}(u)
 =
 \begin{pmatrix}
 \1&0&-u c_jc_{j+1}\\
 p_j&q_{j+1}&0\\
 0&q_j&p_{j+1}
 \end{pmatrix},
\label{eq:supp-FP-L}
\end{equation}
This \(L_j^{\rm FP}(u)\) satisfies the ordinary RLL relation with the
Perk--Schultz \(R\) matrix in Eq.~\eqref{eq:HFF-R} for \(k=3\).

The two local face configurations associated with
Eq.~\eqref{eq:supp-FP-local-identity} are encoded by
\begin{equation}
 W_j=
 \frac{
 \1+\widehat p_j\widehat p_{j+1}+\widehat q_j\widehat q_{j+1}
 +\widehat p_j\widehat q_j\widehat p_{j+1}\widehat q_{j+1}}{2},
 \qquad
 e_j=\frac{\1-W_j}{2}
 =
 \frac{(\1-\widehat p_j\widehat p_{j+1})
 (\1-\widehat q_j\widehat q_{j+1})}{4}.
 \label{eq:supp-FP-projector}
\end{equation}
Here \(W_j^2=\1\) and \(e_j^2=e_j\).  For \(x\neq-1\), define
\begin{equation}
 G_j(x)=\1+xe_j,\qquad
 G_j(x)^{-1}=\1-\frac{x}{1+x}e_j.
 \label{eq:supp-FP-G}
\end{equation}

The face-type consistency equation itself appears when two neighboring
face configurations overlap.  Put
\begin{equation}
 D_j=\widehat p_j\widehat p_{j+1}\widehat q_j\widehat q_{j+1},
 \qquad D_j^2=\1,
 \label{eq:supp-FP-D}
\end{equation}
and define the operator-valued composition law
\begin{align}
 z(D;x,y)
 &=
 z_-(x,y)\frac{\1-D}{2}
 +z_+(x,y)\frac{\1+D}{2},
 \label{eq:supp-FP-zD}\\
 z_-(x,y)&=x+y+xy,
 \\
 z_+(x,y)&=\frac{2(x+y+xy)}{2-xy}.
 \label{eq:supp-FP-zpm}
\end{align}
The face-type Yang--Baxter equation is
\begin{equation}
{
\begin{aligned}
 &(\1+xe_j)
 [\1+z(D_j;x,y)e_{j+1}]
 (\1+ye_j)
 =
 (\1+ye_{j+1})
 [\1+z(D_{j+1};x,y)e_j]
 (\1+xe_{j+1}).
\end{aligned}}
\label{eq:supp-FP-face-YBE}
\end{equation}
Thus the intermediate parameter is \(z_-(x,y)\) in the \(D=-1\)
sector and \(z_+(x,y)\) in the \(D=+1\) sector.  Equation
\eqref{eq:supp-FP-face-YBE} states that the two orders of resolving three
overlapping faces give the same operator.

Let \(d=\pm1\) label the two face sectors.  Denote by
\(L_j^{d'd}(u)\) the component of \(L_j^{\rm face}(u)\) that maps an
incoming label \(d\) to an outgoing label \(d'\).  The corresponding operator
on the color and face-label spaces is
\begin{equation}
 \mathbb L_j(u)
 =
 \sum_{d,d'=\pm1}
 E_{d'd}^{(h)}\otimes L_j^{d'd}(u).
 \label{eq:supp-FP-height-L}
\end{equation}
Ordinary matrix multiplication now sums over the intermediate face labels.
Equations~\eqref{eq:supp-FP-RLGL-transition} and
\eqref{eq:supp-FP-face-YBE} allow the two color lines to pass through every
site.  Closing both the color line and the face-label path gives the face
monodromy matrix
\begin{equation}
\mathbb T(u)=\mathbb L_N(u)\cdots\mathbb L_1(u),
 \label{eq:supp-FP-face-monodromy}
\end{equation}
and the corresponding transfer matrix
\begin{equation}
 t_{\rm FP}^{[1],\rm per}(u)
 =
 \operatorname{tr}_{V\otimes h}^{\rm face}
 \mathbb T(u)
 \label{eq:supp-FP-periodic-transfer}
\end{equation}
with
\begin{equation}
 [t_{\rm FP}^{[1],\rm per}(u),
 t_{\rm FP}^{[1],\rm per}(v)]=0,
 \qquad
 t_{\rm FP}^{[1],\rm per}(0)=\1,
 \qquad
 -\left.\partial_u t_{\rm FP}^{[1],\rm per}(u)\right|_{u=0}
 =H_{\rm FP}.
 \label{eq:supp-FP-periodic-integrability}
\end{equation}

The face-type \(RLL\) relation admits a vertex realization on an enlarged
auxiliary space~\cite{Felder1994,EtingofVarchenko1998}.  In this realization
the six-state Lax operator \(L_{\rm FP}^{[1]}(u)\) and the corresponding
\(36\times36\) intertwiner \(R_b(u,v)\) obey an ordinary RLL relation, written
explicitly below.
Introduce the six-dimensional auxiliary space
\begin{equation}
 V_6=\mathbf3\oplus\bar{\mathbf3}
 =
 \operatorname{span}
 \{a^+,b^+,c^+,b^-,c^-,a^-\}.
 \label{eq:supp-FP-V6}
\end{equation}
In this ordered basis, the FP Lax operator is
\begin{equation}
 L_{\rm FP}^{[1]}(u)=
 \begin{pmatrix}
 \1&-u\sigma^z&0&0&-ub^2\sigma^y&0\\
 0&0&\sigma^x&b\1&0&0\\
 \sigma^x&0&0&0&0&0\\
 \sigma^z&-u\1&0&0&0&0\\
 0&0&0&\sigma^y&0&u\sigma^x\\
 0&0&\sigma^y&0&0&0
 \end{pmatrix} .
\label{eq:supp-FP-L6}
\end{equation}
All omitted entries vanish.  Closing the auxiliary line in this six-state
representation gives
\begin{equation}
 t_{\rm FP}^{[1],\rm per}(u)
 =
 \operatorname{tr}_{V_6}
 [L_{{\rm FP},N}^{[1]}(u)\cdots
  L_{{\rm FP},1}^{[1]}(u)] .
 \label{eq:supp-FP-L6-transfer}
\end{equation}
The loop \(a^+\to c^+\to b^+\to a^+\) generates the
\(\sigma^x\sigma^x\sigma^z\) term, the coin-flip loop generates
\(b\sigma^z\sigma^z\), and the loop crossing the two three-state sectors
generates \(b^2\sigma^z\sigma^y\sigma^y\).  Thus
Eq.~\eqref{eq:supp-FP-L6-transfer} reproduces
Eq.~\eqref{eq:supp-FP-H} at first order.

The corresponding \(36\times36\) auxiliary intertwiner is also explicit.
It is convenient to display it first in the color--coin order
\begin{equation}
 \widetilde V_6=(a^+,c^+,b^+,a^-,c^-,b^-).
 \label{eq:supp-FP-R-color-basis}
\end{equation}
In the induced coin-pair decomposition
\((++),(+-),(-+),(--)\), one has at \(b=1\)
\begin{equation}
 \widetilde R_1(u,v)=
 \begin{pmatrix}
 R^{\rm PS}(u,v)&0&0&0\\
 0&M_{00}(u,v)&M_{01}(u,v)&0\\
 0&M_{10}(u,v)&M_{11}(u,v)&0\\
 0&0&0&R^{\rm PS}(u,v)^{\mathsf T}
 \end{pmatrix},
 \label{eq:supp-FP-R36-block}
\end{equation}
where \(R^{\rm PS}\) is the three-color matrix in
Eq.~\eqref{eq:HFF-R}.  To specify the four crossed blocks, order the color
pairs as
\begin{equation}
 \mathcal C=(aa,ac,ab,ca,cc,cb,ba,bc,bb),
 \qquad
 s=u+v,
 \qquad
 d_-=v-u,
 \qquad
 d_+=u-v.
 \label{eq:supp-FP-R36-color-pairs}
\end{equation}
On the off-diagonal subspace
\(\mathcal O=(ac,ab,ca,cb,ba,bc)\),
\begin{equation}
 M_{00}|_{\mathcal O}=M_{11}|_{\mathcal O}
 =\operatorname{diag}(d_-,d_-,d_+,d_-,d_+,d_+).
 \label{eq:supp-FP-R36-crossed-diagonal}
\end{equation}
On the diagonal-color subspace
\(\mathcal D=(aa,cc,bb)\), the four blocks are
\begin{equation}
\begin{aligned}
 M_{00}|_{\mathcal D}
 &=
 \begin{pmatrix}
 s&-2u&2\ii u\\
 -2v&s&-2\ii u\\
 -2\ii v&2\ii v&s
 \end{pmatrix},
 &
 M_{11}|_{\mathcal D}
 &=
 \begin{pmatrix}
 s&-2v&-2\ii v\\
 -2u&s&2\ii v\\
 2\ii u&-2\ii u&s
 \end{pmatrix},\\[3pt]
 M_{01}|_{\mathcal D}
 &=
 \begin{pmatrix}
 0&2u&2\ii u\\
 2u&0&-2\ii v\\
 2\ii u&-2\ii v&4v
 \end{pmatrix},
 &
 M_{10}|_{\mathcal D}
 &=
 \begin{pmatrix}
 0&2v&-2\ii v\\
 2v&0&2\ii u\\
 -2\ii v&2\ii u&4u
 \end{pmatrix}.
\end{aligned}
\label{eq:supp-FP-R36-diagonal-core}
\end{equation}
The remaining off-diagonal entries of the coin-flip blocks are
\begin{align}
 (M_{01})_{ac,ca}=(M_{01})_{ab,ba}
 &=(M_{01})_{ca,ac}=(M_{01})_{ba,ab}=2u, 
 (M_{01})_{cb,bc}=(M_{01})_{bc,cb}=2v,
 \label{eq:supp-FP-R36-M01}\\
 (M_{10})_{ac,ca}=(M_{10})_{ab,ba}
 &=(M_{10})_{ca,ac}=(M_{10})_{ba,ab}=2v, 
 (M_{10})_{cb,bc}=(M_{10})_{bc,cb}=2u.
 \label{eq:supp-FP-R36-M10}
\end{align}
All entries not fixed by
Eqs.~\eqref{eq:supp-FP-R36-crossed-diagonal}--
\eqref{eq:supp-FP-R36-M10} vanish.  For nonzero \(b\), the corresponding
matrix is obtained by the diagonal color gauge
\begin{equation}
 G_b=\operatorname{diag}(1,b^{-1/2},1,1,b^{-1/2},1),
 \qquad
 \widetilde R_b
 =(G_b\otimes G_b)\widetilde R_1(G_b\otimes G_b)^{-1}.
 \label{eq:supp-FP-R36-b-gauge}
\end{equation}
The permutation \(\pi=(1,3,2,6,5,4)\) converts
Eq.~\eqref{eq:supp-FP-R-color-basis} to the basis in
Eq.~\eqref{eq:supp-FP-V6}.  Applying the same permutation to both auxiliary
factors gives the matrix \(R_b(u,v)\) used with
Eq.~\eqref{eq:supp-FP-L6}.

Direct multiplication gives the ordinary relations
\begin{align}
 R_{b,12}(u,v)L_1^{[1]}(u)L_2^{[1]}(v)
 &=L_2^{[1]}(v)L_1^{[1]}(u)R_{b,12}(u,v),
 \label{eq:supp-FP-R36-RLL}\\
 R_{b,12}(u,v)R_{b,13}(u,w)R_{b,23}(v,w)
 &=R_{b,23}(v,w)R_{b,13}(u,w)R_{b,12}(u,v).
 \label{eq:supp-FP-R36-YBE}
\end{align}
The same structure becomes especially transparent in the checked
convention \(\check R_1=P_6\widetilde R_1\):
\begin{equation}
 \check R_1(u,v)
 =(u+v)\mathcal K_{\rm H}-\ii(u-v)\mathcal E_{\rm dbl}
 =u\mathcal B+v\mathcal B^{-1},
 \qquad
 \mathcal B=\mathcal K_{\rm H}-\ii\mathcal E_{\rm dbl},
 \label{eq:supp-FP-R36-Baxterization}
\end{equation}
with
\begin{equation}
 \mathcal K_{\rm H}^2=\1,
 \qquad
 \mathcal E_{\rm dbl}^2=0,
 \qquad
 [\mathcal K_{\rm H},\mathcal E_{\rm dbl}]=0,
 \qquad
 \operatorname{rank}\mathcal E_{\rm dbl}=14.
 \label{eq:supp-FP-R36-algebra}
\end{equation}
Moreover, \(\mathcal B\) satisfies the constant braid relation.  At the
reflected point,
\begin{equation}
 \check R_1(u,-u)=-2\ii u\mathcal E_{\rm dbl},
 \label{eq:supp-FP-R36-reflected}
\end{equation}
so it has rank fourteen.  If \(F_{14}\) denotes the normalized embedding
of the corresponding fourteen-dimensional subspace, define
\begin{equation}
 L_{\rm FP}^{[2]}(u)
 =
 (F_{14}^{\dagger}\otimes\1)
 L_{\rm FP,1}^{[1]}(u)L_{\rm FP,2}^{[1]}(-u)
 (F_{14}\otimes\1),
 \qquad
 t_{\rm FP}^{[2]}(u)
 =
 \operatorname{tr}_{14}
 \prod_{j=N}^{1}L_{{\rm FP},j}^{[2]}(u).
 \label{eq:supp-FP-L14}
\end{equation}
This is the non-scalar transfer operator retained by the periodic closure.

For a fixed boundary face label, the corresponding open Hamiltonian is
\begin{equation}
 H_{\rm FP}^{\rm open}
 =
 C_1+\sum_{j=2}^{L}(A_j+B_j+C_j),
 \label{eq:supp-FP-open-H}
\end{equation}
where the exterior Pauli variables are fixed by the boundary cap.  Let
\(\tau_{\rm FP}^{[1]}(u)\) denote its projected transfer operator.  It obeys
\begin{equation}
 \tau_{\rm FP}^{[1]}(0)=\1,\qquad
 -\left.\partial_u\tau_{\rm FP}^{[1]}(u)\right|_{u=0}
 =H_{\rm FP}^{\rm open},\qquad
 \tau_{\rm FP}^{[1]}(u)\tau_{\rm FP}^{[1]}(-u)
 =P_L^{\rm FP}(u^2)\1 .
 \label{eq:supp-FP-open-inversion}
\end{equation}
If
\begin{equation}
 P_L^{\rm FP}(u^2)
 =\prod_{a=1}^{d_{\rm FP}}
 (1-u^2\epsilon_{a,L}^2),
 \label{eq:supp-FP-open-P}
\end{equation}
the joint eigenvalues are
\begin{equation}
 \Lambda_{\boldsymbol\eta}^{\rm open}(u)
 =\prod_{a=1}^{d_{\rm FP}}
 (1-u\eta_a\epsilon_{a,L}),
 \qquad
 E_{\boldsymbol\eta}=\sum_{a=1}^{d_{\rm FP}}\eta_a\epsilon_{a,L},
 \qquad \eta_a=\pm1 .
\label{eq:supp-FP-open-solution}
\end{equation}
For even \(L\), all sign choices occur.  For odd \(L\), the fixed face label
imposes one central parity constraint on the allowed signs.

We now give the closed periodic spectral equations.  The scalar source of
the first fusion relation is generated by the six-dimensional matrix
\begin{equation}
 A_6(x)=
 \begin{pmatrix}
 1&0&0&-b^2x&-b^4x&-x\\
 1&0&0&0&0&-x\\
 1&0&0&0&-b^4x&0\\
 1&0&0&0&0&0\\
 0&1&0&0&0&0\\
 0&0&1&0&0&0
 \end{pmatrix}.
 \label{eq:supp-FP-A6}
\end{equation}
Its characteristic polynomial is
\begin{equation}
\begin{aligned}
 \chi_{\rm FP}(\lambda,x)
 ={}&\lambda^6-\lambda^5+b^2x\lambda^4
 +(1+b^4)x\lambda^3 -b^4x^2\lambda^2-b^4x^2\lambda-b^6x^3 .
\end{aligned}
\label{eq:supp-FP-spectral-curve}
\end{equation}
Introduce the eight-dimensional scalar matrix
\begin{equation}
 A_8(u)=A_6(u^2)\oplus(bu)\oplus(-bu)
 \label{eq:supp-FP-A8}
\end{equation}
and define
\begin{equation}
 P_N^{\rm FP}(u^2)=\operatorname{tr}A_8(u)^N,\qquad
 \Xi_N^{\rm FP}(u)
 =
 \frac{P_N^{\rm FP}(u^2)^2-P_{2N}^{\rm FP}(u^2)}{2}.
 \label{eq:supp-FP-scalar-sources}
\end{equation}
The first identity is
\begin{equation}
{
 t_{\rm FP}^{[1]}(u)t_{\rm FP}^{[1]}(-u)
 =
 P_N^{\rm FP}(u^2)\1
 +t_{\rm FP}^{[2]}(u)
 +t_{\rm FP}^{[2]}(-u).}
\label{eq:supp-FP-first-fusion}
\end{equation}
Here and below the superscript \(\mathrm{per}\) is suppressed.

Two four-dimensional companion matrices determine the remaining scalar
returns.  Define \(M_a(u)\) and \(M_b(u)\) by
\begin{align}
 \det[\mu\1-M_a(u)]
 &=
 \mu^4-b^2\mu^2-b^3(1+b^2)u\mu-b^6u^2,
 \label{eq:supp-FP-Ma}\\
 \det[\mu\1-M_b(u)]
 &=
 \mu^4+(1+b^2)\mu^3+b^2\mu^2-b^6u^2,
 \label{eq:supp-FP-Mb}
\end{align}
and set
\begin{align}
 A_{+,N}^{\rm FP}(u)
 &=
 u^N\left\{
 \operatorname{tr}[M_a(u)^N]
 +(-1)^N\operatorname{tr}[M_b(u)^N]\right\},
 \label{eq:supp-FP-Aplus}\\
 A_{-,N}^{\rm FP}(u)
 &=
 u^N\left\{
 \operatorname{tr}[M_a(u)^N]
 +\operatorname{tr}[M_b(u)^N]\right\}.
 \label{eq:supp-FP-Aminus}
\end{align}
With
\begin{equation}
 \mathcal P_z=\prod_{j=1}^{N}\sigma_j^z,\qquad
 \mathcal P_z^2=\1,
 \label{eq:supp-FP-parity}
\end{equation}
the second reflected relation closes entirely on the fundamental and first
fused transfers:
\begin{equation}
{
\begin{aligned}
 t_{\rm FP}^{[2]}(u)t_{\rm FP}^{[2]}(-u)
 ={}&
 \Xi_N^{\rm FP}(u)\1
 +A_{+,N}^{\rm FP}(u)\mathcal P_zt_{\rm FP}^{[1]}(u) +
 A_{-,N}^{\rm FP}(u)\mathcal P_zt_{\rm FP}^{[1]}(-u) +
 (bu)^N[t_{\rm FP}^{[1]}(u)]^2
 +(-bu)^N[t_{\rm FP}^{[1]}(-u)]^2 .
\end{aligned}}
\label{eq:supp-FP-second-fusion}
\end{equation}
For even \(N\), \(A_{+,N}^{\rm FP}=A_{-,N}^{\rm FP}\).  Their distinction is
essential for odd \(N\).  Equations~\eqref{eq:supp-FP-first-fusion} and
\eqref{eq:supp-FP-second-fusion} are the finite FP hierarchy.

To state its exact spectral form, take a simultaneous eigenstate with
\(\mathcal P_z=\eta=\pm1\), and denote the two eigenvalue polynomials by
\(\Lambda_\eta^{[1]}(u)\) and \(\Lambda_\eta^{[2]}(u)\).  Their degree
structure follows from the loop expansion of the two Lax operators.  In the
same parametrization as the FFD chain, write:
\begin{equation}
 \Lambda_\eta^{[1]}(u)
 =\prod_{j=1}^{d_N}\left(1-\frac{u}{\lambda_j}\right),
 \qquad
 \Lambda_\eta^{[2]}(u)
 =\kappa_\eta u^{s_N}
 \prod_{\alpha=1}^{n_N}\left(1-\frac{u}{\gamma_\alpha}\right),
 \qquad
 d_N=\left\lfloor\frac N2\right\rfloor,
 \quad s_N=\left\lceil\frac N2\right\rceil,
 \quad n_N=2d_N-s_N,
 \label{eq:supp-FP-fundamental-roots}
\end{equation}
with \(E=\sum_j\lambda_j^{-1}\).  The lowest power of
\(\Lambda_\eta^{[2]}\) is set by the two-site loops of \(L^{[2]}\), one
factor of \(u\) per two sites.  The highest power is \(2d_N\) by
Eq.~\eqref{eq:supp-FP-first-fusion}.

Matching all powers of \(u\) in Eqs.~\eqref{eq:supp-FP-first-fusion} and
\eqref{eq:supp-FP-second-fusion} gives a polynomial system for the
\(d_N\) roots \(\lambda_j\), the \(n_N\) roots \(\gamma_\alpha\), and
the coefficient \(\kappa_\eta\).  The first relation determines the even
part of \(\Lambda_\eta^{[2]}\), while the second determines its odd part.
Counting shows that this system is square: the number of independent
coefficient equations equals the number of unknowns
(\(5\) for \(N=4\), \(7\) for \(N=6\)).  The physical joint transfer
characters are selected by the extremal coefficients, which are explicit
local operators.  The highest coefficient of
\(t^{[1]}_{\rm FP}\) is generated by the two-site loops
\(b^+\!\to b^-\!\to b^+\) and the four-site loops
\(c^-\!\to a^+\!\to c^+\!\to a^-\!\to c^-\) tiling the ring, and the lowest
coefficient of \(t^{[2]}_{\rm FP}\) by the two two-site loops of
\(L^{[2]}\).  For even \(N\), with the alternating string
\begin{equation}
 \mathcal A_N
 =\sigma_1^x\sigma_2^y\sigma_3^x\sigma_4^y\cdots
 \sigma_{N-1}^x\sigma_N^y
 +\sigma_1^y\sigma_2^x\sigma_3^y\sigma_4^x\cdots
 \sigma_{N-1}^y\sigma_N^x,
 \qquad
 \mathcal A_N^2=2+2\mathcal P_z,
 \label{eq:supp-FP-alternating}
\end{equation}
one finds
\begin{equation}
{
\begin{aligned}
 [u^{N/2}]\,t^{[1]}_{\rm FP}
 &=2(-b)^{N/2}\1-2b^{N/2}\,\mathbf 1_{\{N\,{\rm divisible\ by}\,4\}}\,\mathcal A_N,
 \\
 [u^{N/2}]\,t^{[2]}_{\rm FP}
 &=\ii^{N/2}(1+b^N)\,\mathcal A_N,
 \\
 [u^{N}]\,t^{[2]}_{\rm FP}
 &=\tfrac12\big([u^{N/2}]t^{[1]}_{\rm FP}\big)^2
 -\tfrac12[P_N^{\rm FP}]_{N/2}\1 ,
\end{aligned}}
\label{eq:supp-FP-extremal}
\end{equation}
where the mixed tilings are excluded by color continuity, so that the
\(\mathcal A_N\) term of \(t^{[1]}\) appears only when \(N\) is divisible by four.  Since
\(\mathcal A_N\) takes the values \(\pm2\) on \(\mathcal P_z=+1\) and
\(0\) on \(\mathcal P_z=-1\), the extremal coefficients of every joint
character are quantized:
\begin{equation}
{
 \begin{aligned}
 \eta=-1:\quad
 &\left([u^{N/2}]\Lambda_{\eta}^{[1]},[u^{N/2}]\Lambda_{\eta}^{[2]}\right)
 =\big(2(-b)^{N/2},0\big),\\
 \eta=+1:\quad
 &\left([u^{N/2}]\Lambda_{\eta}^{[1]},[u^{N/2}]\Lambda_{\eta}^{[2]}\right)
 =\big(2(-b)^{N/2}\mp4b^{N/2}\mathbf 1_{\{N\,{\rm divisible\ by}\,4\}},\
 \pm2\,\ii^{N/2}(1+b^N)\big).
 \end{aligned}}
\label{eq:supp-FP-selection}
\end{equation}
For odd \(N\) the extremal coefficients are again finite sums of Pauli
strings, and their joint spectrum
plays the same role. 

At \(b=1\), the fermionic spin structure is obtained from the
parity-graded closure
\begin{equation}
 \mathcal P_\pm=\frac{1\pm\mathcal P_z}{2},
 \qquad
 t_D(u)=\mathcal P_-t^{[1]}_{\rm FP}(u)
 +\mathcal P_+t^{[1]}_{{\rm FP},g_{\rm s}}(u),
 \qquad
 g_{\rm s}=\operatorname{diag}(1,1,1,1,-1,-1).
 \label{eq:supp-FP-DFNR-closure}
\end{equation}
This closure combines the two fermionic spin structures rather than using
an ordinary RTT twist.  Its first derivative gives the supersymmetric chain
of Ref.~\cite{DFNR2016}.  The parity factor in
Eq.~\eqref{eq:supp-FP-second-fusion} is absorbed into this closure, while
the two terms proportional to \(A_{\pm,N}^{\rm FP}\) acquire a minus sign.  The
three extremal coefficients reduce to the scalars
\(2(-1)^{N/2}\), \(0\) and \(-2\).

\end{document}